\documentclass[12pt]{article}

\usepackage{xr-hyper}
\usepackage{amsmath, amsfonts, amssymb, amsthm, fullpage, color, bbm, enumerate, titlesec, graphicx, epstopdf, multirow, array, mathtools, floatpag}
\usepackage[T1]{fontenc}
\usepackage{lmodern}
\usepackage[title]{appendix}

\usepackage[onehalfspacing]{setspace}

\definecolor{darkred}{rgb}{0.5,0,0}

\titleformat*{\paragraph}{\sc}

\usepackage{hyperref}
\hypersetup{allbordercolors={1 1 1},allcolors=darkred,colorlinks=true}
\usepackage[capitalize,nameinlink,noabbrev]{cleveref}

\usepackage[small,bf]{caption}

\usepackage[round]{natbib}
\usepackage[runin]{abstract}

\newcolumntype{C}[1]{>{\centering\arraybackslash}p{#1}}

\theoremstyle{plain}

\crefname{asn}{Assumption}{Assumptions}

\crefname{lem}{Lemma}{Lemmas}

\newtheorem{prop}{Proposition}

\newtheorem*{claim*}{Claim}
\newtheorem*{cor*}{Corollary}

\theoremstyle{definition}

\DeclareMathOperator*{\tr}{trace}

\DeclareMathOperator*{\diag}{diag}

\DeclareMathOperator*{\var}{Var}

\DeclareMathOperator{\avar}{aVar}

\DeclareMathOperator{\abias}{aBias}

\allowdisplaybreaks

\newcommand{\alert}[1]{#1}%\textcolor{blue}{#1}}

\crefname{sappsec}{Supplemental Appendix}{Supplemental Appendices}
\crefname{sappsubsec}{Supplemental Appendix}{Supplemental Appendices}
\crefname{sappsubsubsec}{Supplemental Appendix}{Supplemental Appendices}
\crefname{appsec}{Appendix}{Appendices}
\begin{document}

\title{\texorpdfstring{\vspace{-1.5\baselineskip}}{} Local Projections vs. VARs: \\ Lessons From Thousands of DGPs\thanks{Email: {\tt dakel@twosigma.com}, {\tt mikkelpm@princeton.edu}, and {\tt ckwolf@mit.edu}. We received helpful comments from Isaiah Andrews, R{\'e}gis Barnichon, Gabe Chodorow-Reich, Viet Hoang Dinh, \`{O}scar Jord\`{a}, Helmut L\"{u}tkepohl, Massimiliano Marcellino, Pepe Montiel Olea, Ulrich M\"{u}ller, Emi Nakamura, Frank Schorfheide, Chris Sims, Lumi Stevens, Jim Stock, Mark Watson, Andrei Zeleneev, several anonymous referees, and numerous seminar and conference participants. Samya Aboutajdine, Tom{\'a}s Caravello, Chun-Beng Leow, and Eric Qian provided excellent research assistance. Plagborg-M{\o}ller acknowledges that this material is based upon work supported by the NSF under Grant {\#}2238049, and Wolf does the same for Grant {\#}2314736. Any opinions, findings, and conclusions or recommendations expressed in this material are those of the authors and do not necessarily reflect the views of the NSF. In addition, the views expressed herein are solely the views of the authors and are not necessarily the views of Two Sigma Investments, LP or any of its affiliates. They are not intended to provide, and should not be relied upon for, investment advice. \alert{This research was mostly conducted while Dake Li was affiliated with Princeton University.}}}
\author{\begin{tabular}{ccc}
Dake Li & Mikkel Plagborg-M{\o}ller & Christian K. Wolf \\
{\small Two Sigma Investments, LP} & {\small Princeton University} & {\small MIT \& NBER}
\end{tabular}}
\date{\texorpdfstring{\bigskip}{ }\today}
\maketitle

\vspace{-1em}
\begin{abstract}
We conduct a simulation study of Local Projection (LP) and Vector Autoregression (VAR) estimators of structural impulse responses across thousands of data generating processes, designed to mimic the properties of the universe of U.S. macroeconomic data. Our analysis considers various identification schemes and several variants of LP and VAR estimators, employing bias correction, shrinkage, or model averaging. A clear bias-variance trade-off emerges: LP estimators have lower bias than VAR estimators, but they also have substantially higher variance at intermediate and long horizons. Bias-corrected LP is the preferred method if and only if the researcher overwhelmingly prioritizes bias. For researchers who also care about precision, VAR methods are the most attractive---Bayesian VARs at short and long horizons, and least-squares VARs at intermediate and long horizons.
\end{abstract}

\clearpage

\section{Introduction}
\label{sec:intro}

Since \citet{Jorda2005} introduced the popular local projection (LP) impulse response estimator, there has been a debate about its benefits and drawbacks relative to Vector Autoregression (VAR) estimation \citep{Sims1980}. Recently, \citet{Plagborg2020} proved that these two methods in fact estimate precisely the same impulse responses asymptotically, provided that the lag length used for estimation tends to infinity. This result holds regardless of identification scheme and regardless of the underlying data generating process (DGP). Nevertheless, the question of which estimator to choose in finite samples remains open. It is also an urgent question, since researchers have remarked that LPs and VARs can give conflicting results when applied to central economic questions such as the effects of monetary or fiscal stimulus \citep[e.g.,][]{Ramey2016,Nakamura2018}.

Whereas the LP estimator utilizes the sample autocovariances flexibly by directly projecting an outcome at the future horizon $h$ on current covariates, a VAR($p$) estimator instead extrapolates longer-run impulse responses from the first $p$ sample autocovariances. Hence, though the estimates from the two methods agree approximately at horizons $h \leq p$, they can disagree substantially at intermediate and long horizons.\footnote{See \citet[Proposition 2]{Plagborg2020} for a formal result.} Intuitively, the extrapolation employed by VARs should yield a lower variance but potentially a higher bias than for LPs, perfectly analogous to the trade-off between direct and iterated reduced-form forecasts \citep{Schorfheide2005,Kilian2017}.\footnote{The trade-off is also conceptually similar to the relationship between polynomial series estimators and kernel estimators in cross-sectional nonparametric regression.} How much more should one care about bias than variance to optimally choose the LP estimator over the VAR estimator in realistic sample sizes? And how does the trade-off depend on the DGP? Unfortunately, these questions are challenging to answer analytically, due to the dynamic and nonlinear nature of the time series estimators, as well as the breadth of DGPs encountered in applied practice.

In this paper we illuminate the bias-variance trade-off in impulse response estimation through a comprehensive simulation study, applying LP and VAR methods to thousands of empirically relevant DGPs. Our goal is to identify which estimators perform well \emph{on average} across many DGPs and thus may serve as practical default procedures. Rather than insisting on the usual binary distinction between ``local projections'' and ``VARs'', we furthermore consider an entire menu of related estimation approaches that employ bias correction, shrinkage, or model-averaging. We find that the usual least-squares LP estimator tends to have lower bias than the least-squares VAR estimator, as expected, but also that this bias reduction comes at the cost of substantially higher variance. Out of all the procedures we analyze, bias-corrected LP is the most attractive estimator if \emph{and only if} the researcher overwhelmingly prioritizes bias. If, however, the researcher also cares about precision (as in the conventional mean squared error criterion), then VAR methods are the most attractive; in particular, Bayesian VARs perform well at short and long horizons, while it is difficult to beat the least-squares VAR estimator at intermediate and long horizons.

Our simulation study considers an extensive array of DGPs, obtained by drawing specifications at random from a large-scale, empirically calibrated dynamic factor model (DFM).\footnote{Our overall approach is inspired by \citet{Lazarus2018}, who are instead interested in the question of how to select among different long-run variance estimators.} We fit the DFM to the data set of \citet{Stock2016}, which contains a large number of quarterly U.S. macroeconomic time series spanning a wide variety of variable categories. As emphasized by \citeauthor{Stock2016}, such DFMs can accurately capture the joint co-movements of conventional macroeconomic data, and so our simulation results will be informative about the universe of standard U.S. time series. This estimated DFM exhibits realistic and complex dynamics in the short and long run, including cointegrating relationships among the latent factors. From the encompassing 207-variable DFM we then draw 6,000 random subsets of five variables (subject to constraints that emulate applied practice); all results reported below are essentially unchanged if instead we limit attention to 17 of the most commonly used macro series out of the 207. The randomly drawn subsets of time series constitute the set of DGPs that we consider for our simulation study. As the calibrated DFM is known to us, we can compute the true impulse responses, and therefore also estimator biases and mean squared errors. Importantly, none of these many DGPs can be exactly represented as a finite-order VAR model, yielding a non-trivial bias-variance trade-off between the LP and VAR estimators. Moreover, our DGPs exhibit substantial heterogeneity in how well they can be approximated by VAR models, in persistence and shape of impulse response functions, and in the invertibility of the structural shocks, consistent with the heterogeneity faced by applied researchers. While our results inevitably depend on the specification of the encompassing model, we believe that an estimation method that works well across our multitude of empirically calibrated DGPs has substantial promise as a default procedure.

We study the ability of several variants of LP and VAR methods to accurately estimate impulse response functions. Consistent with the majority of applied work, the estimators are applied to data in levels, rather than transforming to stationarity prior to estimation. Since VARs with very large lag lengths are asymptotically equivalent to LPs \citep{Plagborg2020,Xu2023}, we focus on VAR estimators with moderate lag length choices, as conventionally found in the literature. In addition to the popular least-squares LP and VAR estimators, we further enrich the bias-variance possibility frontier by considering: (i) small-sample bias correction of the VAR coefficients \citep{Pope1990,Kilian1998} and LP impulse response estimates \citep{Herbst2021}; (ii) penalized LP \citep{Barnichon2019}, which smooths out impulse response functions; (iii) Bayesian VAR estimation, with priors selected as in \citet{Giannone2015}; and (iv) model averaging of univariate and multivariate VAR models of various lag lengths \citep{Hansen2016}. For each estimation method, we consider three oft-used structural identification schemes: observed shocks, instrumental variables (IVs)/proxies, and recursive identification. For IV identification, we further distinguish between internal IV methods \citep{Ramey2011,Plagborg2020} and external IV methods \citep{Stock2008,Stock2012bpea,Mertens2013}. We then evaluate the performance of these estimators through the lens of loss functions with varying weights on bias and variance.

Applying the estimation methods to simulated data from the thousands of DGPs, a clear and unavoidable bias-variance trade-off emerges. We highlight four main lessons:

\begin{enumerate}[1.]

\item Though they perform similarly at short horizons, least-squares LP and VAR estimators lie on opposite ends of the bias-variance spectrum at intermediate and long horizons: small bias and large variance for LPs, and large bias and small variance for VARs. Strictly speaking, this statement is only true after applying the small-sample bias correction procedures of \citet{Pope1990}, \citet{Kilian1998}, and \citet{Herbst2021}, which partially ameliorate the deleterious effects of the high persistence of our DGPs on the biases of the respective estimators. We find such bias correction to be particularly important for LPs.

\item Out of all the estimators we consider, bias-corrected LP is the preferred option if \emph{and only if} the loss function almost exclusively puts weight on bias (at the expense of variance). This is because the lower bias of LP relative to VAR comes at the cost of substantially higher variance, especially at longer horizons.

\item If the loss function attaches at least moderate weight to variance (in addition to bias), such as in the case of mean squared error loss, VAR methods are attractive. But the optimal VAR method depends on the horizon: Bayesian VARs tend to perform well at short horizons, least-squares VARs at intermediate horizons, and the two methods are comparable at long horizons.

\item In the case of IV identification, the SVAR-IV estimator is heavily median-biased, but provides substantial reduction in dispersion, measured by the interquartile range. Depending on the weight attached to bias, it may therefore be justifiable to use external IV methods despite their lack of robustness to non-invertibility (unlike internal IV methods).

\end{enumerate}

Our findings provide a novel perspective on recent work emphasizing the potential dangers of VAR model mis-specification \citep{Ramey2016,Nakamura2018}. We consider DGPs that do not admit finite-order VAR representations, so VAR methods indeed suffer from larger bias, as cautioned there. Reducing that bias via direct projection, however, tends to incur a steep cost in terms of increased sampling variance at intermediate and long horizons. Researchers who prefer to employ LP estimators should therefore be prepared to pay that price, and furthermore should apply the \citet{Herbst2021} bias correction procedure when their data is persistent, as is usually the case.

\paragraph{Literature.}
Our simulation study is inspired by the seminal work of \citet{Marcellino2006} on direct and iterated multi-step forecasts, though we focus instead on structural impulse responses. While simulation studies in the forecasting literature often analyze low-dimensional specifications, we consider multi-variable systems, consistent with standard practice in the applied structural macroeconometrics literature. The structural perspective also requires us to contend with issues such as the variety of different popular shock identification schemes, normalization of impulse responses, and the special role of external instrumental variables.

Our large-scale model set-up differs from prior simulation studies of LP and VAR methods, which have considered at most a handful of DGPs. Examples here include \citet{Jorda2005}, \citet{Meier2005}, \citet{Kilian2011}, \citet{Brugnolini2018}, \citet{Choi2019}, \citet{Austin2020}, and \citet{Bruns2021}. These papers either obtain their DGPs from stylized, low-dimensional VARMA models, calibrated DSGE models, and/or a few empirically calibrated VAR models. Our encompassing DGP is instead designed to closely mimic applied practice: we consider a non-stationary DFM with rich common and idiosyncratic dynamics that accurately captures key properties of the kinds of aggregate time series typically used in standard macroeconometric analyses. Our analysis also differs in the following respects: we consider shrinkage estimation procedures as competitors to the least-squares estimators; we study several popular structural identification schemes; and we examine how our conclusions vary with the impulse response horizon and the researcher's loss function. All these features are essential to the above-mentioned main lessons that we draw from our results.

Even though the simulation results are at the heart of our analysis, we start off by illustrating the bias-variance trade-off through an analytical example that builds on \citet{Schorfheide2005}. That paper develops a general theory of the asymptotic bias and variance of direct and iterated (reduced-form) forecasts under local mis-specification. While these theoretical results are valuable for analytically distilling the forces at work, they do not by themselves resolve the bias-variance trade-off faced by practitioners, as this trade-off invariably depends in a complicated fashion on many features of the DGP.

Finally, we stress that our paper focuses solely on point estimation, as opposed to inference or hypothesis testing. See \citet{Inoue2020}, \citet{MontielOlea2020}, and \citet{Xu2023} for theoretical as well as simulation results on VAR and LP confidence interval procedures. Moreover, we focus exclusively on impulse response estimands, rather than variance decompositions or historical decompositions.

\paragraph{Outline.}
\cref{sec:biasvar} illustrates the bias-variance trade-off for LP and VAR estimators using a simple analytical example. \cref{sec:dgp} describes the empirically calibrated dynamic factor model that we use to generate our many DGPs. \cref{sec:estim} defines the menu of LP- and VAR-based estimation procedures. \cref{sec:results} contains our main simulation results and robustness checks. \cref{sec:conclusion} summarizes the lessons for applied researchers and then offers guidance for future research. The appendix contains implementation details. A supplemental appendix with proofs and further simulation results as well as a Matlab code suite are available online.\footnote{\url{https://github.com/dake-li/lp_var_simul} \label{fn:github}}

\section{The bias-variance trade-off}
\label{sec:biasvar}

This section motivates our simulation study with an analytical discussion of the bias-variance trade-off between LP and VAR impulse response estimators. \cref{subsec:biasvar_illustration} analyzes these estimators in the context of a simple toy model that cleanly illustrates the trade-off, and \cref{subsec:biasvar_outlook} connects this analytical discussion to the rest of the paper.

\subsection{Illustrative example}
\label{subsec:biasvar_illustration}

\citet{Plagborg2020} show that the impulse response estimands of VAR and LP estimators with $p$ lags generally differ at horizons $h > p$: the VAR extrapolates from the first $p$ sample autocovariances, while LP exploits all autocovariances out to horizon $h+p$. This observation suggests the presence of a bias-variance trade-off whenever the true DGP is not a finite-order VAR, perfectly analogous to the choice between ``direct'' and ``iterated'' predictions in multi-step forecasting \citep{Marcellino2006}. We here formalize this basic intuition by extending the arguments of \citet{Schorfheide2005} to structural impulse response estimation in a simple, albeit non-stationary DGP.

\paragraph{Model.}
Consider a simple sequence of drifting DGPs for the scalar time series $y_t$:
\begin{equation}
	y_t = y_{t-1} + \varepsilon_{1,t} + \tau\varepsilon_{1,t-1} + \frac{\alpha}{\sqrt{T}}\varepsilon_{1,t-2} + \varepsilon_{2,t}, \label{eq:DGP_simple}
\end{equation}
where $\varepsilon_t \equiv (\varepsilon_{1,t},\varepsilon_{2,t})'$ is an i.i.d. white noise process with $\var(\varepsilon_t) = \diag(1,\sigma_2^2)$, and $y_0=0$. We assume that the researcher observes $w_t \equiv (\varepsilon_{1,t},y_t)'$, i.e., she observes the shock $\varepsilon_{1,t}$ but not $\varepsilon_{2,t}$. The above DGP drifts towards a unit-root VAR(1) process in $w_t$ at rate $T^{-1/2}$, where $T$ is the sample size. We show below that this ensures a non-trivial bias-variance trade-off in the limit $T\to\infty$. The DGP captures the notion that finite-order autoregressive models are often a good---but not exact---approximation to the true underlying DGP. The degree of autoregressive mis-specification is governed by the parameter $\alpha$.\footnote{To interpret its units, consider a distributed lag regression of $\Delta y_t$ on $\varepsilon_{1,t}$, $\varepsilon_{1,t-1}$, and $\varepsilon_{1,t-2}$. Then it is standard to show that the t-statistic for significance of the second lag converges in distribution to $N(\alpha/\sigma_2,1)$.} 

We are interested in the impulse responses of $y_t$ with respect to a unit impulse in $\varepsilon_{1,t}$. The true impulse response function implied by the model \eqref{eq:DGP_simple} equals $\theta_{h,T} \equiv 1+\tau \mathbbm{1}(h \geq 1) + \alpha T^{-1/2}\mathbbm{1}(h \geq 2)$ at horizon $h$. This impulse response function reflects---in stark fashion---the common empirical finding that signal-to-noise ratios are especially low at longer horizons, here $h \geq 2$, in the sense that the increment $\theta_{2,T}-\theta_{1,T}$ is of the same asymptotic order as the standard errors of the LP and VAR estimators, as shown formally below.

\paragraph{Estimators.}
For now, we consider two estimators of $\theta_{h,T}$.

\begin{enumerate}[1.]
	
	\item {\bf LP.} The least-squares local projection estimator $\hat{\beta}_h$ is obtained from the OLS regression
	\begin{equation}
		y_{t+h} = \hat{\beta}_h \varepsilon_{1,t} + \hat{\zeta}_h' w_{t-1} + \text{residual}_{t,h} \label{eq:LP_simple},
	\end{equation}
	at each horizon $h$. Notice that this LP specification controls for one lag of the data.
	
	\item {\bf VAR.} We consider a recursive VAR specification in $w_t=(\varepsilon_{1,t},y_t)'$, again with one lag. Define the usual least-squares coefficient estimator $\hat{A} \equiv (\sum_{t=2}^T w_t w_{t-1}')(\sum_{t=2}^T w_{t-1} w_{t-1}')^{-1}$ and residual covariance matrix $\hat{\Sigma} \equiv T^{-1}\sum_{t=2}^T \hat{u}_t\hat{u}_t'$, where $\hat{u}_t \equiv w_t - \hat{A} w_{t-1}$. Define the lower triangular Cholesky factor $\hat{C}$, where $\hat{C}\hat{C}' = \hat{\Sigma}$. The un-normalized VAR impulse responses with respect to the first orthogonalized shock at horizon $h$ are given by $\hat{A}^h \hat{C}e_1$, where $e_j$ is the $j$-th unit vector of dimension 2, $j=1,2$. To facilitate comparison with LP, we normalize the impact response of the first variable in the VAR (i.e., $\varepsilon_{1,t}$) with respect to the first shock to be 1. This yields the estimator $\hat{\delta}_h \equiv e_2'\hat{A}^h \hat{\gamma}$, where  $\hat{\gamma} \equiv (1, \hat{\kappa})'$ and $\hat{\kappa} \equiv \hat{\Sigma}_{21}/\hat{\Sigma}_{11}$.\footnote{We have $\hat{C} = \begin{psmallmatrix}
			\sqrt{\hat{\Sigma}_{11}} & 0 \\
			\hat{\Sigma}_{21}/\sqrt{\hat{\Sigma}_{11}} & \hat{\Sigma}_{22} - \hat{\Sigma}_{21}^2/\hat{\Sigma}_{11}
		\end{psmallmatrix}$. We therefore achieve the desired normalization of the impact effect of the shock by dividing $\hat{C}e_1$ by $\sqrt{\hat{\Sigma}_{11}}$. This gives the normalized impulse responses $\hat{A}^h \hat{\gamma}$.}
	
\end{enumerate}

\paragraph{Trade-off.}
Along the stated asymptote, the researcher faces a clear bias-variance trade-off between the LP and VAR impulse response estimators:

\begin{prop}
	\label{prop:biasvar_simple}
	
	Consider the model \eqref{eq:DGP_simple}, and fix $h \geq 0$, $\tau \in \mathbb{R}$, $\sigma_2 > 0$, and $\alpha \in \mathbb{R}$. Assume $E(\varepsilon_{j,t}^4)<\infty$ for $j=1,2$. Then, as $T \rightarrow \infty$,
	\begin{equation}
		\sqrt{T}(\hat{\beta}_h-\theta_{h,T}) \stackrel{d}{\to} N(\abias_{\text{LP},h}, \avar_{\text{LP},h}),\quad \sqrt{T}(\hat{\delta}_h-\theta_{h,T}) \stackrel{d}{\to} N(\abias_{\text{VAR},h}, \avar_{\text{VAR},h}),
	\end{equation}
	where for all $h \geq 0$,
	\[\abias_{\text{LP},h} \equiv 0, \quad \avar_{\text{LP},h} \equiv \lbrace 1+(h-1)(1+\tau)^2\rbrace \mathbbm{1}(h \geq 1) + (h+1)\sigma_2^2.\]
	For $h \in \lbrace 0,1\rbrace$, we have $\abias_{\text{VAR},h}=\abias_{\text{LP},h}=0$ and $\avar_{\text{VAR},h}=\avar_{\text{LP},h}$. For $h \geq 2$,
	\[\abias_{\text{VAR},h} \equiv -\alpha, \quad \avar_{\text{VAR},h} \equiv (1+\tau)^2+2\sigma_2^2.\]
\end{prop}

\begin{proof}
	Please see \cref{app:proofs}.
\end{proof}

At horizons $h \in \lbrace 0,1\rbrace$, there is no bias-variance trade-off: on impact, the two estimators are numerically equivalent; at $h=1$, both are asymptotically unbiased with identical asymptotic variance, consistent with \citet{Plagborg2020}. Intuitively, the equivalence at $h=1$ reflects the fact that the VAR(1) estimator does not extrapolate, instead reporting the direct projection of $y_{t+1}$ on $w_t$, exactly as LP does \citep{Plagborg2020}.

At horizons $h \geq 2$ (i.e., exceeding the lag length used for estimation), the bias-variance trade-off is non-trivial. Specifically, the asymptotic biases satisfy $|\abias_{\text{VAR},h}|=|\alpha| > 0 =|\abias_{\text{LP},h}|$ whenever $\alpha \neq 0$, while the asymptotic variances satisfy $\avar_{\text{LP},h}-\avar_{\text{VAR},h} = 1+\sigma_2^2+(h-2)[(1+\tau)^2+\sigma_2^2]>0$. Intuitively, LP directly projects $y_{t+h}$ on the shock $\varepsilon_{1,t}$, which is uncorrelated with any lagged controls, so the asymptotic bias is always zero. In contrast, the VAR(1) estimator extrapolates the response at horizon $h$: the model's structure implies that the precisely estimated autocovariances at lag 1 suffice to compute impulse responses at longer horizons. Though this tight parametric extrapolation yields a low variance relative to LP, it incurs a bias due to dynamic mis-specification when $\alpha \neq 0$. In the simple DGP \eqref{eq:DGP_simple}, the asymptotic bias of the VAR estimator could be eliminated by simply increasing the lag length to 2 or higher, but in practice it may be difficult to determine the appropriate lag length, as we demonstrate below in \cref{sec:results_selection}. The fact that both the asymptotic bias and variance of the VAR impulse response estimator are constant at horizons $h \geq 2$ is a special feature of the stylized DGP \eqref{eq:DGP_simple}.\footnote{In this DGP, the VAR coefficients on lagged $y_t$ are estimated super-consistently due to the unit root, so to first order, estimation uncertainty arises only from the coefficients on lagged $\varepsilon_{1,t}$.} Nevertheless, our simulation study below will demonstrate the robustness of the qualitative predictions that (i) the bias of VAR is high at intermediate and long horizons relative to LP, and (ii) the difference between the variance of LP and that of VAR tends to increase as a function of the horizon.\footnote{We derived similar analytical results for a stationary DGP in a previous working paper version of this article \citep{Li2022}.}

How does the optimal choice of estimator depend on the researcher's preferences concerning bias and variance? To evaluate the performance of a given estimator $\hat{\theta}_h$ of $\theta_{h,T}$, we will throughout this paper consider loss functions of the form\footnote{The objective function \eqref{eq:loss_simple} is not a loss function in the usual decision theoretic sense (which would call it a risk function when $\omega=\frac{1}{2}$). We proceed with the non-standard terminology for ease of exposition.}
\begin{equation}
	\mathcal{L}_\omega(\theta_{h,T}, \hat{\theta}_h) = \omega \times \left(\mathbb{E}[\hat{\theta}_h - \theta_{h,T}] \right )^2 + (1-\omega) \times \var(\hat{\theta}_h). \label{eq:loss_simple}
\end{equation}
For $\omega = \frac{1}{2}$, this is proportional to the mean squared error (MSE). For $\omega > \frac{1}{2}$, the researcher is more concerned about (squared) bias than variance, and for $\omega=1$ the researcher exclusively cares about bias. Substituting the asymptotic bias and variance expressions in \cref{prop:biasvar_simple} into the above loss function, we find that LP is preferred over VAR (asymptotically) if and only if the researcher prioritizes bias sufficiently heavily at the expense of variance, namely when $\omega \geq \omega_h^* \equiv 1-\alpha^2/(\alpha^2 + \avar_{\text{LP},h}-\avar_{\text{VAR},h}) \in (0,1)$ (focusing here on the interesting case $h \geq 2$ and $\alpha \neq 0$). We remark that---even in the very simple DGP \eqref{eq:DGP_simple}---the indifference weight $\omega_h^*$ depends sensitively on all the model parameters and the horizon $h$.

\subsection{Outlook}
\label{subsec:biasvar_outlook}

Because analytical bias-variance calculations will invariably end up depending in complicated ways on a multitude of parameters, we will in the rest of this paper use simulations to explore the nature of the bias-variance trade-off across a rich and empirically relevant set of DGPs. In the language of \cref{subsec:biasvar_illustration}, these DGPs will inform us about empirically plausible degrees of mis-specification $\alpha$, impulse response function shapes $\tau$, and relative shock importances $\sigma_2^2$, and therefore about the practically relevant bias weight $\omega_h^*$ necessary to justify the use of one linear projection technique over another one. Moreover, we will also consider several variants of the standard least-squares LP and VAR estimators, thus allowing us to further trace out the bias-variance possibility frontier.

\section{Data generating processes}
\label{sec:dgp}

This section presents our DGPs. We define the empirically calibrated encompassing model in \cref{sec:dgp_model}, from which we draw thousands of DGPs with corresponding structural impulse response estimands, as described in \cref{sec:dgp_estimand}. We discuss implementation details in \cref{sec:dgp_implementation}, and provide summary statistics for the DGPs in \cref{sec:dgp_summ_stat}. Various modifications to this baseline set of DGPs are  considered later in \cref{sec:results_robustness}.

\subsection{Encompassing model}
\label{sec:dgp_model}

We construct our simulation DGPs from an encompassing model that is known to accurately capture the time series properties of many U.S. macroeconomic time series: a dynamic factor model (DFM) fitted to the well-known \citet{Stock2016} data set. Because we seek to follow applied practice in using data in levels rather than first differences, we employ a \emph{non-stationary} variant of the DFM estimated by \citeauthor{Stock2016}.

The DFM postulates that a large-dimensional $n_X \times 1$ vector $X_t$ of observed macroeconomic time series is driven by a low-dimensional $n_f \times 1$ vector $f_t$ of latent factors, as well as an $n_X \times 1$ vector $v_t$ of idiosyncratic components. The latent factors are assumed to follow a non-stationary Vector Error Correction Model (VECM) with VAR($p_f$) representation
\begin{equation}
f_t = \Phi(L) f_{t-1} + H\varepsilon_t, \label{eq:factors}
\end{equation}
where $\varepsilon_t=(\varepsilon_{1,t},\dots,\varepsilon_{n_f,t})'$ is an $n_f \times 1$ vector of aggregate shocks, which are i.i.d. and mutually uncorrelated, with $\var(\varepsilon_t) = I_{n_f}$. The $n_f \times n_f$ matrix $H$ determines the impact impulse responses of the factors with respect to the aggregate shocks. The observed macroeconomic series $X_t$ are given by
\begin{equation}
X_t = \Lambda f_t + v_t, \label{eq:observables}
\end{equation}
where the idiosyncratic component $v_{i,t}$ for macro observable $X_{i,t}$ follows the potentially non-stationary AR($p_v$) process
\begin{equation}
v_{i,t} = \Gamma_i(L) v_{i,t-1} + \Xi_i \xi_{i,t}, \label{eq:idio_errors}
\end{equation}
with $\xi_{i,t}$ i.i.d. across $t$ and $i$. We assume that all shocks and innovations are jointly normal and homoskedastic. We will next in \cref{sec:dgp_estimand} describe how we construct our many lower-dimensional DGPs from this encompassing large-scale DFM; \cref{sec:dgp_implementation} then follows up with implementation details, including in particular a discussion of how the parameters of this non-stationary DFM are calibrated to the \citet{Stock2016} data set.

\subsection{DGPs and impulse response estimands}
\label{sec:dgp_estimand}

We use the encompassing model \eqref{eq:factors}--\eqref{eq:idio_errors} to build thousands of lower-dimensional DGPs for our simulation study. Specifically, for each DGP, we draw a random subset of $n_{\bar{w}}$ variables $\bar{w}_t$ from the large vector $X_t$, i.e., $\bar{w}_t \subset X_t$. The variables $\bar{w}_t$ follow the time series process implied by the encompassing model \eqref{eq:factors}--\eqref{eq:idio_errors}. In particular, $\bar{w}_t$ is driven by some combination of aggregate structural shocks $\varepsilon_t$ and idiosyncratic components $v_t$. We draw thousands of such random combinations of variables, thus yielding thousands of lower-dimensional DGPs. The details of how we select the variable combinations are postponed until \cref{sec:dgp_implementation}.

For each DGP drawn in this way, we consider three types of structural impulse response estimands, chosen to mimic as closely as possible popular schemes for identifying the effects of policy shocks in applied macroeconometrics \citep{Ramey2016,Stock2016}. In the following, $y_t \in \bar{w}_t$ denotes a response variable of interest in the DGP, $i_t \in \bar{w}_t$ is a policy variable used to normalize the scale of the shock (if applicable), $z_t$ is an external instrument (if applicable), and $w_t$ denotes the vector of all observed time series in the DGP.

\begin{enumerate}[1.]

\item {\bf Observed shock identification.} In this identification scheme we assume that the econometrician observes both the endogenous variables $\bar{w}_t$ and the first structural shock $\varepsilon_{1,t}$, so the full vector of observables is $w_t = (\varepsilon_{1,t}, \bar{w}_t)'$. The objects of interest are the impulse responses of an outcome variable $y_t$ with respect to a one standard deviation (i.e., one unit) innovation to $\varepsilon_{1,t}$:
\begin{equation}
\theta_h \equiv \bar{\Lambda}_{\iota_y, \bullet} \Theta_{\bullet, 1, h}^f, \quad h=0,1,2,\dots, \label{eq:IRF_obsshock}
\end{equation}
where $\Theta^f(L)$ are the impulse responses of the factors $f_t$ to the structural shocks $\varepsilon_t$ implied by \eqref{eq:factors}, while $\bar{\Lambda}$ are those rows of $\Lambda$ that correspond to the observables $\bar{w}_t$. The index $\iota_y$ corresponds to the location of $y_t$ in the vector $\bar{w}_t$.

This set-up captures those empirical studies in which the researcher has constructed a plausible direct measure of the shock of interest. Examples include the monetary shock series of \citet{Romer2004} or the fiscal shock series of \citet{Ramey2011}. While one may worry about measurement error in practice, it is common in applied work to treat shocks as known, so we include this identification approach as a useful baseline. Measurement error is introduced in the next identification scheme.

\item {\bf IV/proxy identification.} In this scheme, instead of directly observing the structural shock $\varepsilon_{1,t}$, the econometrician observes the noisy proxy
\begin{equation}
z_t = \rho_z z_{t-1} + \varepsilon_{1,t} + \nu_t, \label{eq:IV}
\end{equation}
where $\nu_t$ is an i.i.d. process (independent of all shocks and innovations in the DFM) with $\var(\nu_t) = \sigma_\nu^2$. The full vector of observables is thus $w_t = (z_t, \bar{w}_t)'$. As is standard in IV applications, we here adopt the ``unit effect'' normalization of \citet{Stock2016}, so the object of interest becomes
\begin{equation}
\theta_h \equiv \frac{\bar{\Lambda}_{\iota_y, \bullet} \Theta_{\bullet, 1, h}^f}{\bar{\Lambda}_{\iota_i, \bullet} \Theta_{\bullet, 1, 0}^f}, \quad h=0,1,2,\dots, \label{eq:IRF_IV}
\end{equation}
where the index $\iota_i$ corresponds to the location of a policy variable $i_t$ in the vector $\bar{w}_t$. The above unit effect normalization defines the magnitude of the shock $\varepsilon_{1,t}$ such that it raises the policy variable $i_t$ by one unit on impact.

One example of an IV $z_t$ is the high-frequency change in futures prices around monetary policy announcements employed by \citet{Gertler2015} to identify the effects of monetary policy shocks.

\item {\bf Recursive identification.} The final identification scheme is recursive (Cholesky) shock identification \citep[e.g.,][]{Christiano1999}. Because it turns out that the simulation results for such shocks are qualitatively similar to the results when the shock is directly observed, we relegate discussion of recursive identification to a robustness check in \cref{sec:results_robustness}, with technical definitions in \cref{app:recursive}.

\end{enumerate}

\subsection{Implementation}
\label{sec:dgp_implementation}

This section first discusses how we estimate the DFM and then specifies the particular DGPs and structural impulse responses that we consider in the simulation study.

\paragraph{DFM parameters.}
We parametrize the DFM \eqref{eq:factors}--\eqref{eq:idio_errors} by estimating the model on the \citet{Stock2016} data set. Recall that we model the variables in levels rather than first differences, unlike \citeauthor{Stock2016}. We provide a brief overview of our approach here, with details in \cref{app:dfm_estim}.

We begin with the vector of observables $X_t$. As in \citet{Stock2016}, that vector contains quarterly observations on 207 time series for 1959Q1--2014Q4, mostly consisting of real activity variables, price measures, interest rates, asset and wealth variables, and productivity series.\footnote{Table 1 and the Data Appendix of \citet{Stock2016} list all variables and their categories.} Each series is seasonally adjusted as in \citet{Stock2016}. However, unlike those authors we do not transform the series to stationarity. Instead, variables that they transform to (non-log or log) first differences, we now keep in (non-log or log) levels; and variables that they transform to log second differences (which are mostly price indices), we only transform to log first differences. \alert{All estimation procedures mentioned below control for series-specific and common deterministic linear time trends.}

Our estimation approach is intended to allow for rich long- and short-run dynamics, opening the door for meaningful mis-specification of short-lag VARs. Following \citet{Bai2004} and \citet{Barigozzi2021}, we estimate the non-stationary DFM by extracting factors from differenced \alert{(and subsequently de-meaned)} data, cumulating these factors, and then fitting a VECM to the cumulated factors; the aforementioned papers show that this estimation strategy consistently estimates the true VECM parameters under weak conditions that allow for cointegration. We set the number of factors $n_f$ equal to 6 as in \citet{Stock2016}. \alert{The VECM is estimated by quasi-maximum-likelihood without restricting the adjustment coefficients or cointegrating relations.}\footnote{\alert{The reason we fit a VECM rather than an unrestricted VAR in levels to the factors is that the VAR estimator may underestimate persistence in finite samples, as is well known. While bias correction procedures exist, they may not always work well in practice. The VECM approach instead errs on the side of overstating the role of permanent shocks, consistent with our goal of allowing for rich long-run dynamics.}} The cointegration rank of the VECM is selected by the \citet{Johansen1995} maximum eigenvalue test, which indicates that the latent factors are driven by four common stochastic trends. As in \citet{Stock2016}, we fit AR($p_v$) processes by OLS to each idiosyncratic residual after removing the estimated factors, separately for each $i$. We use lag lengths $p_f=p_v = 4$ for both the factor process and the idiosyncratic component processes, which is at the upper end of what is preferred by the Akaike Information Criterion, consistent with our goal of allowing rich dynamics. The above-mentioned estimation procedure pins down all parameters of the DFM except for the structural impact response matrix $H$; we discuss below how we construct that matrix.

While the estimated encompassing DFM assumes a cointegrated VECM for the latent factors, the lower-dimensional DGPs that we subsequently extract from the DFM will not satisfy exact finite-order VECM or VAR processes, and will not be exactly cointegrated. This follows from the presence of the idiosyncratic components and the mismatch in dimensions between the latent factors and the observable series (as specified below). Indeed, we show below that most of the lower-dimensional DGPs feature a combination of exact unit roots (imposed in the factor VECM), several roots near unity (owing partly to the factor process and partly to the idiosyncratic component processes), as well as smaller roots that induce transitory dynamics. Our DGPs are therefore consistent with the common empirical finding that there is often substantial ambiguity about the appropriate VAR lag lengths, the exact magnitude of roots, and the presence or absence of cointegrating relationships.

\paragraph{DGP and estimand selection.}
To provide a comprehensive picture of the bias-variance trade-off, we select thousands of different sets of observables $\bar{w}_t \subset X_t$. We consider two protocols for selecting these observables---one aimed at mimicking monetary policy shock applications, and one aimed at fiscal policy shock applications. Specifically, for each type of policy shock, we randomly draw 3,000 configurations of $n_{\bar{w}} = 5$ macroeconomic observables $\bar{w}_t$. Thus, we end up with a total of 6,000 DGPs. For the monetary policy DGPs we restrict $\bar{w}_t$ to always contain the federal funds rate, while for the fiscal policy DGPs we restrict $\bar{w}_t$ to contain federal government spending. These two series are chosen as the policy variables $i_t$ for the IV and recursive estimands. The remaining four variables in $\bar{w}_t$ are then selected uniformly at random from $X_t$, except we impose that at least one variable should be a measure of real activity, and at least one other variable a measure of prices.\footnote{Real activity series correspond to categories 1--3 in the classification in Table 1 of \citet{Stock2016}, while price series correspond to category 6.} The impulse response variable $y_t$ is selected uniformly at random from the four series (other than $i_t$).

For each of the DGPs, we implement the structural impulse response estimands as follows:

\begin{enumerate}[1.]

\item {\bf Observed shock.} We select the structural impact response matrix $H$ in the factor equation \eqref{eq:factors} so as to maximize the impact effect of the shock $\varepsilon_{1,t}$ on the federal funds rate (for monetary shocks) and government spending (for fiscal shocks), subject to the constraint that $H$ is consistent with our estimate of the reduced-form innovation variance-covariance matrix for the factors. This ensures that monetary and fiscal shocks account for substantial short-run variation in nominal interest rates and government spending, respectively. Additionally, we avoid issues related to division by near-zeros when normalizing the impulse responses for the IV estimand. See \cref{app:subsec:dgp_select} for further details.

\item {\bf IV.} The matrix $H$ is defined just as in the ``observed shock'' case. Next, turning to the IV parameters in equation \eqref{eq:IV}, we draw $\rho_z$ uniformly at random from the set $\{ 0, 0.25, 0.5 \}$.\footnote{The external IVs used in empirical practice tend to have low to moderate autocorrelation \citep{Ramey2016}, consistent with our assumptions on $\rho_z$.} To ensure an empirically plausible signal-to-noise ratio, we calibrate $\sigma_\nu^2$ to three different values that yield population IV first-stage F-statistics between 10 and 30, roughly in line with heterogeneity in applied practice. See \cref{app:subsec:iv_calibration} for details.

\item {\bf Recursive identification.} Implementation details are in \cref{app:recursive}.

\end{enumerate}

\subsection{Summary statistics}
\label{sec:dgp_summ_stat}
Consistent with the experience of applied researchers, our DGPs exhibit substantial heterogeneity along several dimensions. \cref{tab:dgp_summ} displays the distribution of various population parameters across our 6,000 DGPs. The table focuses on impulse responses with respect to directly observed monetary policy and government spending shocks, though results for recursively defined shocks are similar, as shown in \cref{app:results_recursive}.

\begin{table}[t]
\centering
\textsc{DGP summary statistics} \\[0.5\baselineskip]
\renewcommand{\arraystretch}{1.2}
\begin{tabular}{l|ccccccc}
Percentile	&	min	&	10	&	25	&	50	&	75	&	90	&	max	\\
\hline\hline
& \multicolumn{6}{l}{ } \\[-2ex]
\emph{Data and shocks} & \multicolumn{6}{l}{} \\[0.5ex]
trace(long-run var)/trace(var)	&	0.03	&	0.27	&	0.54	&	1.02	&	1.98	&	3.54	&	23.73	\\
Fraction of VAR coef's $\ell \geq 5$	&	0.07	&	0.14	&	0.17	&	0.23	&	0.29	&	0.37	&	0.81	\\
Degree of shock invertibility	&	0.24	&	0.30	&	0.34	&	0.39	&	0.44	&	0.49	&	0.65	\\
IV first stage F-statistic	&	7.18	&	7.91	&	10.55	&	21.13	&	24.20	&	33.29	&	33.97	\\
\hline
& \multicolumn{6}{l}{ } \\[-2ex]
\emph{Impulse responses up to $h=20$} & \multicolumn{6}{l}{} \\[0.5ex]
No. of interior local extrema	&	0	&	1	&	2	&	2	&	3	&	3	&	5	\\
Horizon of max abs. value	&	0	&	0	&	1	&	4	&	8	&	19	&	20	\\
Average/(max abs. value) &	-0.87 	&	-0.67 	&	-0.48	 &	0.01  	&	0.33  	&	0.64  	&	0.89	\\
$R^2$ in regression on quadratic	&	0.04	&	0.46	&	0.70	&	0.85	&	0.95	&	0.98	&	1.00	 \\
\hline
\end{tabular}
\caption{Quantiles of various population parameters across the DGPs for observed shock and IV identification. ``long-run var'': long-run variance of differenced series. ``var'': variance of differenced series. ``Fraction of VAR coef's $\ell \geq 5$'': $\sum_{\ell=5}^{1000} \|A_\ell^w\|/\sum_{\ell=1}^{1000} \|A_\ell^w\|$, where $A_\ell^w$ are the population VAR($\infty$) coefficient matrices and $\|\cdot\|$ is the Frobenius norm. ``Degree of shock invertibility'': $R^2$ in a projection of $\varepsilon_{1,t}$ on $\lbrace \bar{w}_{t-\ell} \rbrace_{\ell=0}^\infty$. ``IV first stage F-statistic'': $T \times R^2/(1-R^2)$, where $T=200$ and $R^2$ is the population $R^2$ in a projection of $i_t$ on $z_t$, controlling for $\lbrace w_{t-\ell} \rbrace_{\ell=1}^\infty$. ``Average/(max abs. value)'': $(\frac{1}{21}\sum_{h=0}^{20}\theta_h)/\max_h \lbrace |\theta_h|\rbrace$. ``$R^2$ in regression on quadratic'': $R^2$ from a regression of the impulse response function $\lbrace \theta_h \rbrace_{h=0}^{20}$ on a quadratic polynomial in $h$.}
\label{tab:dgp_summ}
\end{table}

First of all, the DGPs feature varying degrees of persistence. All DGPs have unit roots by construction; nevertheless, the DGPs differ in how heavily they load on the various non-stationary and stationary linear combinations of the latent factors. \cref{tab:dgp_summ} reports the ratio of the traces of the long-run variance matrix and variance matrix applied to differenced data, $\tr(\mathit{LRV}(\Delta \bar{w}_t))/\tr(\var(\Delta \bar{w}_t))$. This measure varies widely across the DGPs, with median equal to 1.02 (as when all series are simple random walks), and the 90th percentile equal to 3.54 (consistent with strong positive autocorrelation of the first differences). We will consider an alternative set of moderately persistent, stationary DGPs in one of our main robustness checks in \cref{sec:results_robustness}.

Second, the DGPs are heterogeneous in terms of how well they can be approximated by a low-order VAR. \cref{tab:dgp_summ} reports the ratio $\sum_{\ell=5}^{1000} \|A_\ell^w\|/\sum_{\ell=1}^{1000} \|A_\ell^w\|$, which measures the relative magnitude of the coefficient matrices $\lbrace A_\ell^w \rbrace_{\ell}$ in the VAR($\infty$) representation for $\lbrace \bar{w}_t \rbrace$ at or after lag 5 (with $\|\cdot\|$ here denoting the Frobenius matrix norm). The 10th and 90th percentiles equal 0.14 and 0.37, respectively. Hence, the analysis in \cref{sec:biasvar} suggests that the bias of low-order VAR procedures will vary substantially across the various DGPs that we consider in our simulations.

Third, for the IV specifications, we note that our DGPs differ in terms of shock invertibility and IV strength. The degree of invertibility is defined as the R-squared in a population projection of the shock of interest on current and lagged macro observables $\lbrace \bar{w}_{t-\ell} \rbrace_{\ell=0}^\infty$.\footnote{Projections on infinite collections of lagged variables are defined as the limit when the lag length tends to $\infty$, using a diffuse initialization of the Kalman filter.} The bias of some SVAR-based external instrument procedures depends on how far below 1 this measure is, as discussed further in \cref{sec:estim}. The table shows that 90\% of the DGPs have degrees of invertibility below 49\%, i.e., substantial non-invertibility.\footnote{\citet{Leeper2013} argue that adding forward-looking variables to a VAR ameliorates the invertibility problem. However, if we restrict attention to the 1,457 DGPs that contain at least one time series in the ``Asset Price \& Sentiment'' category (see \cref{app:results_cat}), the 90th percentile of the degree of invertibility increases only marginally to 51\%.} This is not surprising: the DFM \eqref{eq:factors}--\eqref{eq:idio_errors} features a realistic amount of idiosyncratic noise $v_t$, making it challenging to accurately back out the aggregate shock $\varepsilon_{1,t}$ from a small number of time series $\bar{w}_t$. The strength of the IV is by construction borderline weak to moderate, as the population first stage F-statistic (from a regression of the policy variable $i_t$ on the IV $z_t$, controlling for lagged data) is calibrated to vary between approximately 10 and 30, given sample size $T=200$.

\begin{figure}[tp]
\centering
\textsc{Selected impulse response function estimands} \\[0.5\baselineskip]
\includegraphics[width=0.9\linewidth]{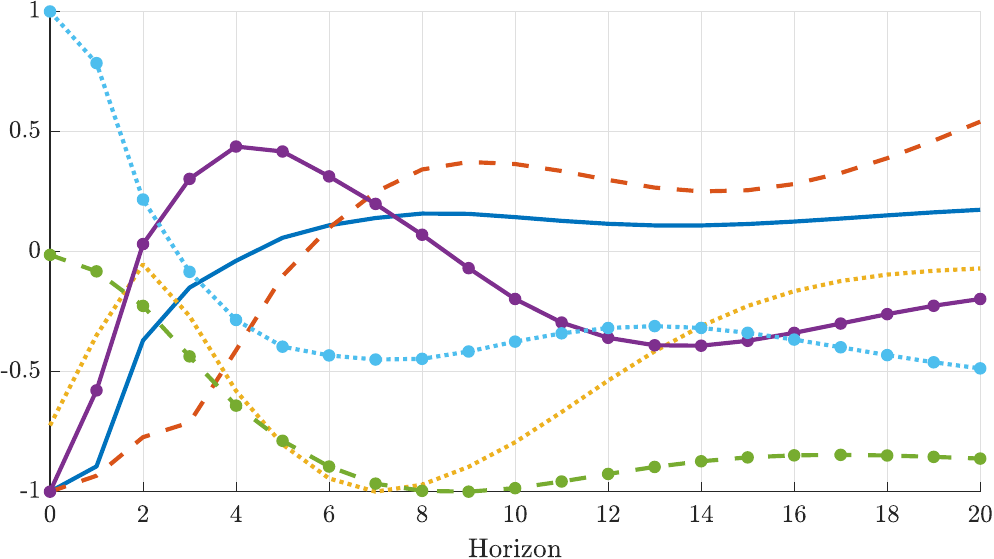}
\caption{Selected impulse responses of macro observables to monetary and fiscal policy shocks. Here the impulse response functions are normalized to have a maximum value of $1$ or $-1$.}
\label{fig:dgp_irfs}
\end{figure}

Finally, the true impulse response estimands exhibit a wide variety of shapes. \cref{tab:dgp_summ} shows that the impulse response functions peak at very different horizons and are typically not simple monotonically decaying or even hump-shaped functions: the median number of interior local extrema of the impulse response functions is 2 (a monotonic function would have 0; a hump-shaped function would have 1). Many impulse response functions change sign at some horizon, as evidenced by the average response (across horizons) typically being much smaller than the maximal response. Finally, the smoothness of the impulse response functions varies substantially: the R-squared value in a regression of the impulse responses $\lbrace \theta_h \rbrace_{h=0}^{20}$ on a quadratic polynomial $b_0+b_1 \times h + b_2 \times h^2$ has 10th and 90th percentiles given by 0.46 and 0.98, respectively. For further illustration, \cref{fig:dgp_irfs} displays the true values of six impulse response functions, providing a representative picture of the heterogeneity. The figure illustrates that, while some impulse response functions approximately return to 0 at long horizons, many do not, and some have the largest response even beyond horizon $h = 20$.

\section{Estimation methods}
\label{sec:estim}

We now give a brief overview of the different VAR- and LP-based estimation methods that we consider in the simulation study.\footnote{To visualize the various estimation methods, \cref{app:results_irfs} plots the estimated impulse response functions in a few data sets simulated from a single DGP.} Though all these methods aim at estimating the same population impulse responses defined in \cref{sec:dgp_estimand}, they differ in terms of their bias-variance properties, and in terms of their robustness to non-invertibility. Further implementation details are relegated to \cref{app:estimators}. All estimators include an intercept.

\paragraph{Local projection approaches.}
The basic idea behind local projections, as proposed by \citet{Jorda2005}, is to estimate the impulse responses separately at each horizon by a direct regression of the future outcome on current covariates. We consider three such approaches:
\begin{enumerate}
\item \textbf{Least-squares LP.} OLS regression of the response variable $y_{t+h}$ on some innovation variable $x_t$, controlling for $p$ lags of all data series $w_t$. The innovation variable equals $x_t=\varepsilon_{1,t}$ for ``observed shock'' identification. For recursive identification, $x_t$ equals the policy variable $i_t$, and we additionally control for the contemporaneous values of the variables that are ordered before $i_t$ in the system \citep{Plagborg2020}. For IV identification, we set $x_t=i_t$ and instrument for this variable using the IV $z_t$ (this is the LP-IV estimator of \citealp{Stock2018}). Since least-squares LP does not mechanically impose any functional form on the relationship between impulse responses at different horizons $h$, it does not suffer from extrapolation bias. However, these estimated impulse response functions tend to look jagged in finite samples and be estimated with high variance at longer horizons.
\item \textbf{Bias-corrected LP} (abbreviated ``BC LP''). \citet{Herbst2021} propose a bias-corrected version of LP, which partially removes the bias that is due to high persistence in the data. Though this bias is theoretically of order $T^{-1}$ (where $T$ is the sample size) and thus asymptotically negligible relative to the standard deviation, \citeauthor{Herbst2021} demonstrate that the bias can be sizable in sample sizes typical in the applied macroeconometrics literature.
\item \textbf{Penalized LP} (abbreviated ``Pen LP''). To lower the variance of least-squares LP at the expense of potentially increasing the bias, \citet{Barnichon2019} propose a penalized regression modification of LP. The estimator minimizes the sum of squared forecast residuals (across both horizons and time) plus a penalty term that encourages the estimation of smooth impulse responses. This is a type of shrinkage estimation: the unrestricted least-squares estimate is pushed in the direction of a smooth quadratic function of the horizon. The degree of shrinkage is chosen by cross-validation.
\end{enumerate}

\paragraph{VAR approaches.}
Like local projections, a VAR with lag length $p$ flexibly estimates the impulse responses out to horizon $p$; however, the VAR extrapolates the responses at longer horizons $h>p$ using only the sample autocovariances out to lag $p$. As suggested by the analysis in \cref{sec:biasvar}, this tends to generate impulse response estimates with lower variance but higher bias than LP estimates at intermediate and long horizons. We consider four such VAR-based approaches:
\begin{enumerate}
\item \textbf{Least-squares VAR.} Standard VAR impulse response estimates based on equation-by-equation OLS estimates of the reduced-form coefficients.
\item \textbf{Bias-corrected VAR} (abbreviated ``BC VAR''). As above, but follows \citet{Kilian1998} in using the formula in \citet{Pope1990} to analytically correct the order-$T^{-1}$ bias of the reduced-form coefficients caused by persistent data.\footnote{\citet[Chapter 12.3]{Kilian2017} argue that this analytical bias correction yields similar results to more computationally intensive bootstrap bias correction methods.}
\item \textbf{Bayesian VAR} (abbreviated ``BVAR''). As above, but where the reduced-form coefficients are estimated from a Bayesian VAR with automatic prior selection as in \citet{Giannone2015}. We report the posterior means of the impulse responses calculated from 100 draws. The prior specification follows the popular Minnesota prior, but with modifications that allow for cointegration. The prior variance hyper-parameters (and thus the degree of shrinkage) are chosen in a data-dependent way by maximizing the marginal likelihood.

\item \textbf{VAR model averaging} (abbreviated ``VAR Avg''). \citet{Hansen2016} develops a data-driven method for averaging across the impulse response estimates produced by several different VAR specifications. We construct a weighted average of 40 different specifications, each of which is estimated by OLS: univariate AR(1) to AR(20) models, and multivariate VAR(1) to VAR(20) models. The weights are chosen to minimize an empirical estimate of the final impulse response estimator's MSE.

The VAR model averaging estimator effectively includes LP among the list of candidate estimators (as in the related approach of \citealp{MirandaAgrippino2021}). This is because the candidate VAR(20) model gives results similar to LP with several lagged controls, at all horizons considered in our study \citep{Plagborg2020}.
\end{enumerate}
Observed shock identification is carried out by simply ordering the shock first in the recursive VAR. Recursive identification is implemented as usual in the VAR literature. We consider two different approaches to IV estimation:
\begin{enumerate}[i)]
\item \textbf{Internal instruments.} Proceed as if the IV were equal to the true shock of interest, i.e., order the IV first in the VAR and compute responses to the first orthogonalized innovation \citep{Ramey2011}. \citet{Plagborg2020} prove that this approach consistently estimates the normalized structural impulse responses \eqref{eq:IRF_IV} even if the IV is contaminated with measurement error as in \eqref{eq:IV}, and even if the shock is non-invertible.
\item \textbf{SVAR-IV} (also known as proxy-SVAR). Exclude the IV from the reduced-form VAR, and estimate the structural shock by projecting the IV on the reduced-form VAR innovations  \citep{Stock2008,Stock2012bpea,Mertens2013,Gertler2015}. This estimator is consistent if the shock of interest is invertible, but not otherwise \citep{Forni2019,Plagborg2020_var_decomp,MirandaAgrippino_invert}. We shall see that the SVAR-IV estimator tends to exhibit lower dispersion than the ``internal instruments'' estimator due to the smaller dimension of the VAR system.
\end{enumerate}
We implement the ``internal instruments'' approach using all four types of VAR estimation techniques described earlier. For brevity, we only consider the least-squares version of the ``external instrument'' SVAR-IV estimator.

\paragraph{Lag length selection.}
As a baseline, the LP and VAR estimators use $p=4$ lags for estimation (except of course VAR model averaging, which uses many different lag lengths). In our DGPs, the Akaike Information Criterion almost always selects very short lag lengths $\hat{p}_{AIC}$, as we discuss further in \cref{sec:results_selection} below. Thus, for all intents and purposes, our results may be interpreted as having been generated by the lag length selection rule $p = \max\lbrace \hat{p}_{AIC}, 4\rbrace$. Our reading of applied practice is that researchers typically include at least 4 lags in quarterly data. Results for $p=8$ are discussed in \cref{sec:results_robustness}.

\section{Results}
\label{sec:results}

This section presents our simulation results. We summarize the results through four lessons, presented in separate subsections. The first three lessons focus on observed shock identification. The fourth lesson is concerned with IV identification. We show in \cref{sec:results_robustness} that these conclusions are qualitatively robust to several alterations of our baseline simulation specification (including less persistent DGPs and recursive identification). Finally, in \cref{sec:results_selection}, we justify our focus on the \emph{average} performance of estimators across DGPs, by arguing that there is limited scope for selecting among estimators in a data-dependent way.

Throughout this section we present results for our 6,000 monetary and fiscal policy shock DGPs considered jointly rather than separately. For each DGP, we simulate time series of length $T=200$ quarters and approximate the population bias and variance of the estimators by averaging across 5,000 Monte Carlo simulations. The main results (excluding robustness checks) take about one week to produce in Matlab on a research computing cluster with 300 parallel cores.

\subsection{There is a clear bias-variance trade-off between LP and VAR}
\label{sec:results_1}

Our first takeaway is that researchers invariably face a bias-variance trade-off: because most of our DGPs are not well approximated by finite-order VAR models, least-squares LPs tend to have lower bias, while least-squares VAR estimators tend to have lower variance, consistent with the simple analytical example provided in \cref{sec:biasvar}. Strictly speaking, these statements are only exactly true for the bias-corrected versions of the estimators \citep{Herbst2021,Pope1990,Kilian1998}, as the high persistence of our DGPs imparts a sizable finite-sample bias in the estimators at intermediate and long horizons, particularly for LPs. This bias correction, however, is not a free lunch, as it increases variance.

\begin{figure}[tp]
\centering
\textsc{Observed shock: Bias of estimators} \\[0.5\baselineskip]
\includegraphics[width=0.85\linewidth]{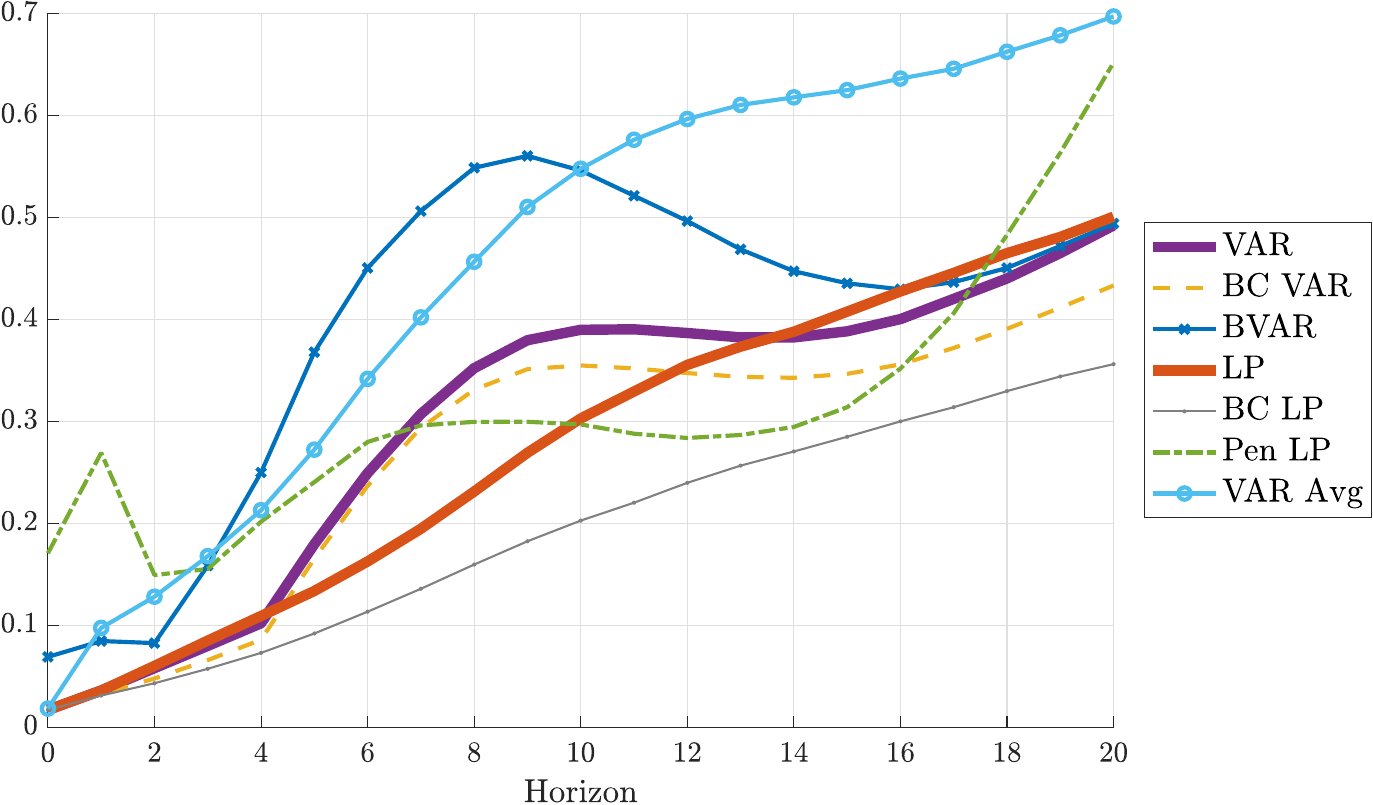}
\caption{Median (across DGPs) of absolute bias $|E(\hat{\theta}_h-\theta_h)|$ of the different estimation procedures, relative to $\sqrt{\frac{1}{21}\sum_{h=0}^{20}\theta_h^2}$.}
\label{fig:bias_obsshock}

\vspace*{\floatsep}

\centering
\textsc{Observed shock: Standard deviation of estimators} \\[0.5\baselineskip]
\includegraphics[width=0.85\linewidth]{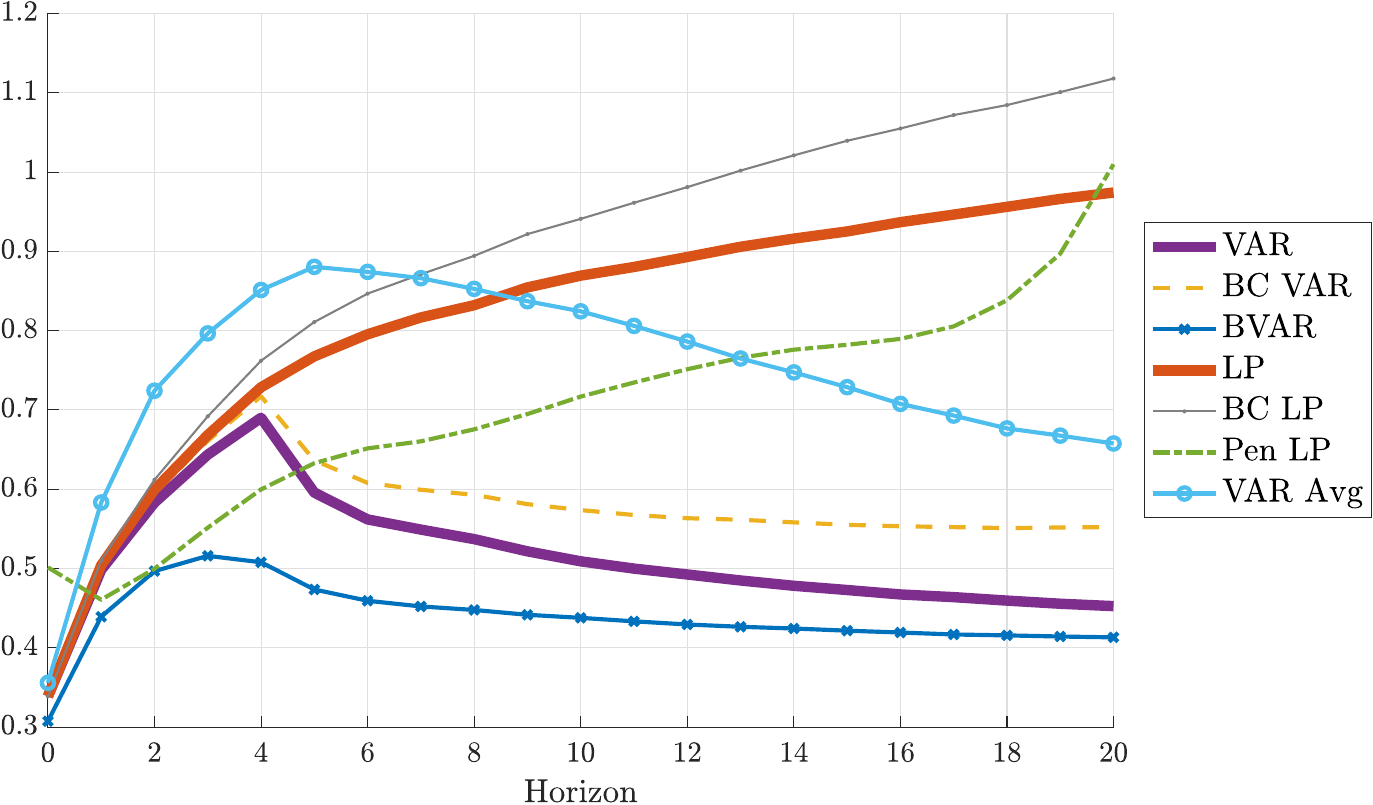}
\caption{Median (across DGPs) of standard deviation $\sqrt{\var(\hat{\theta}_h)}$ of the different estimation procedures, relative to $\sqrt{\frac{1}{21}\sum_{h=0}^{20}\theta_h^2}$.}
\label{fig:std_obsshock}
\end{figure}

\cref{fig:bias_obsshock,fig:std_obsshock} depict the bias-variance trade-off at various horizons. These figures show the median (across our 6,000 DGPs) of the absolute bias $|E(\hat{\theta}_h-\theta_h)|$ or the standard deviation $\sqrt{\var(\hat{\theta}_h)}$, respectively, as a function of the horizon. The different lines correspond to different estimators $\hat{\theta}_h$, with least-squares LP and VAR being the thick lines. Before taking the median, we cancel out the units of the response variables by dividing the bias and standard deviation by $\sqrt{\frac{1}{21}\sum_{h=0}^{20}\theta_h^2}$, i.e., the root mean squared value of the \emph{true} impulse response function out to horizon 20. Note that the scale of the vertical axis differs between the bias and standard deviation plots.

The figures show that least-squares LP and VAR estimators have similar bias and variance at horizons $h \leq p = 4$, but not at longer horizons $h > p$. The median biases then generally increase with the horizon, with the bias of VAR exceeding that of LP, except at long horizons.\footnote{\citet{Kilian2011} find in simulations that LP does not have lower bias than VAR estimators, but they consider a different variant of LP that uses an auxiliary VAR to identify the structural shocks.} While the median standard deviation of LP is increasing in the horizon, that of VAR instead displays a hump-shaped pattern. At long horizons, the median standard deviation of LP is about double that of VAR. These observations are broadly consistent with the asymptotic results in \cref{sec:biasvar}, \citet{Schorfheide2005}, and \citet{Plagborg2020}.

Our results also show that the bias correction procedure of \citet{Herbst2021} is critical to achieving uniformly low bias for the LP approach. Though the \emph{asymptotic} bias of LP is zero when the shock is observed, as discussed in \cref{sec:biasvar}, the high persistence of our DGPs implies that the small-sample bias of least-squares LP is non-negligible at intermediate and long horizons, especially the latter. The bias-corrected version of LP proposed by \citeauthor{Herbst2021} (thin line with small dots in the figures) eliminates about a third of the bias at all horizons. In comparison with the LP case, bias correction is not as critical for VAR estimation, though the bias-corrected VAR estimator (dashed line) does have a somewhat lower bias than the least-squares VAR estimator at long horizons. After bias correction, LP has lower (median) bias than VAR at all horizons, as predicted by asymptotic theory. We further show in \cref{sec:results_robustness} below that such bias correction is not nearly as important in less persistent, stationary DGPs.

\begin{figure}[p]
	\centering
	\textsc{Observed shock: Least-squares LP vs. Bias-Corrected LP} \\[0.5\baselineskip]
	\includegraphics[width=0.7\linewidth]{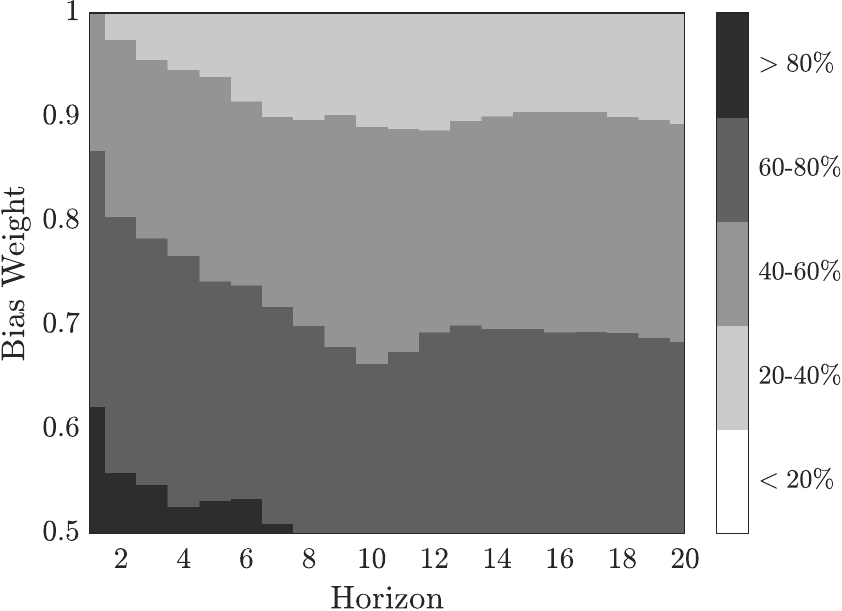}
	\caption{Fraction of DGPs for which the least-squares LP estimator has a lower loss than bias-corrected LP. The darker the region, the higher the fraction of DGPs for which least-squares LP is preferred. Horizontal axis: impulse response horizon $h$. Vertical axis: weight $\omega$ on squared bias in the loss function \eqref{eq:loss_simple}. The loss function is normalized by the scale of the true impulse response function, as in \cref{fig:bias_obsshock,fig:std_obsshock}. The impact horizon $h=0$ is omitted due to numerical equivalence between the estimators.}
	\label{fig:bclpvslp_obsshock}
	
	\vspace*{\floatsep}
	
	\centering
	\textsc{Observed shock: Least-squares VAR vs. Bias-Corrected VAR} \\[0.5\baselineskip]
	\includegraphics[width=0.7\linewidth]{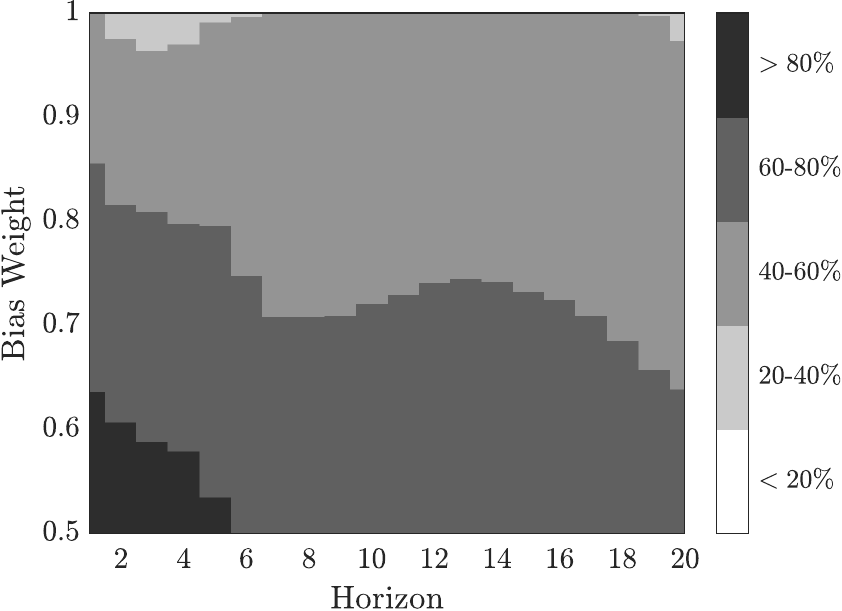}
	\caption{Fraction of DGPs for which the least-squares VAR estimator has a lower loss than the bias-corrected VAR estimator. The darker the region, the higher the fraction of DGPs for which least-squares VAR is preferred. See caption for \cref{fig:bclpvslp_obsshock}. The impact horizon $h=0$ is omitted due to numerical equivalence between the estimators.}
	\label{fig:bcvarvsvar_obsshock}
\end{figure}

Bias correction is not a free lunch, however, as it is associated with a substantial increase in variance. \cref{fig:std_obsshock} shows that the bias-corrected LP and VAR estimators have uniformly higher median standard deviation than the uncorrected estimators. In fact, bias-corrected LP has not only the uniformly lowest median bias among the methods we consider, it also has the uniformly highest standard deviation. \cref{fig:bclpvslp_obsshock,fig:bcvarvsvar_obsshock} show head-to-head comparisons of the least-squares and bias-corrected estimators for the LP and VAR cases, respectively. The figures show the fraction of DGPs for which the least-squares estimator achieves a lower loss \eqref{eq:loss_simple} than the bias-corrected estimator, as a function of the horizon $h$ and the weight $\omega$ attached to squared bias in the loss function; to interpret the figures, recall that $\omega=0.5$ corresponds to MSE loss, while $\omega=1$ corresponds to an exclusive focus on bias at the expense of variance. The darker the plot, the more often is the least-squares estimator preferred over the bias-corrected one. Evidently, one has to attach a very high weight $\omega$ to bias in the loss function to prefer the bias-corrected estimator in more than 60\% of DGPs; furthermore, a researcher with MSE loss would usually prefer the uncorrected estimators.

\subsection{Bias-corrected LP is the best estimator if and only if the researcher overwhelmingly prioritizes bias}
\label{sec:results_2}

Our second takeaway is that bias-corrected LP is the single best estimator in our choice set if \emph{and only if} the researcher's loss function overwhelmingly prioritizes bias. In contrast, uncorrected LP is \emph{never} the best option if the goal is to minimize average loss across our DGPs. Under MSE loss, penalized LP typically outperforms the other LP procedures as it has substantially lower variance, though at the expense of a moderate increase in bias.

\begin{figure}[t]
	\centering
	\textsc{Observed shock: Optimal estimation method} \\
	\includegraphics[width=0.9\linewidth,clip=true,trim=0 0.5em 0 2.4em]{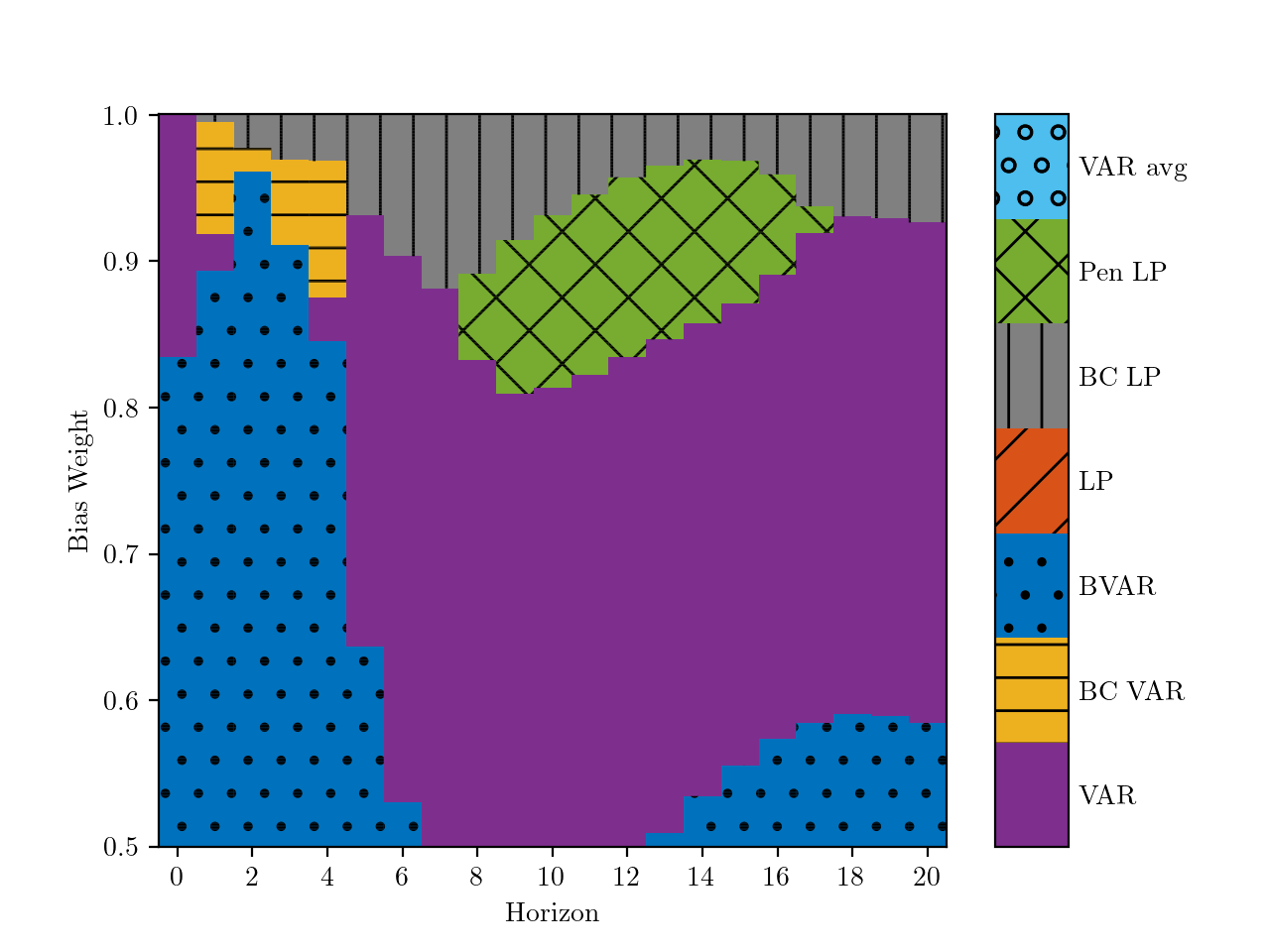}
	\caption{Method that minimizes the average (across DGPs) loss function \eqref{eq:loss_simple}. Horizontal axis: impulse response horizon. Vertical axis: weight on squared bias in loss function. The loss function is normalized by the scale of the impulse response function, as in \cref{fig:bias_obsshock,fig:std_obsshock}. At $h=0$, VAR and LP are numerically identical; we break the tie in favor of VAR.}
	\label{fig:bestmethod_obsshock}
\end{figure}

\cref{fig:bestmethod_obsshock} shows the optimal estimation method as a function of the horizon $h$ and the bias weight $\omega$. The colors and patterns indicate the estimation method that minimizes the \emph{average} loss \eqref{eq:loss_simple} across DGPs, after normalizing the loss to cancel out units as in \cref{fig:bias_obsshock,fig:std_obsshock}. In this subsection we focus on the top part of \cref{fig:bestmethod_obsshock}, i.e., where the weight $\omega$ on bias in the loss function is high. Bias-corrected LP emerges as the best estimator at most horizons in this case, as is to be expected given its excellent bias properties in \cref{fig:bias_obsshock}; nevertheless, the figure shows that the optimality of bias-corrected LP is predicated on $\omega$ exceeding roughly 0.9, or even higher at some horizons, corresponding to an overwhelming focus on minimizing bias rather than variance. In contrast, uncorrected least-squares LP is essentially dominated: it has greater bias than bias-corrected LP (notably at longer horizons), yet materially higher variance than least-squares VAR or other shrinkage methods, and so no part of \cref{fig:bestmethod_obsshock} is orange with diagonal lines. We discuss the rest of \cref{fig:bestmethod_obsshock} in the next subsection.

\begin{figure}[p]
	\centering
	\textsc{Observed shock: Bias-corrected LP vs. Bias-corrected VAR} \\[0.5\baselineskip]
	\includegraphics[width=0.7\linewidth]{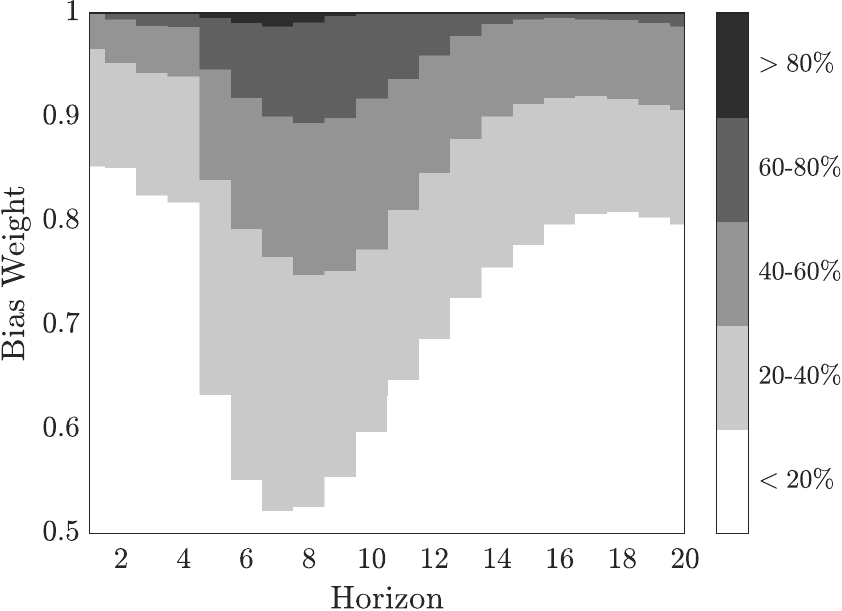}
	\caption{Fraction of DGPs for which the bias-corrected LP estimator has a lower loss than bias-corrected VAR. The darker the region, the higher the fraction of DGPs for which bias-corrected LP is preferred. Horizontal axis: impulse response horizon $h$. Vertical axis: weight $\omega$ on squared bias in the loss function \eqref{eq:loss_simple}. The loss function is normalized by the scale of the true impulse response function, as in \cref{fig:bias_obsshock,fig:std_obsshock}. The impact horizon $h=0$ is omitted due to numerical equivalence between the estimators.}
	\label{fig:bclpvsbcvar_obsshock}
	
	\vspace*{\floatsep}
	
	\centering
	\textsc{Observed shock: Bias-corrected LP vs. Penalized LP} \\[0.5\baselineskip]
	\includegraphics[width=0.7\linewidth]{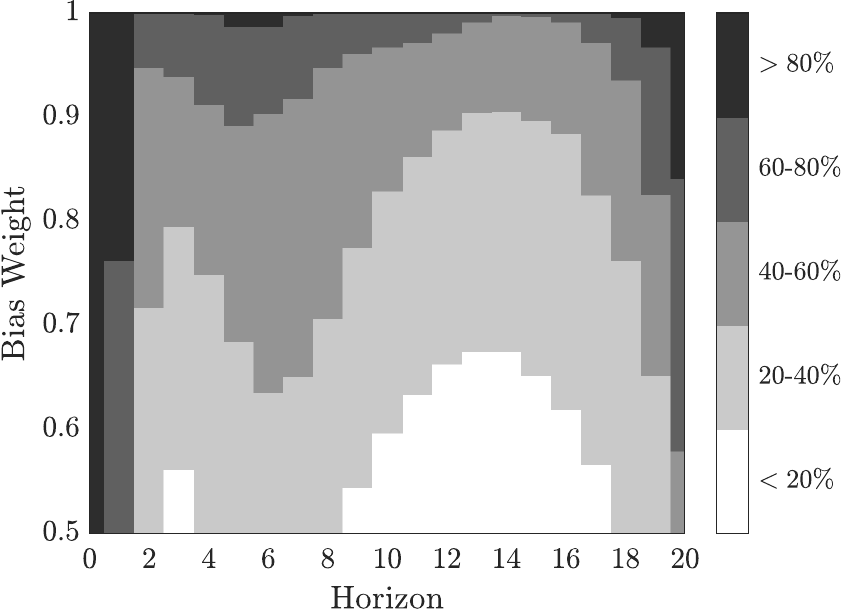}
	\caption{Fraction of DGPs for which the bias-corrected LP estimator has a lower loss than the penalized LP estimator. The darker the region, the higher the fraction of DGPs for which bias-corrected LP is preferred. See caption for \cref{fig:bclpvslp_obsshock}.}
	\label{fig:bclpvspenlp_obsshock}
\end{figure}

\cref{fig:bclpvsbcvar_obsshock,fig:bclpvspenlp_obsshock} compare bias-corrected LP to bias-corrected VAR and to penalized LP, respectively. The former figure shows that bias-corrected LP is only preferred to bias-corrected VAR in at least 60\% of DGPs when $\omega \geq 0.9$. In the latter figure, we see that the smoothing of impulse responses across horizons that the penalized LP estimator performs is usually attractive whenever $\omega \leq 0.9$, except at very short and very long horizons. By ``betting on smoothness'', penalized LP achieves a substantial variance reduction relative to the un-penalized LP procedures, at the expense of a moderate increase in bias, see \cref{fig:bias_obsshock,fig:std_obsshock}. In fact, there is a region of \cref{fig:bestmethod_obsshock} with intermediate horizons and moderately high weight on bias where penalized LP (green with diagonal cross-hatching) is the single best estimator. These findings underscore our conclusion that, across the majority of the DGPs, the use of bias-corrected LP can only be justified by committing to a nearly exclusive focus on minimizing bias, with little regard for precision.

\subsection{VARs are attractive if there is some concern for precision}
\label{sec:results_3}

Our third takeaway is that VAR estimators are attractive to researchers who place at least moderate weight on variance in their loss function. But the choice of VAR method depends on the horizon: Bayesian VARs perform well at short horizons, least-squares VARs at intermediate horizons, and at long horizons the two are comparable. VAR model averaging, on the other hand, performs poorly regardless of bias-variance preferences.

Returning to \cref{fig:bestmethod_obsshock}, we see that for bias weights $\omega$ below 0.9, the optimal estimation method is almost always either least-squares VAR (purple areas) or BVAR (solid-dotted blue). The key attractive property of BVAR is that it has the lowest (median) standard deviation at all horizons among the methods we consider, as seen in \cref{fig:std_obsshock}, though it also has high bias relative to least-squares VAR at intermediate horizons, as shown in \cref{fig:bias_obsshock}. The relatively high bias at intermediate horizons is possibly due to the fact that its prior specification, which is conventional in the literature, is motivated by one-step-ahead and long-run forecasting properties, as opposed to medium-run properties.\footnote{Moreover, the \citet{Giannone2015} approach of choosing the prior hyper-parameters to maximize the marginal likelihood implicitly targets one-step-ahead forecasts (see Equation 5 in their paper).}

\begin{figure}[t]
	\centering
	\textsc{Observed shock: Least-squares VAR vs. Bayesian VAR} \\[0.5\baselineskip]
	\includegraphics[width=0.7\linewidth]{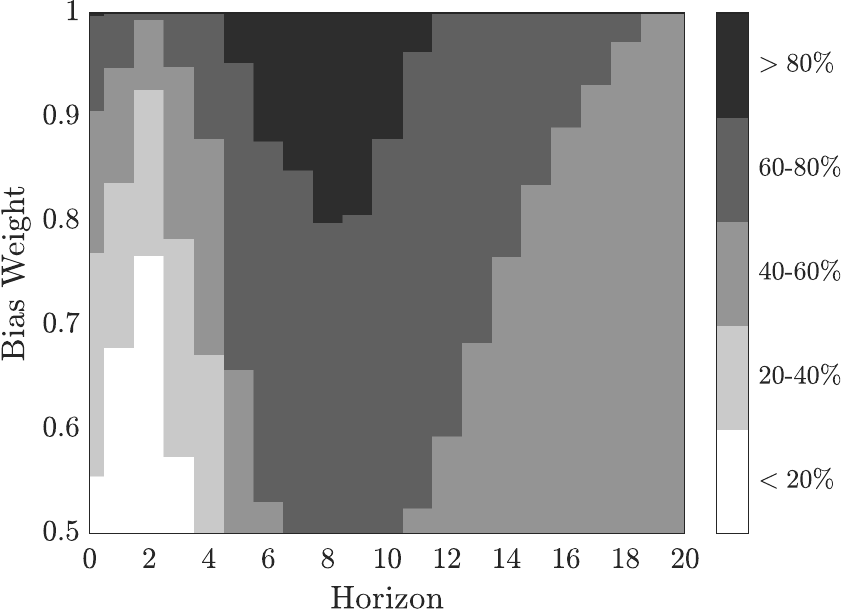}
	\caption{Fraction of DGPs for which the least-squares VAR estimator has a lower loss than the BVAR estimator. The darker the region, the higher the fraction of DGPs for which least-squares VAR is preferred. See caption for \cref{fig:bclpvslp_obsshock}.}
	\label{fig:bvarvsvar_obsshock}
\end{figure}

\cref{fig:bvarvsvar_obsshock} shows that the head-to-head performance of least-squares VAR vs. Bayesian VAR depends on the horizon. At short horizons $h \leq 4$, BVAR is preferred in the majority of DGPs, and indeed it is the overall best estimator for most loss functions that place non-trivial weight on variance (see \cref{fig:bestmethod_obsshock}). However, at intermediate horizons $h \in [5,12]$, least-squares VAR is preferred over BVAR in the clear majority of DGPs for most loss functions, and the former estimator is the overall preferred method for loss functions with $\omega \leq 0.8$. At long horizons $h \geq 13$, the two VAR methods are comparable and outperform all other methods, unless the weight on bias in the loss function is high.

Finally, we remark that bias-corrected VAR and VAR model averaging are rarely, if ever, optimal. Bias-corrected VAR (yellow with horizontal lines in \cref{fig:bestmethod_obsshock}) can be rationalized at short horizons if the concern for bias is high, but the difference compared to least-squares VAR is small at these horizons, as discussed in \cref{sec:results_2}. VAR model averaging performs poorly regardless of loss function and horizon, as it has substantial bias as well as a high standard deviation relative to other VAR-based estimators (see \cref{fig:bias_obsshock,fig:std_obsshock}). Closer inspection reveals that the high standard deviation is a consequence of a very fat-tailed sampling distribution, with a non-negligible probability of erratic estimates.\footnote{We use \citeauthor{Hansen2016}'s (\citeyear{Hansen2016}) code off the shelf. It would be interesting to investigate whether the procedure could be modified to avoid erratic estimates, perhaps by regularizing the averaging weights.}

\subsection{SVAR-IV is heavily biased, but has relatively low dispersion}
\label{sec:results_4}

Our last takeaway is concerned with IV/proxy identification. Among the invertibility-robust ``internal instruments'' estimators, the bias-variance trade-off is very similar to that already discussed above for the case of an observed shock. The alternative ``external instruments'' SVAR-IV procedure, however, contributes starkly to the trade-off: it can be severely biased due to its lack of robustness to non-invertibility, but at the same time it also has substantially lower dispersion than the ``internal instruments'' procedures.

Since first and second moments of IV estimators may not exist theoretically \citep{Sawa1972}, we in this subsection report median bias (i.e., in each DGP, the median of the estimation error) instead of (mean) bias, and the interquartile range instead of the standard deviation.\footnote{For completeness, (mean) bias and standard deviation are reported in \cref{app:results_iv}.} We refer to the latter as ``dispersion.''

\begin{figure}[tp]
\centering
\textsc{IV: Median bias of estimators} \\[0.5\baselineskip]
\includegraphics[width=0.85\linewidth]{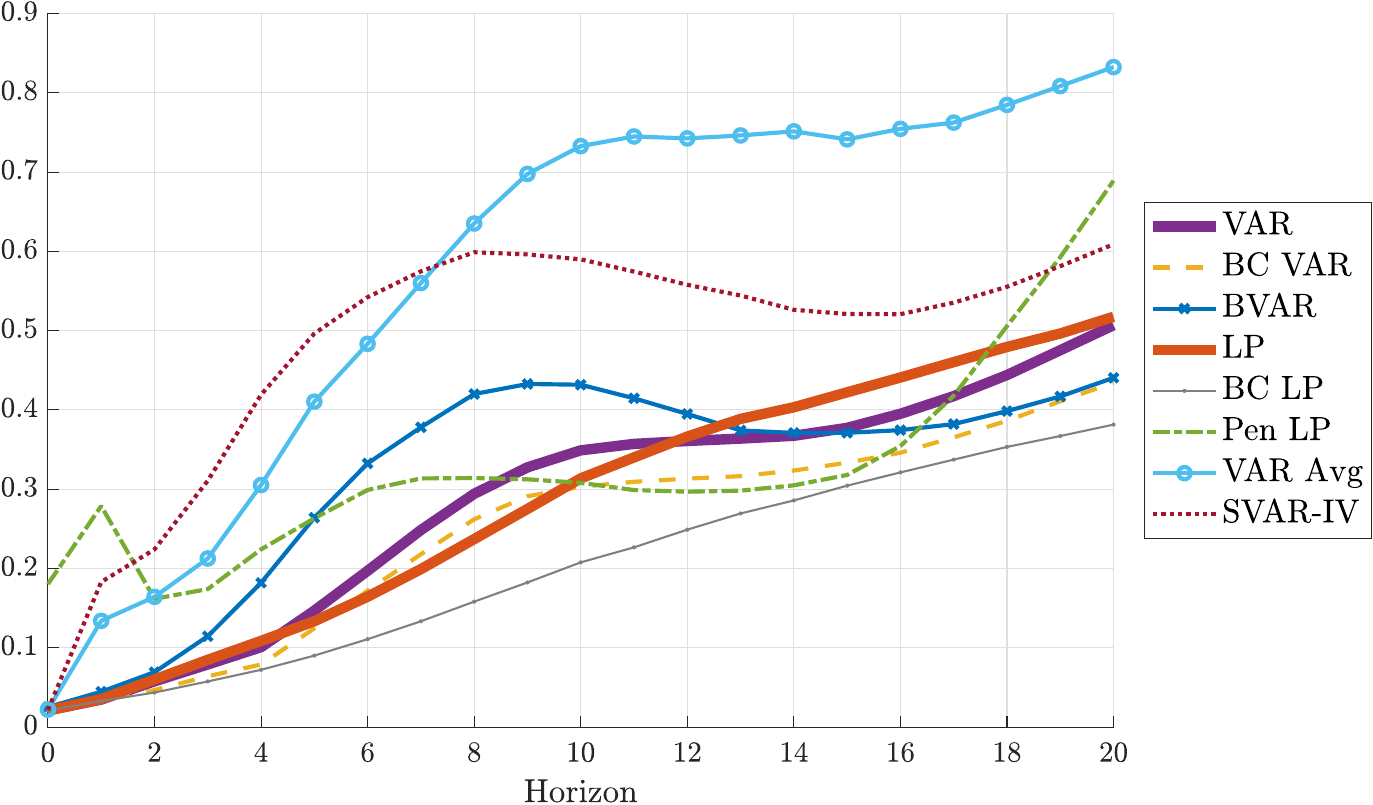}
\caption{Median (across DGPs) of absolute median bias of the different estimation procedures, relative to $\sqrt{\frac{1}{21}\sum_{h=0}^{20}\theta_h^2}$. The first seven estimators listed in the figure legend include the proxy/IV directly in the observed data vector, see \cref{sec:estim}.}
\label{fig:bias_iv}

\vspace*{\floatsep}

\centering
\textsc{IV: Interquartile range of estimators} \\[0.5\baselineskip]
\includegraphics[width=0.85\linewidth]{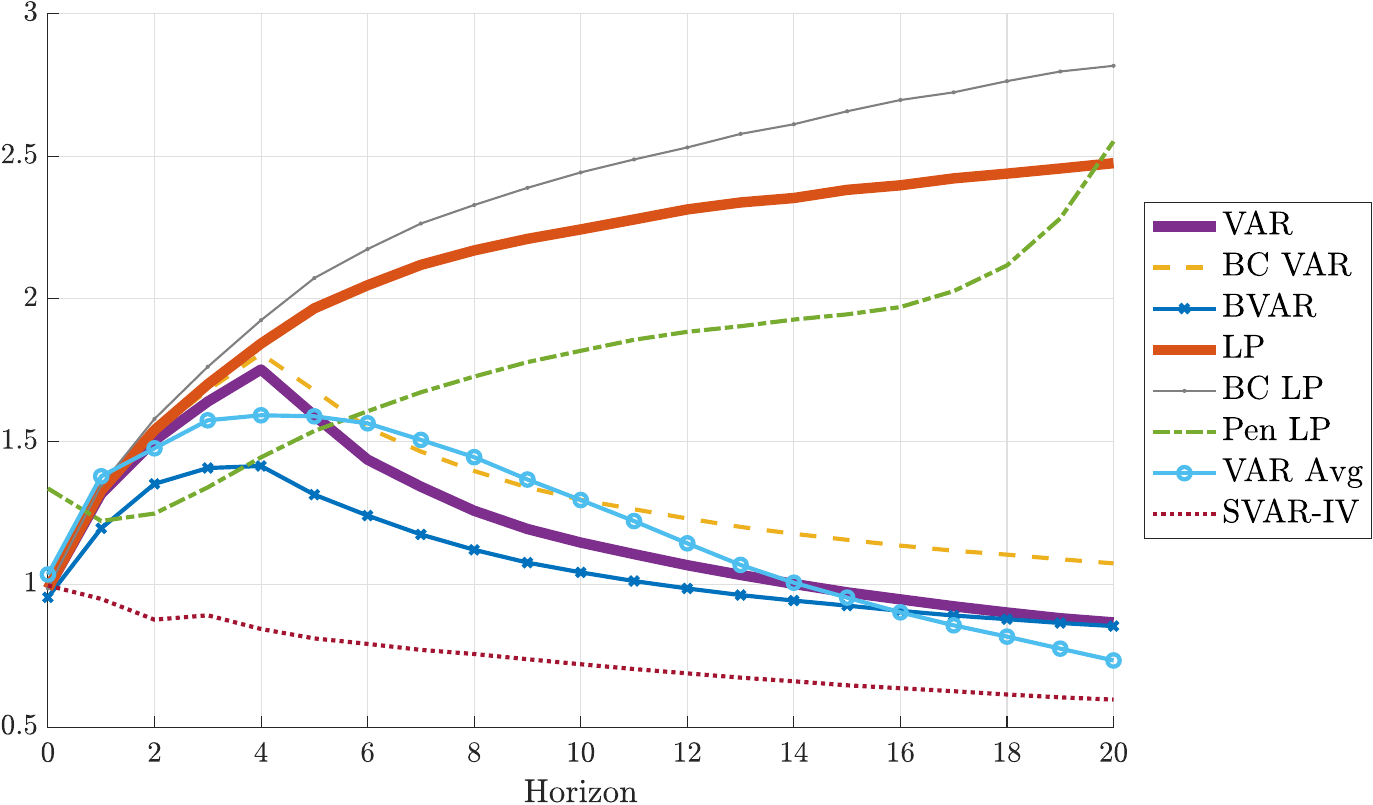}
\caption{Median (across DGPs) of interquartile range of the different estimation procedures, relative to $\sqrt{\frac{1}{21}\sum_{h=0}^{20}\theta_h^2}$. See caption for \cref{fig:bias_iv}.}
\label{fig:std_iv}
\end{figure}

\cref{fig:bias_iv,fig:std_iv} show the median bias and interquartile range of the various IV estimators.  If we ignore the dotted line representing SVAR-IV, these figures are qualitatively similar to those presented in \cref{sec:results_1}. However, SVAR-IV stands out by exhibiting especially high median bias and especially low interquartile range at all horizons. This is consistent with the existing theoretical work referenced in \cref{sec:estim}: unlike the ``internal instruments'' procedures, SVAR-IV is asymptotically biased when the shock is not invertible, and we saw in \cref{sec:dgp_summ_stat} that the degree of invertibility is generally low in our DGPs.\footnote{\alert{Consistent with theory, we furthermore find that the median bias of SVAR-IV is particularly large relative to other estimation methods in the subset of DGPs with the smallest degree of invertibility. See \cref{app:results_iv}.}} On the other hand, the SVAR-IV procedure has fewer parameters to estimate (as it excludes the IV $z_t$ from the reduced-form VAR regression), causing a reduction in dispersion relative to the other procedures. Though we view the high median bias of SVAR-IV across our DGPs as worrying, its low dispersion is intriguing and may in some cases trump the bias concerns.

\subsection{Robustness}
\label{sec:results_robustness}

This section argues that our main conclusions in \cref{sec:results_1,sec:results_2,sec:results_3,sec:results_4} are robust to several alterations of our baseline simulation specification. We pay particular attention to an exercise that replaces our non-stationary encompassing DFM with a stationary version. Various other robustness checks are listed subsequently, with details relegated to \cref{app:results}.

\paragraph{Stationary DGPs.}
While the majority of applied papers estimate VARs and LPs with a mix of non-stationary and stationary variables in levels \citep[e.g.,][]{Ramey2016}, in some cases researchers transform all their data to stationarity prior to the analysis. To cover such applications, we have repeated our analysis using the stationary estimated DFM of \citet{Stock2016} as our encompassing model. We construct impulse response estimands as before and compare the performance of the same estimation methods, except that the BVAR estimator uses a prior that shrinks towards white noise rather than random walks. Details on the implementation and results are presented in \cref{app:results_stationary}.

Our headline qualitative conclusions go through in the stationary DGPs. We observe the same bias-variance trade-off as in our main analysis, with LPs achieving lower bias than VARs at the cost of elevated variance. As a result, except for researchers that exclusively prioritize bias, least-squares VARs or some kind of shrinkage---in the form of Bayesian VARs or penalized LPs---are preferred. The two most notable differences from our baseline analysis are that (i) penalized LP outperforms BVAR for MSE loss at very short horizons, and (ii) due to the moderate persistence of the stationary DGPs, bias correction has less bite, and uncorrected least-squares LP has near-zero bias at all horizons.

\paragraph{Other robustness checks.}
The following modifications to our baseline simulation specification all leave our main conclusions qualitatively unchanged.

\begin{itemize}

\item \textbf{Recursive identification:} To complement the earlier results with observed shocks and proxy identification, we also consider recursive (Cholesky) identification schemes. We sidestep the controversial issue of whether recursive identification is an economically valid identification strategy by taking as the parameter of interest the shared large-sample limit of the recursive LP/VAR estimators (as the lag length tends to infinity). Details on the definition and empirical implementation are provided in \cref{app:recursive}. Simulation results for recursively identified shocks are similar to those for observed shock identification when the weight $\omega$ on (squared) bias in the loss function exceeds 0.8. However, when $\omega \leq 0.8$, BVAR is more attractive than in our baseline analysis. This is because recursive (i.e., Cholesky) identification relies heavily on estimation of the reduced-form innovation variance-covariance matrix. Uniquely among the estimation procedures we consider, BVAR imposes useful shrinkage on this matrix through the prior. See  \cref{app:results_recursive}.

\item \textbf{Salient observables:} Our results remain essentially unchanged if we restrict attention to a subset of 17 oft-used, salient macroeconomic time series out of the 207 ones in the full \citet{Stock2016} data set. We consider the exhaustive list of \emph{all} 1,581 five-variable DGPs that can be formed from these 17 series, subject to the selection rules in \cref{sec:dgp_implementation}. See \cref{app:results_salientobs}.

\item \textbf{Near-worst-case performance:} Whereas our baseline results pertain to the median performance of estimators across DGPs, some researchers may instead prefer to focus on ensuring acceptable performance for particularly challenging DGPs. To this end, \cref{app:results_losspcntl} reports the 90th percentiles of the bias and standard deviation across DGPs. Interestingly, adopting this ``near-worst-case'' perspective does not alter much the \emph{relative} magnitudes of bias and standard deviation across estimation procedures. Hence, none of the estimation procedures seem to have a particular advantage in ensuring robustness to challenging environments, over and above their performance in typical DGPs.

\item \textbf{Monetary vs. fiscal shocks:} If we consider the monetary shock DGPs separately from the fiscal shock DGPs, then the bias-variance trade-off is almost identical to that when we consider the DGPs jointly. See \cref{app:results_fiscal_monetary}.

\item \textbf{\alert{Larger} lag length:} If the lag length $p$ is set to 8 instead of 4, then LP and VAR are approximately equivalent out to horizon 8, as predicted by asymptotic theory. BVAR is relatively more attractive than in the $p=4$ case, as the prior reduces the effective dimensionality of the otherwise high-dimensional VAR system. Beyond that our conclusions on the overall nature of the bias-variance trade-off are unaffected. See \cref{app:results_lag}.

\item \textbf{\alert{Smaller} sample size:} Halving the sample size to $T=100$ quarters tends to increase the estimator standard deviations more than the biases, so shrinkage techniques look even more desirable than in our baseline, including in particular BVAR. Conversely, for bias-corrected LP to be optimal, bias needs to be prioritized even more heavily. See \cref{app:results_smallsample}.

\item \alert{\textbf{Larger sample size and lag length:} We set sample size $T=720$ and lag length $p=12$, a configuration reminiscent of monthly data. However, we caution that the set-up does not faithfully represent actual monthly data sets, since our DFM parameters remain fixed at the quarterly calibration described in \cref{sec:dgp}. As expected, least-squares LP and VAR have approximately equivalent properties out to horizon 12, while the trade-off between estimators at longer horizons is qualitatively similar to our baseline. At horizons below 12, shrinkage via BVAR or penalized LP is even more attractive than in our baseline, unless the bias weight in the loss function is high. See \cref{app:results_large}.}

\item \alert{\textbf{More observables:} If we increase the number of observed macro variables per DGP from 5 to 7, our conclusions are not affected. The only notable quantitative change is that, for IV identification, SVAR-IV has slightly smaller bias relative to the internal instruments procedures, due to the mechanical increase in the degree of invertibility. See \cref{app:results_more}.}

\item \textbf{Variable categories:} We find little evidence that the biases or standard deviations of individual impulse response estimators depend systematically on which categories of time series are included in the DGP (e.g., how many real activity or price series are used). See \cref{app:results_cat}.

\end{itemize}

\subsection{Discussion: can we select the estimator based on the data?}
\label{sec:results_selection}

It is natural to ask whether, instead of selecting estimators based on average performance across DGPs, the choice of estimator can be guided by the data at hand in each given DGP. We now show that this appears to be difficult, as conventional model selection or evaluation criteria are unable to detect even substantial mis-specification of the VAR(4) model in the vast majority of our DGPs. These findings are consistent with the previously documented poor performance of the VAR model averaging estimator. For simplicity, we focus here on observed shock identification.

First, the Akaike Information Criterion tends to select very short lag lengths $\hat{p}_{AIC}$ in our DGPs, as already mentioned earlier. The 90th percentile of $\hat{p}_{AIC}$ (across simulations) does not exceed 2 in any of our 6,000 DGPs, and it in fact equals 2 in only 68.3\% of those DGPs. This frequently used model selection tool therefore essentially never indicates that the VAR(4) specification is mis-specified.

Second, the Lagrange Multiplier test of residual serial correlation has low power in most of our DGPs. We carry out this test by regressing the sample VAR residuals on their first lags, controlling for four lags of the observed variables, and employing the likelihood ratio test defined in \citet{Johansen1995}. Using a 10\% significance level for the test, only around 8\% of the DGPs exhibit a rejection probability above 25\%, and none of the DGPs have a rejection probability above 50\%. Hence, this conventional specification test of the VAR(4) model is under-powered, despite the fact that many of our DGPs are in fact not well approximated by a VAR(4) model in population, as shown in \cref{sec:dgp_summ_stat}.

It is of course possible that other model selection criteria or specification tests will work better. However, at a minimum, the performance of the VAR model averaging estimator discussed in \cref{sec:results_3} and the evidence presented in this subsection together suggest that it is not straightforward to develop effective data-dependent estimator selection rules for use on conventional macroeconomic time series data.

\section{Conclusion and directions for future research}
\label{sec:conclusion}

We conducted a large-scale simulation study of the performance of LP and VAR structural impulse response estimators, as well as several variants of these methods. We drew the following four main conclusions.

\begin{enumerate}[1.]

\item As predicted by theory, there is a non-trivial bias-variance trade-off between least-squares LP and VAR estimators (after bias correction). Empirically relevant DGPs are unlikely to admit exact finite-order VAR representations, and so mis-specification of VAR estimators is indeed a valid concern, as discussed by \citet{Ramey2016} and \citet{Nakamura2018}, among others. Nevertheless, the slope of the trade-off is steep, with the lower bias of LP coming at the cost of substantially higher variance.

\item Bias-corrected LP is the preferred estimator if \emph{and only if} the researcher overwhelmingly prioritizes minimizing bias, with little regard to precision. Researchers who use LP should acknowledge their focus on bias, and they should apply the \citet{Herbst2021} bias correction procedure when the data are persistent.

\item For researchers that attach at least moderate weight to variance in their loss function (such as under the conventional MSE criterion), VAR methods are attractive. Specifically, Bayesian VARs perform well at short horizons, least-squares VARs at intermediate horizons, and the two methods are comparable at long horizons. The fact that no single VAR method dominates at all horizons means that researchers must take a stand not only on their preferences for bias and variance, but also on their primary horizons of interest, or alternatively ensure that their findings are supported by multiple procedures.

\item In the case of IV identification, the popular SVAR-IV (or proxy-SVAR) procedure can be severely biased, but it has substantially lower dispersion at all horizons than ``internal instruments'' procedures such as LP-IV or internal-IV VARs. The high (median) bias of SVAR-IV is due to its lack of robustness to non-invertibility, which is a pervasive and realistic feature of our DGPs.

\end{enumerate}

These conclusions inevitably depend on the choice of encompassing model and the specific implementation of the impulse response estimators. Our paper first and foremost has aimed to bring the bias-variance trade-off in impulse response estimation to the attention of applied researchers. Our particular quantification of this trade-off has sought to capture the wide range of applied settings faced by macroeconomists, by fitting a dynamic factor model with rich short-run and long-run dynamics to the well-known \citet{Stock2016} data set. Our online code repository (see \cref{fn:github}) facilitates experimentation with alternative encompassing models or estimation procedures.

Our findings point to several potential areas for future research. First, we conjecture that the bias-variance trade-off may differ quantitatively in panel data settings, to the extent that the availability of a large cross section reduces the sampling variance of the estimators for a given time dimension, thus potentially making LP relatively more attractive than in the pure time series case. Second, our analysis has focused on the average performance of estimators across DGPs because we find that conventional model selection or evaluation tools are unable to detect substantial mis-specification of low-order VARs in our simulations; nevertheless, we view data-dependent estimator selection as an area ripe for further investigation. Third, it may be worth investigating whether the performance of the Bayesian VAR procedure at intermediate horizons can be improved by developing alternative prior specifications that are specifically aimed at structural impulse response estimation rather than forecasting, unlike the priors used in much of the literature. Fourth, for the case of IV/proxy identification, an interesting question is whether it is possible to develop alternative invertibility-robust estimation procedures that capture some of the variance improvement enjoyed by the non-robust SVAR-IV estimator. Fifth, we leave exploration of other structural shock identification schemes---such as sign restrictions, long-run restrictions, and non-recursive short-run restrictions---to future work. Sixth, while our simulations were calibrated to quarterly data, it would be illuminating to see whether our conclusions apply also to \alert{monthly calibrations}.

\clearpage

\appendix
\begin{appendices}
\numberwithin{equation}{section}
\numberwithin{figure}{section}
\numberwithin{lem}{section}
\numberwithin{prop}{section}
\numberwithin{cor}{section}

\section{Details on DGP definitions}
\label{app:dgp}

\subsection{Shock definition}
\label{app:subsec:dgp_select}

Our definition of the structural shock of interest, $\varepsilon_{1,t}$, ensures that it has the largest possible contemporaneous effect on nominal interest rates (for monetary shocks) and government spending (for fiscal shocks). Letting $\eta_t \equiv H \varepsilon_t$, $\Sigma_\eta \equiv \var(\eta_t)$, and $\iota^*$ denote the index of the policy instrument $i_t$ in the vector $X_t$, the shock is thus defined through the solution of the following problem:
\begin{equation*}
\max_{H} \quad \Lambda_{\iota^*, \bullet} H e_1 \quad \text{s.t.} \quad H H' = \Sigma_\eta,
\end{equation*}
where $e_1$ selects the first column of $H$. The solution equals $H_{\bullet, 1} = \Sigma_\eta \Lambda_{\iota^*, \bullet}' (\Lambda_{\iota^*, \bullet} \Sigma_\eta \Lambda_{\iota^*, \bullet}')^{-1/2}$.\footnote{The remaining columns in $H$ are chosen arbitrarily to satisfy the variance-covariance constraint; these columns only matter for the simulation results through the implications for reduced-form dynamics.}

\subsection{IV process calibration}
\label{app:subsec:iv_calibration}

We calibrate the innovation noise $\sigma_\nu^2$ in the IV equation to target population IV first-stage F-statistics between 10 and 30 when $T=200$, consistent with borderline weak to moderately strong identification, as in the majority of applied work. This yields $\sigma_\nu \in \{ 1.1, 1.5, 2.3 \}$. We draw $\rho_z$ and $\sigma_\nu$ uniformly at random from their two sets.

\section{Details on estimation procedures}
\label{app:estimators}

\paragraph{Least-squares LP.}
The least-squares LP estimator of the impulse response at horizon $h$ is based on the coefficient $\hat{\beta}_h$ in the $h$-step-ahead OLS regression
\begin{equation} \label{eq:estim_lp}
y_{t+h} = \hat{\mu}_h + \hat{\beta}_h x_t + \hat{\zeta}_h q_t + \sum_{\ell=1}^p \hat{\varphi}_{h,\ell} w_{t-\ell} + \text{residual}_{t,h},
\end{equation}
that is, we regress on the variable $x_t$, with controls given by the vector $q_t$ as well as $p$ lags of all of the data $w_t$. The estimands of \cref{sec:dgp_estimand} can now be estimated as follows:

\begin{enumerate}

\item {\bf Observed shock.} We set $x_t$ equal to the observed shock $\varepsilon_{1,t}$ and omit the contemporaneous controls $q_t$ (we still control for lagged data).\footnote{The lags are not needed for consistency in this case, but they often improve efficiency.}

\item {\bf IV.} We estimate a Two-Stage Least Squares (2SLS) version of \eqref{eq:estim_lp}, setting $x_t$ equal to the policy variable $i_t$, and instrumenting for this variable with the IV $z_t$. We omit $q_t$ in this specification (but still include lagged controls). This is numerically the same as doing a LP of $y_{t+h}$ on $z_t$ (with lagged controls), and dividing this coefficient by the LP coefficient in a regression of $i_t$ on $z_t$ (with lagged controls), see \citet{Stock2018} and \citet{Plagborg2020}.

\item {\bf Recursive identification.} $x_t = i_t$ is the policy variable, while $q_t$ are the variables ordered before $i_t$ in the identification scheme \citep{Plagborg2020}.

\end{enumerate}

\paragraph{Bias-corrected LP.} We implement the bias-corrected LP estimator of \citet{Herbst2021}, using their approximate analytical bias formula for LP with controls and with population autocovariances substituted with sample analogues.\footnote{\citeauthor{Herbst2021}'s analytical derivations assume stationarity, but we will apply the formula regardless. This is similar to how analytical bias correction is typically carried out in VAR contexts \citep{Kilian1998}, as \citet{Pope1990} also assumes stationarity.} Following their recommendation, we implement an iterative bias correction, where the impulse response estimate at horizon $h$ is bias-corrected using the previously corrected impulse response estimates at horizons $1,2,\dots,h-1$.

\paragraph{Penalized LP.}
The \citet{Barnichon2019} estimator lowers the variance of LP by exploiting a prior belief in smoothness of the impulse response function across horizons. Following their preferred implementation, we model the impulse response function using B-spline basis functions. The jaggedness penalty function penalizes deviations from a quadratic function of the horizon $h$. We penalize impulse responses up to horizon 20. The penalty parameter is selected in a data-dependent way using 5-fold cross-validation. We do not penalize the coefficients on the control variables in the LP. When reporting relative impulse responses \eqref{eq:IRF_IV}, we divide by the \emph{least-squares} LP estimate of the impact response of the policy variable $i_t$ to the structural shock.

\paragraph{Least-squares VAR.}
The least-squares VAR coefficient estimates are obtained through equation-by-equation OLS regressions. We perform a Cholesky decomposition of the forecast error variance-covariance matrix and compute impulse response functions with respect to the orthogonalized shocks. The estimands of \cref{sec:dgp_estimand} can now be estimated as follows:

\begin{enumerate}

\item {\bf Observed shock.} The shock $\varepsilon_{1,t}$ is ordered first in $w_t$, and we compute responses to the first innovation.

\item {\bf IV.} We initially consider an ``internal instruments'' approach as in \citet{Ramey2011}. That is, we include the IV $z_t$ in the data vector $w_t$, order the IV first, and compute responses with respect to the first innovation \citep{Plagborg2020}. The relative impulse response \eqref{eq:IRF_IV} is obtained by dividing by the impact response of the policy variable $i_t$.

\item {\bf Recursive identification.} The ordering of variables in $w_t$ equals the ordering of the desired population impulse response estimand (cf. \cref{sec:dgp_estimand}), and we compute responses to the innovation of the policy instrument $i_t$.

\end{enumerate}

In contrast to the above internal instruments approach, the SVAR-IV (or ``proxy-SVAR'') estimator of \citet{Stock2008} is obtained by computing the reduced-form impulse responses $\hat{\Psi}_h$ ($h=0,1,\dots$) corresponding to a VAR in $\bar{w}_t$ (i.e., excluding $z_t$), and then reporting relative impulse responses \eqref{eq:IRF_IV} corresponding to the absolute structural impulse responses $\hat{\Psi}_h\hat{\gamma}$, where $\hat{\gamma}$ is the sample covariance vector of the reduced-form VAR residuals $\hat{u}_t$ and the IV $z_t$.

\paragraph{Bias-corrected VAR.}
We follow \citet{Kilian1998} and consider a modification of the standard least-squares VAR estimator that applies the \citet{Pope1990} analytical bias correction to the reduced-form VAR coefficient matrices. We use \citeauthor{Kilian1998}'s procedure for ensuring the largest eigenvalue of the bias-corrected VAR companion matrix does not exceed 1.

\paragraph{Bayesian VAR.}
Our BVAR implementation follows the default prior recommendations of \citet{Giannone2015}, as implemented in their replication code. The prior is therefore a Minnesota prior, extended with the ``sum-of-coefficients'' and ``dummy-initial-observation'' priors to improve long-run forecasts. The degrees of shrinkage provided by each of the three prior components are governed by three prior hyper-parameters, which are selected by maximizing the marginal likelihood.\footnote{Note that we do not use \citeauthor{Giannone2015}'s computationally intensive hierarchical Bayesian procedure but instead select hyper-parameters to maximize the marginal likelihood. In doing this, we substitute their custom optimization routine with the built-in Matlab function {\tt fminunc}.} To save on computation time, we do not optimize the hyper-parameter-vector $\psi$ (in their notation), i.e., the diagonal of the scale matrix in the Wishart prior on the innovation variance matrix; instead, these hyper-parameters are fixed at the residual variance estimates from preliminary AR(1) regressions.

\paragraph{VAR model averaging.}
\citet{Hansen2016} proposes a data-dependent procedure for averaging across impulse responses estimates produced by a collection of different AR and VAR models with different lag lengths. Let $\hat{\delta}_h(r)$ denote the un-normalized, least-squares recursive impulse response estimate at some horizon $h$ for model $r=1,\dots,R$. We estimate $\hat{\delta}_h(r)$ from $R=40$ candidate models: first, univariate AR models for $y_t$ with lag lengths from $p=1$ up to $p=20$; and second, VAR models in $w_t$ with lag lengths from $p=1$ up to $p=20$. As in \citet{Hansen2016}, the variance-covariance matrix of innovations $\Sigma$ and thus the impact effect $\delta_0$ are fixed across candidate models and treated as known without error.\footnote{To match the impact effect estimate in our benchmark method of least-squares VAR, we use $\hat{\Sigma}$ from the $p=4$ VAR estimate as the true value across all the candidate models.} The VAR model averaging estimator is given by $\sum_{r=1}^R \hat{\omega}_r \hat{\delta}_h(r)$, where the weights $\lbrace \hat{\omega}_r \rbrace_{r=1}^R$ are chosen to minimize the data-dependent approximated MSE estimate $\hat{M}(\omega_1,\dots,\omega_R) \approx E[T(\sum_{r=1}^R \omega_r \hat{\delta}_h(r)-\delta_h)^2]$, subject to the constraints that all weights are nonnegative and $\sum_{r=1}^R \omega_r=1$. Details of the MSE estimate are given in \citet[Section 6]{Hansen2016}.\footnote{The object of interest, $\hat{\delta}_h(r)$, is a scalar, which allows us to omit the weighting matrix required in the MSE estimate in \citet{Hansen2016}.} We run this optimization for the weights separately at each impulse response horizon. Relative impulse responses \eqref{eq:IRF_IV} are computed by dividing the absolute impulse response by the least-squares VAR(4) impact impulse response estimate of $i_t$ with respect to the identified shock.

\phantomsection
\addcontentsline{toc}{section}{References}
\small
\bibliography{lp_var_simul_ref}

\end{appendices}

\end{document}

\title{Online Appendix for ``Local Projections vs. VARs: Lessons From Thousands of DGPs''}
\author{Dake Li \and Mikkel Plagborg-M{\o}ller \and Christian K. Wolf}
\date{\today}
\maketitle

\vspace{-\baselineskip}

\tableofcontents

\clearpage

\begin{appendices}
\crefalias{section}{sappsec}
\crefalias{subsection}{sappsubsec}
\crefalias{subsubsection}{sappsubsubsec}
\setcounter{section}{2}

\section{Estimation of non-stationary DFM}
\label{app:dfm_estim}

We here elaborate on our estimation of the non-stationary DFM, complementing the brief discussion in \cref{sec:dgp_implementation}.

\paragraph{Data transformation.}
We use the same dataset as \citet{Stock2016} but, unlike those authors, do not transform the series to stationarity prior to estimating the DFM. This means that we change \citeauthor{Stock2016}'s transformation codes 2 (first differences), 3 (second differences), 5 (first log differences), and 6 (second log differences) to 1 (levels), 2, 4 (log levels), and 5, respectively. We then correct outliers using the same procedure as \citeauthor{Stock2016}, meaning that when the outlier adjustment procedure is applied to a series in levels, we difference the series, then adjust outliers, and then finally cumulate the series again. There are 8 series who have outliers adjusted, all of which are in levels. Finally, and unlike \citeauthor{Stock2016}, we refrain from subtracting a nonparametric trend estimate from the series. Instead, we control for a linear time trend in the estimation, as discussed below.

\paragraph{Estimation.}
We set the number of factors $n_f$ equal to 6 as in \citet{Stock2016}, since these authors selected the number of factors based on information criteria applied to $\Delta X_t$ (essentially), and this also remains a valid way to select $n_f$ in our non-stationary model. We then apply \citeauthor{Stock2016}'s PCA procedure for unbalanced panels to the differenced data $\Delta X_t$.\footnote{Like \citeauthor{Stock2016}, we estimate the factors using only a subset of the variables in $X_t$, since some variables are essentially aggregates of other variables.} \alert{The differenced data is de-meaned and standardized prior to performing PCA, thus removing any series-specific deterministic linear time trends.} This gives us estimates $\Delta \hat{f}_t$ of the differenced factors (up to rotation), and we then cumulate $\hat{f}_t = \sum_{s=1}^t \Delta \hat{f}_s$. Asymptotically, $\hat{f}_t$ equals $f_t$ up to rotation and a linear time trend \citep{Barigozzi2021}.\footnote{\alert{Since \citet{Stock2016} also estimate factors from essentially the same differenced data, we refer to their interpretation of the factors (p. 488): ``[T]he first factor explains large fractions of the variation in the growth of GDP and employment, but only small fractions of the variation in prices and financial variables. The second through fourth factors explain the variation in headline inflation, oil prices, housing starts, and some financial variables.''}} We then estimate the loadings $\Lambda=(\lambda_1,\dots,\lambda_{n_X})'$ through series-by-series OLS regressions of $X_{i,t}$ onto $\hat{f}_t$, controlling for a linear time trend. The OLS residuals $\hat{v}_{i,t}$ are estimates of the idiosyncratic components $v_{i,t}$.

As in \citet{Stock2016}, we next fit AR($p_v$) processes to each idiosyncratic residual $\hat{v}_{i,t}$ by OLS, separately for each $i$. Unlike \citeauthor{Stock2016}, we apply the \citet{Pope1990} bias correction to the AR coefficient estimates to avoid understating persistence.

Differently from the stationary specification in \citet{Stock2016}, we fit a VECM to the 6 estimated factors $\hat{f}_t$, explicitly allowing for unit roots and cointegration. We specify that the VECM in error correction form contains $p_f-1$ lagged difference terms, implying a VAR($p_f$) model for the factors, and control for a time trend. \alert{That is, the fitted VECM is of the form
\[\Delta \hat{f}_t = \hat{\alpha}\hat{\beta}'\hat{f}_{t-1} + \sum_{\ell=1}^{p_f-1} \hat{B}_\ell \Delta \hat{f}_{t-\ell} + \hat{\nu}_0 + \hat{\nu}_1t +  \hat{\eta}_t,\]
where $\hat{\alpha}$ and $\hat{\beta}$ are $6 \times r$ matrices, with $r$ denoting the cointegration rank, $\hat{B}_1,\dots,\hat{B}_{p_f-1}$ are $6 \times 6$ matrices, and $\hat{\nu}_0$ and $\hat{\nu}_1$ are $6 \times 1$ vectors. All estimated parameters (including the residual variance-covariance matrix) are unrestricted, with the exception of a conventional normalization on the first $r$ columns of $\hat{\beta}$.}\footnote{\alert{We use the command {\tt jcitest} in Matlab's Econometrics Toolbox to test and estimate the VECM. We refer to the command's documentation for implementation details.}} We estimate the cointegration rank \alert{$r$} using the \citet{Johansen1995} maximum eigenvalue test, again controlling for an unrestricted time trend. We apply a sequential testing procedure using a 5\% significance level. Given the final non-rejected cointegration rank, \alert{we estimate the VECM parameters by quasi-MLE \citep{Johansen1995} and then transform these parameters into VAR parameters \citep[Equation 3.1.4]{Kilian2017}}. Note that, while the empirical estimation has controlled for deterministic linear time trends \alert{in the individual series and in the latent factor process (essentially equivalent with linearly detrending the data prior to analysis)}, we omit any such deterministic terms from the final calibrated DFM \eqref{eq:factors}--\eqref{eq:idio_errors} used in our simulation study.

To ensure comparability with \citet{Stock2016}, the principal components routine uses data starting in 1959Q1, while the loadings, idiosyncratic components, and factor VECM are estimated on data starting in 1959Q3.

\paragraph{Further estimation details and results.}
We select the factor and idiosyncratic error lag lengths using the Akaike information criterion (AIC). For the factors, AIC is minimized at $p_f=3$, but the criterion is essentially flat for $p_f \in [2,4]$; thus, to err on the side of allowing for richer long-run dynamics, we set $p_f = 4$. For the idiosyncratic errors, we apply the AIC to each individual series. The 90th percentile of selected lags equals $4$; we thus set $p_v = 4$ to be consistent with the clear majority of the series. All estimated idiosyncratic AR(4) processes are technically stationary, but half the processes have largest AR root exceeding 0.86, with 25\% exceeding 0.93.

The Johansen test selects a cointegration rank of 2 for the factor VECM, corresponding to $6-2=4$ common stochastic trends.

\clearpage

\section{Definition of recursive shock estimand}
\label{app:recursive}

Here we define the impulse response estimand for the recursive identification scheme.

The researcher observes only the endogenous variables $\bar{w}_t \subset X_t$, with no further direct or noisy measures of structural shocks. Thus, the total vector of observables is $w_t = \bar{w}_t$. Consistent with a large literature on recursive shock identification in VARs \citep[e.g.,][]{Christiano1999}, we take as the estimand the impulse responses with respect to a recursive (Cholesky) orthogonalization of the reduced-form (Wold) forecast errors in the VAR($\infty$) process for $\bar{w}_t$ implied by the DFM. The shock of interest is the orthogonalized innovation to a policy variable $i_t$ in $w_t$. For monetary policy DGPs, we order the federal funds rate last, as in \citet{Christiano1999}; this restricts the other included variables to not respond contemporaneously to the monetary innovation. For fiscal policy DGPs, we order the government expenditure series first; this restricts the fiscal authority to respond to other innovations in the recursive VAR with a lag, as in \citet{Blanchard2002}.

Note that the recursively orthogonalized innovation differs across DGPs, and it generally does not equal any of the structural shocks $\varepsilon_{j,t}$ in the DFM. We nevertheless consider this impulse response estimand due to its popularity in applied work.

We now provide the mathematical definition of the estimand. Recall that the encompassing DFM takes the form \eqref{eq:factors}--\eqref{eq:idio_errors}. We map this DFM into the ``ABCD'' form of \citet{Fernandez2007} as follows. The general ABCD representation takes the form
\begin{eqnarray}
s_t &=& A s_{t-1} + B e_t, \label{eq:state_ABCD} \\
y_t &=& C s_{t-1} + D e_t, \label{eq:measurement_ABCD}
\end{eqnarray}
where $e_t \sim \mathcal{N}(0,I)$. Define the notation $\Phi_{1:p} = (\Phi_1,\dots,\Phi_p)$, $\Gamma_\ell = \diag(\delta_{1,\ell},\dots,\delta_{n_X,\ell})$, $\Gamma_{1:p} = (\Gamma_1,\dots,\Gamma_p)$, and $\Xi = \diag(\Xi_1,\dots,\Xi_{n_X})$. Then the mapping from \eqref{eq:factors}--\eqref{eq:idio_errors} to \eqref{eq:state_ABCD}--\eqref{eq:measurement_ABCD} given a selected set of observables $\bar{w}_t=\bar{S}X_t$ is
\begin{eqnarray}
\underbrace{\begin{pmatrix} f_{t} \\ \vdots \\ f_{t-p_f+1} \\ v_t \\ \vdots \\ v_{t-p_v+1} \end{pmatrix}}_{s_t} &=& \underbrace{\begin{pmatrix} \Phi_{1:p_f-1} & \Phi_{p_f} & 0 & 0 \\ I & 0 & 0 & 0  \\ 0 & 0 & \Gamma_{1:p_v-1} & \Gamma_{p_v} \\ 0 & 0 & I & 0 \end{pmatrix}}_{A} \begin{pmatrix} f_{t-1} \\ \vdots \\ f_{t-p_f} \\ v_{t-1} \\ \vdots \\ v_{t-p_v} \end{pmatrix} + \underbrace{\begin{pmatrix} H & 0 \\ 0 & 0 \\ 0 & \Xi \\ 0 & 0 \end{pmatrix}}_{B} \underbrace{\begin{pmatrix} \varepsilon_{t} \\ \xi_t \end{pmatrix}}_{e_t}, \label{eq:state_ABCD_DFM} \\
\bar{w}_t &=& \underbrace{\bar{S} \begin{pmatrix} \Lambda\Phi_{1:p_f} & \Gamma_{1:p_v}  \end{pmatrix}}_{C} \begin{pmatrix} f_{t-1} \\ \vdots \\ f_{t-p_f} \\ v_{t-1} \\ \vdots \\ v_{t-p_v} \end{pmatrix} + \underbrace{\bar{S} \begin{pmatrix} \Lambda H & \Xi \end{pmatrix}}_{D} \begin{pmatrix} \varepsilon_{t} \\ \xi_t \end{pmatrix}. \label{eq:measurement_ABCD_DFM}
\end{eqnarray}

Next we derive the ``innovations representation'' of the ABCD model, proceeding exactly as in \citet{Fernandez2007}:\footnote{To ensure that the matrix $\Sigma=\var(s_t \mid \bar{w}_t,\bar{w}_{t-1},\dots)$ is positive semidefinite despite numerical rounding errors, we write Equation 9 in \citet{Fernandez2007} in the following equivalent way:
\[\Sigma = L\tilde{\Sigma}L', \quad \text{where } L \equiv \tilde{A}-\tilde{A}\tilde{\Sigma}\tilde{C}'(\tilde{C}\tilde{\Sigma}\tilde{C}')^{-1}\tilde{C},\; \tilde{\Sigma} \equiv \begin{pmatrix} \Sigma & 0 \\ 0 & I \end{pmatrix},\; \tilde{A}\equiv(A,B),\; \tilde{C}\equiv (C,D).\]}
\begin{eqnarray*}
	\hat{x}_t &=& A\hat{x}_{t-1} + K\bar{u}_t, \\
	\bar{w}_t &=& C\hat{x}_{t-1} + \bar{u}_t,
\end{eqnarray*}
where $\hat{x}_t = E[s_t \mid \bar{w}_t,\bar{w}_{t-1},\dots]$ and $\bar{u}_t = \bar{w}_t - E[\bar{w}_t \mid \bar{w}_{t-1},\bar{w}_{t-2},\dots]$. The innovations representation immediately yields the impulse responses of the observables $\bar{w_t}$ with respect to the Wold innovations $\bar{u}_t$. We orthogonalize the Wold innovations using a Cholesky decomposition, given the chosen ordering of the variables. In particular, letting $\var(\bar{u}_t) =\Sigma_{\bar{u}} = \bar{B} \bar{B}'$, with $\bar{B}$ lower triangular, and denoting the Wold innovation impulse responses by $\bar{C}(L)$, we define the recursive impulse response estimands as
\begin{equation}
\theta_h \equiv \frac{\bar{C}_{\iota_y, \bullet, h} \bar{B}_{\bullet, \iota_i}}{\bar{C}_{\iota_i, \bullet, 0} \bar{B}_{\bullet, \iota_i}}, \quad h=0,1,2,\dots, \label{eq:IRF_recursive}
\end{equation}
where $\iota_y$ and $\iota_i$ are the indices corresponding to $y_t$ and $i_t$ in the vector $\bar{w}_t$, respectively. Unlike the observed shock and IV estimands considered in \cref{sec:dgp_estimand}, the estimand \eqref{eq:IRF_recursive} might not equal the model-implied \emph{structural} impulse response of the variable $y_t$ with respect to any aggregate shock $\varepsilon_{j,t}$ in the DFM.\footnote{A necessary condition for the impulse responses \eqref{eq:IRF_recursive} to equal structural impulse responses from \eqref{eq:factors}--\eqref{eq:idio_errors} is that $\varepsilon_{j,t} \in \sspan( \{ \bar{w}_{t-\ell} \}_{\ell=0}^\infty)$ for at least one shock $j$. A sufficient condition for $\varepsilon_{t} \in \sspan( \{ \bar{w}_{t-\ell} \}_{\ell=0}^\infty)$ is that $n_{\bar{w}} = n_f$, $\bar{\Lambda}$ is non-singular, and $\Xi_i = 0$ for all $i$ in $\bar{w}_t$.} In other words, the expression \eqref{eq:IRF_recursive} is the impulse response with respect to a potentially non-structural innovation.

To interpret the estimand \eqref{eq:IRF_recursive}, consider two popular applied identification schemes. First, in the monetary policy shock identification scheme of \citet{Christiano1999}, $y_t$ may be aggregate output, $i_t$ is the nominal interest rate, and the nominal rate is typically ordered \emph{after} all other observables. The recursive estimand is then the impulse response of output to a residualized interest rate innovation, normalized by the impact response of interest rates. Second, for the fiscal policy shock identification procedure in \citet{Blanchard2002}, we may again take $y_t$ to be aggregate output and let $i_t$ be aggregate government spending. In this case, reduced-form innovations in the government spending equation are treated as structural shocks, and so \eqref{eq:IRF_recursive} gives impulse responses to those innovations, normalized by the impact response of government spending.

\clearpage

\section{Examples of estimated IRFs}
\label{app:results_irfs}

\cref{fig:supp:lp_irfs,fig:supp:var_irfs,fig:supp:bclp_irfs,fig:supp:bcvar_irfs,fig:supp:bvar_irfs,fig:supp:penlp_irfs,fig:supp:varavg_irfs} provide a visual illustration of estimated impulse response functions (IRFs) from the seven estimation procedures defined in \cref{sec:estim}. We fix a single (randomly chosen) DGP with an observed fiscal shock and simulate ten data sets with sample size $T=200$. We then apply the seven estimation methods to these ten data sets.

\begin{figure}[p]
\centering
\textsc{Observed fiscal shock: LP IRFs} \\[0.5\baselineskip]
\includegraphics[width=0.7\linewidth]{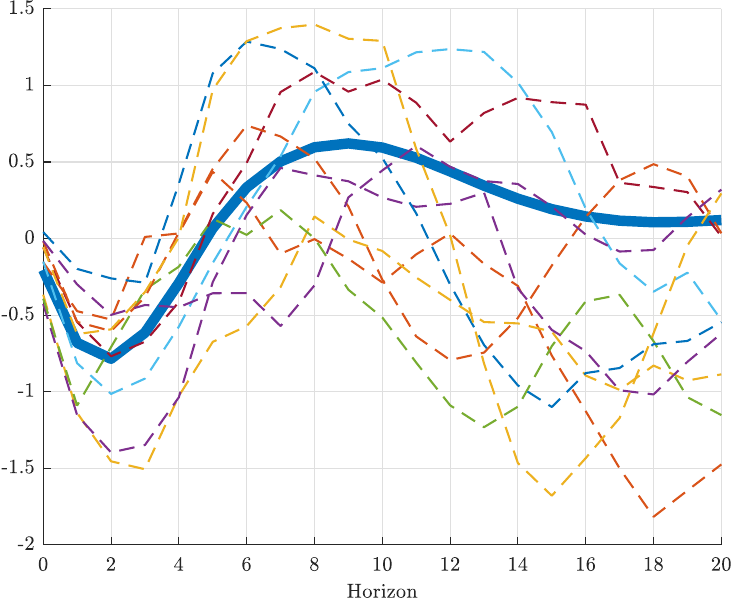}
\caption{Structural impulse response estimand (thick blue) for one specification with an observed fiscal spending shock vs. ten least-squares LP impulse response estimates.}
\label{fig:supp:lp_irfs}

\vspace*{\floatsep}

\centering
\textsc{Observed fiscal shock: Bias-corrected LP IRFs} \\[0.5\baselineskip]
\includegraphics[width=0.7\linewidth]{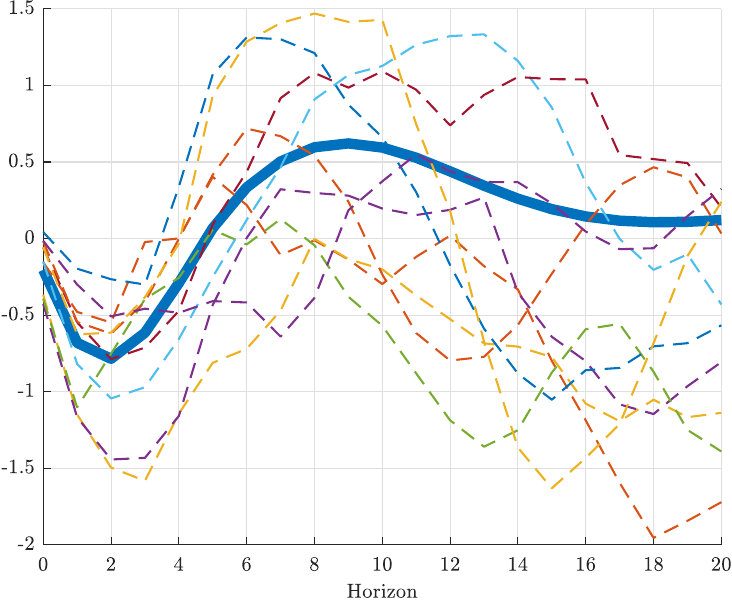}
\caption{Structural impulse response estimand (thick blue) for one specification with an observed fiscal spending shock vs. ten bias-corrected LP impulse response estimates.}
\label{fig:supp:bclp_irfs}
\end{figure}

\begin{figure}[p]
	\centering
	\textsc{Observed fiscal shock: Penalized LP IRFs} \\[0.5\baselineskip]
	\includegraphics[width=0.7\linewidth]{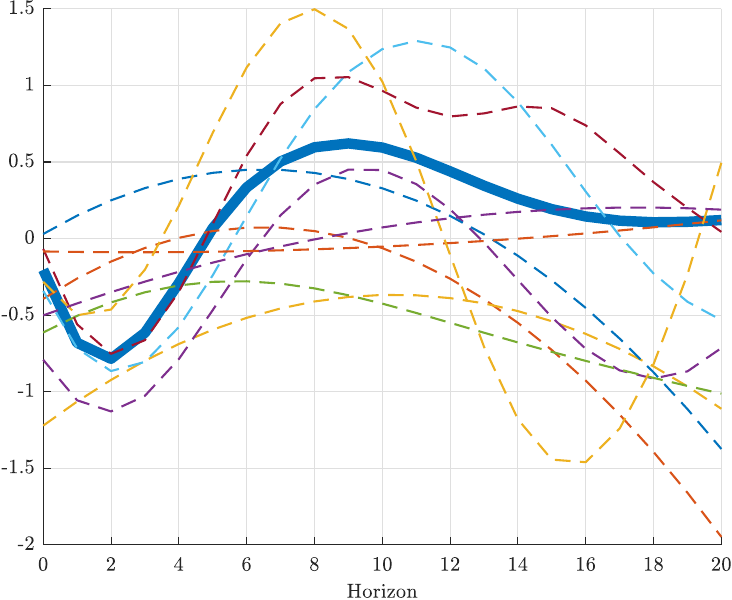}
	\caption{Structural impulse response estimand (thick blue) for one specification with an observed fiscal spending shock vs. ten penalized LP impulse response estimates.}
	\label{fig:supp:penlp_irfs}
\end{figure}

\begin{figure}[p]
\centering
\textsc{Observed fiscal shock: VAR IRFs} \\[0.5\baselineskip]
\includegraphics[width=0.7\linewidth]{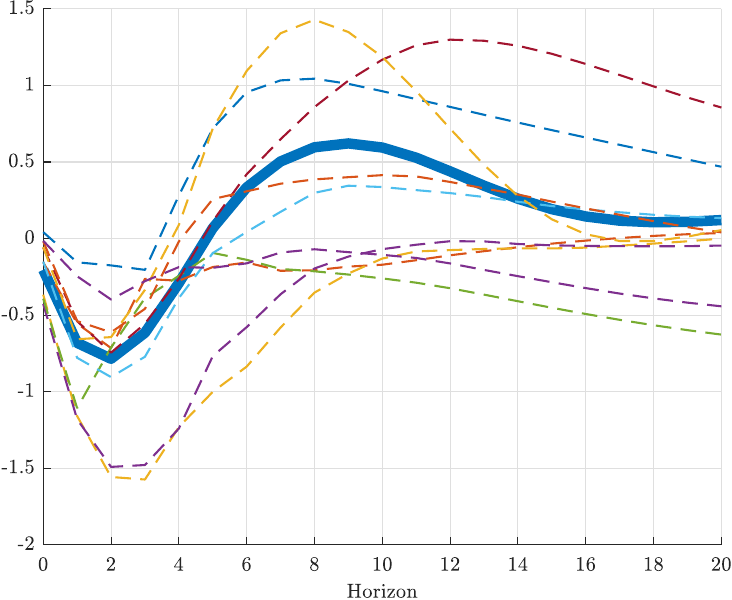}
\caption{Structural impulse response estimand (thick blue) for one specification with an observed fiscal spending shock vs. ten least-squares VAR impulse response estimates.}
\label{fig:supp:var_irfs}

\vspace*{\floatsep}

\centering
\textsc{Observed fiscal shock: Bias-corrected VAR IRFs} \\[0.5\baselineskip]
\includegraphics[width=0.7\linewidth]{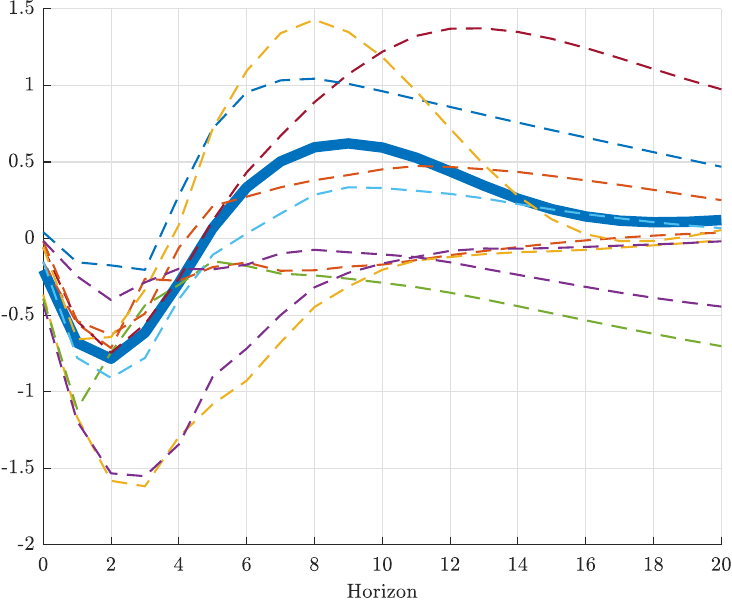}
\caption{Structural impulse response estimand (thick blue) for one specification with an observed fiscal spending shock vs. ten bias-corrected VAR impulse response estimates.}
\label{fig:supp:bcvar_irfs}
\end{figure}

\begin{figure}[p]
\centering
\textsc{Observed fiscal shock: BVAR IRFs} \\[0.5\baselineskip]
\includegraphics[width=0.7\linewidth]{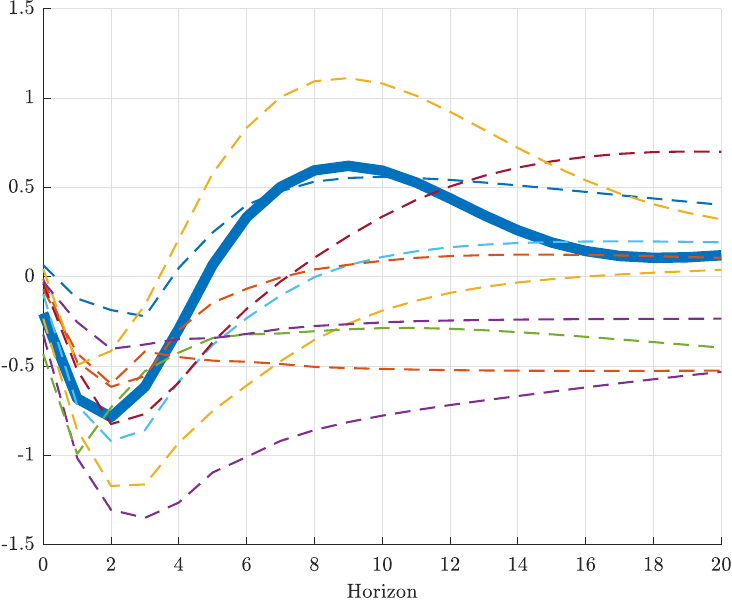}
\caption{Structural impulse response estimand (thick blue) for one specification with an observed fiscal spending shock vs. ten Bayesian VAR impulse response estimates.}
\label{fig:supp:bvar_irfs}

\vspace*{\floatsep}

\centering
\textsc{Observed fiscal shock: VAR model averaging IRFs} \\[0.5\baselineskip]
\includegraphics[width=0.7\linewidth]{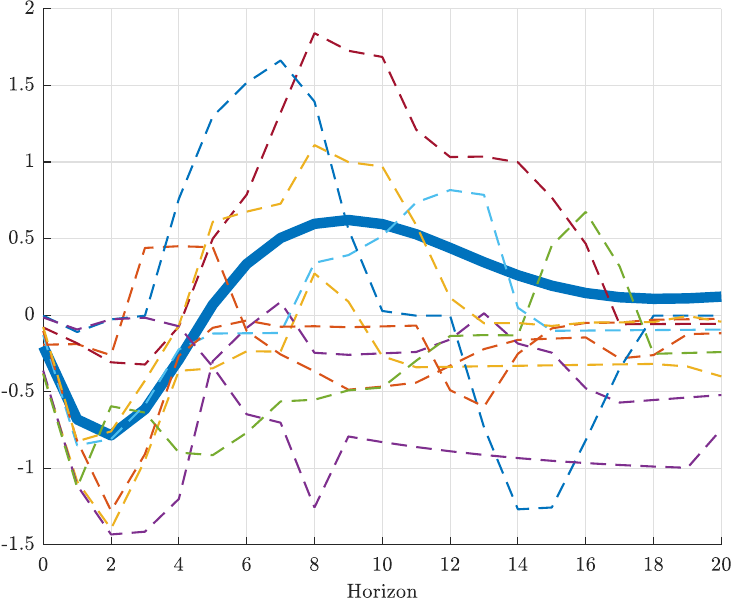}
\caption{Structural impulse response estimand (thick blue) for one specification with an observed fiscal spending shock vs. ten VAR model averaging impulse response estimates.}
\label{fig:supp:varavg_irfs}
\end{figure}

\clearpage

\section{Further simulation results and robustness}
\label{app:results}

\subsection{IV estimators}
\label{app:results_iv}

\cref{fig:supp:bias_iv_supp,fig:supp:std_iv_supp} plot the mean bias and standard deviation of the estimation procedures in the case of IV identification. The relative ranking of the various estimation procedures is essentially the same as in the median bias and interquartile range plots presented in our main analysis in \cref{sec:results_4}.

\alert{\cref{fig:supp:medbias_iv_lowinv_supp,fig:supp:medbias_iv_highinv_supp} show the median bias of our estimation procedures, but now for the 10 percent of DGPs with the smallest and largest degrees of invertibility, respectively. As expected, the median bias for SVAR-IV is particularly elevated relative to other estimation methods if the degree of invertibility is small, as predicted by theory.}

\begin{figure}[tp]
\centering
\textsc{IV: Mean bias of estimators} \\[0.5\baselineskip]
\includegraphics[width=0.85\linewidth]{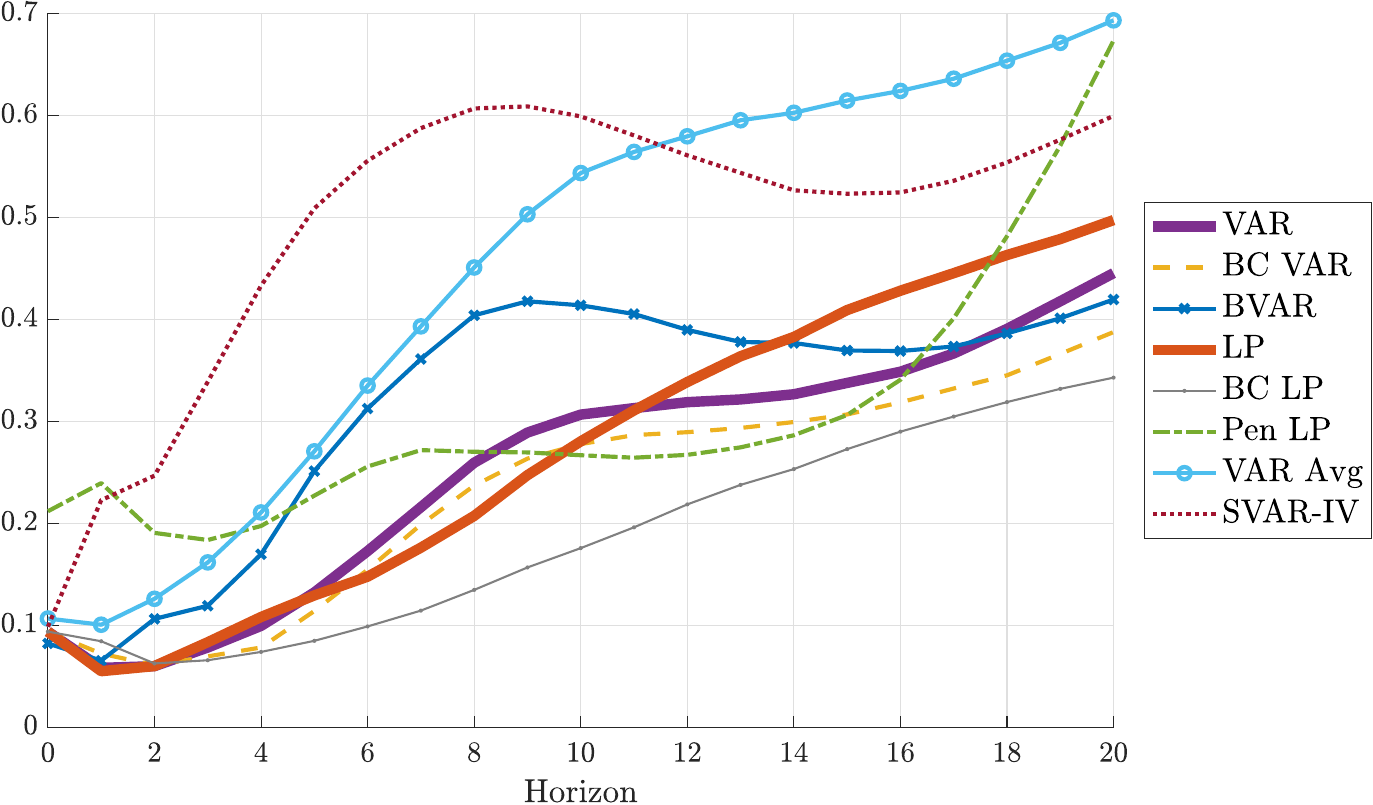}
\caption{Median (across DGPs) of absolute mean bias of the different estimation procedures, relative to $\sqrt{\frac{1}{21}\sum_{h=0}^{20}\theta_h^2}$.}
\label{fig:supp:bias_iv_supp}

\vspace*{\floatsep}

\centering
\textsc{IV: Standard deviation of estimators} \\[0.5\baselineskip]
\includegraphics[width=0.85\linewidth]{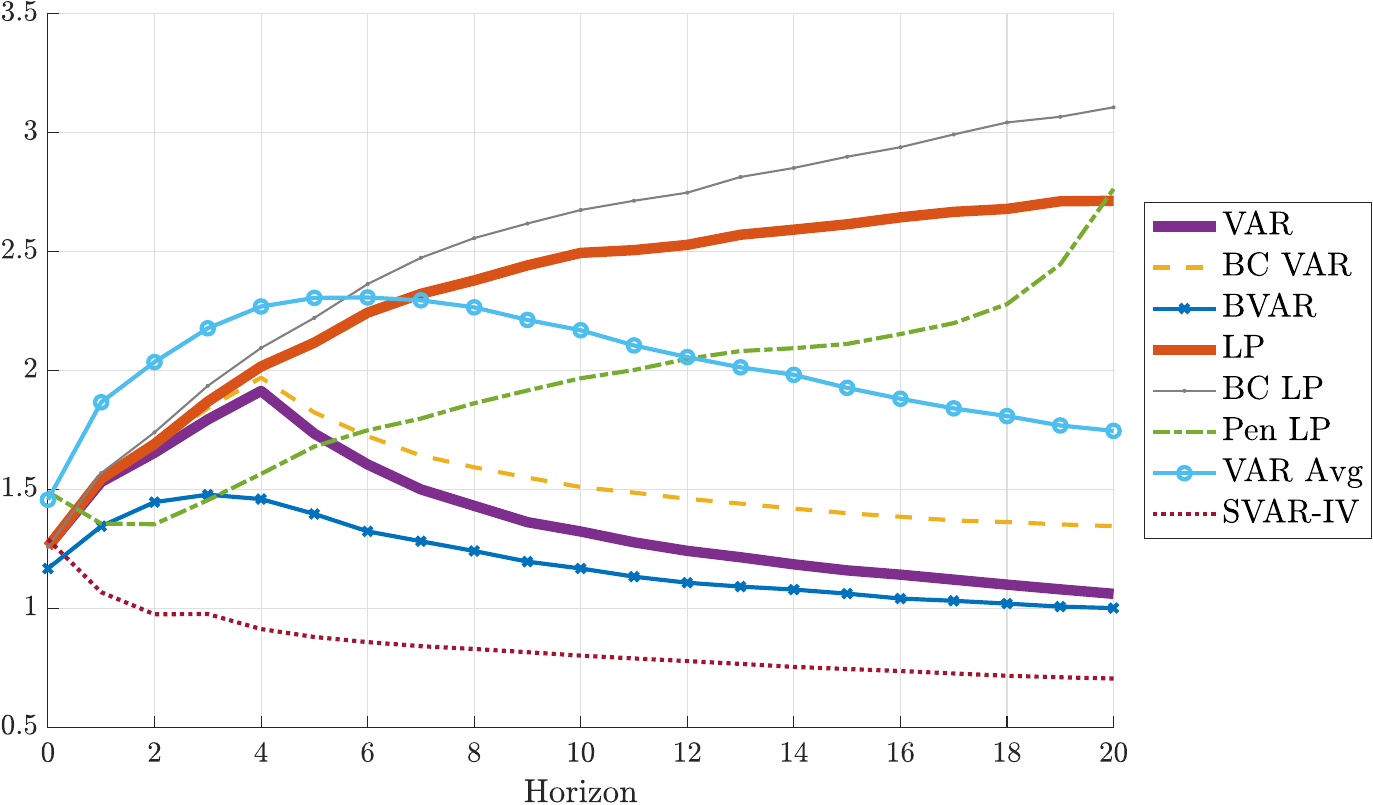}
\caption{Median (across DGPs) of standard deviation of the different estimation procedures, relative to $\sqrt{\frac{1}{21}\sum_{h=0}^{20}\theta_h^2}$.}
\label{fig:supp:std_iv_supp}
\end{figure}

\begin{figure}[tp]
\centering
\textsc{IV: Median bias of estimators, small degree of invertibility} \\[0.5\baselineskip]
\includegraphics[width=0.85\linewidth]{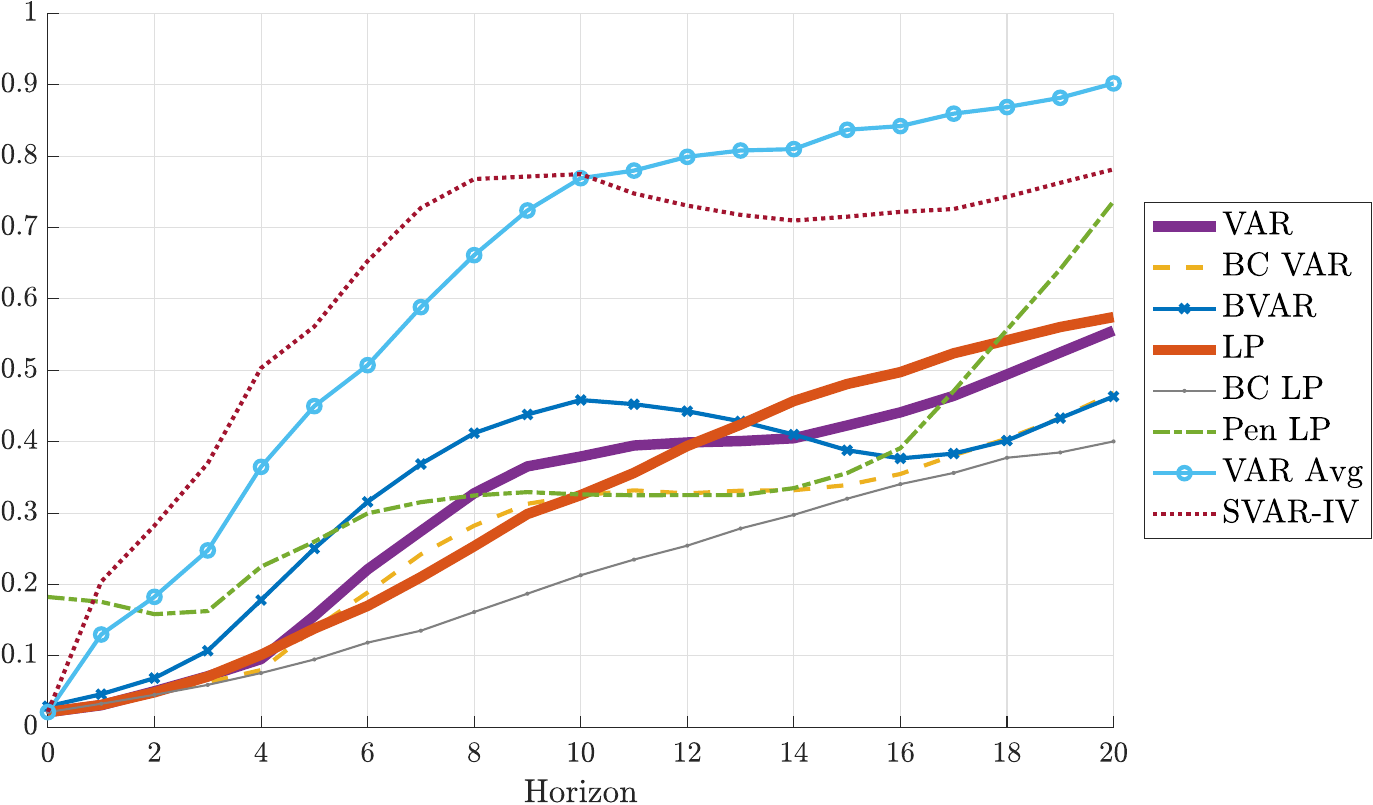}
\caption{Median (across DGPs) of absolute median bias of the different estimation procedures, relative to $\sqrt{\frac{1}{21}\sum_{h=0}^{20}\theta_h^2}$, for the 10\% of DGPs with the smallest degree of invertibility.}
\label{fig:supp:medbias_iv_lowinv_supp}

\vspace*{\floatsep}

\centering
\textsc{IV: Median bias of estimators, large degree of invertibility} \\[0.5\baselineskip]
\includegraphics[width=0.85\linewidth]{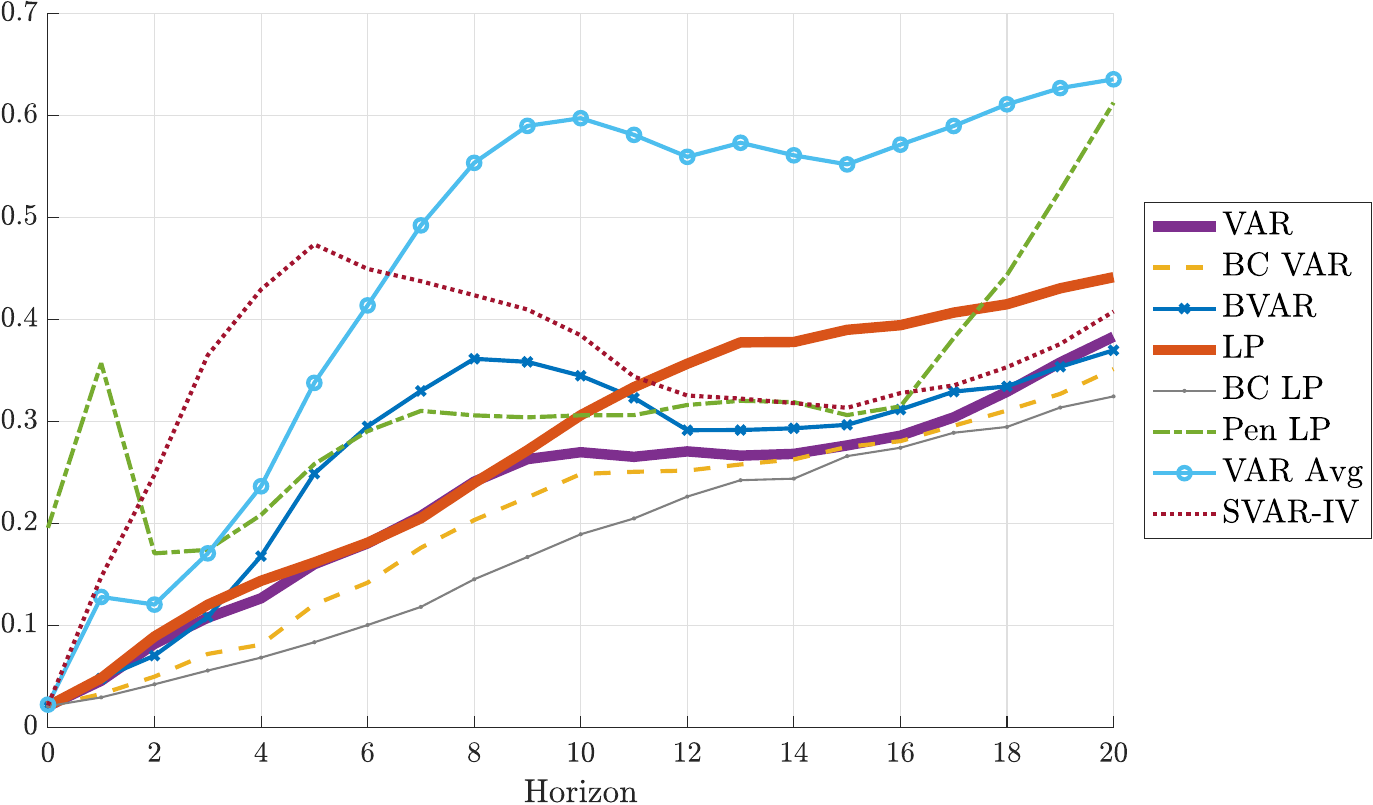}
\caption{Median (across DGPs) of standard deviation of the different estimation procedures, relative to $\sqrt{\frac{1}{21}\sum_{h=0}^{20}\theta_h^2}$, for the 10\% of DGPs with the largest degree of invertibility.}
\label{fig:supp:medbias_iv_highinv_supp}
\end{figure}

\clearpage

\subsection{Stationary DGPs}
\label{app:results_stationary}

For our baseline analysis we used a non-stationary DFM with variables in levels as our encompassing model. We did this to follow the predominant practice in applied work of using variables in levels in LP and VAR regressions. Nevertheless, some researchers work with data that has been transformed to stationarity prior to the analysis. We here discuss results from an analogous simulation study conducted with a stationary encompassing DFM.

\paragraph{Encompassing model.}
We consider a stationary version of the encompassing DFM \eqref{eq:factors}--\eqref{eq:idio_errors}. In particular, we parametrize the model based on the empirical reduced-form parameter estimates from \citet{Stock2016}, using the same specification as in \citet{Lazarus2018}. We provide a brief summary here and refer to \citeauthor{Stock2016} for further details. Differently from our main analysis, each series in the vector of observables $X_t$ is now transformed to ensure approximate stationarity. We follow \citeauthor{Stock2016} in selecting $n_f = 6$ factors, two lags in the factor equation \eqref{eq:factors}, and two lags in the idiosyncratic component equation \eqref{eq:idio_errors}. The reduced-form parameters are estimated by principal components and least-squares procedures; in particular, the factor VAR equation \eqref{eq:factors} is estimated by OLS instead of imposing a VECM model. This pins down all parameters of the DFM except for the structural impact response matrix $H$, which we construct exactly as in our main analysis.

\paragraph{DGP summary statistics.}
\begin{table}[t]
	\centering
	\textsc{Stationary DGPs: Summary statistics} \\[0.5\baselineskip]
	\renewcommand{\arraystretch}{1.2}
	\begin{tabular}{l|ccccccc}
		Percentile	&	min	&	10	&	25	&	50	&	75	&	90	&	max	\\
		\hline\hline
		& \multicolumn{6}{l}{ } \\[-2ex]
		\emph{Data and shocks} & \multicolumn{6}{l}{} \\[0.5ex]
		trace(long-run var)/trace(var)	&	0.36	&	0.76	&	1.00	&	1.45	&	2.18	&	3.71	&	18.09 \\
		Largest VAR eigenvalue & 0.84	&	0.84	&	0.84	&	0.84	&	0.85	&	0.86	&	0.91 \\
		Fraction of VAR coef's $\ell \geq 5$	& 0.02	&	0.11	&	0.15	&	0.21	&	0.27	&	0.34	&	0.77 \\
		\hline
		& \multicolumn{6}{l}{ } \\[-2ex]
		\emph{Impulse responses up to $h=20$} & \multicolumn{6}{l}{} \\[0.5ex]
		No. of interior local extrema	&	1	&	2	&	2	&	2	&	3	&	4	&	6 \\
		Horizon of max abs. value	&	0	&	0	&	0	&	0	&	1	&	2	&	8 \\
		Average/(max abs. value) &	-0.42	&	-0.16	&	-0.08	&	-0.02	&	0.06	&	0.11	&	0.43	\\
		$R^2$ in regression on quadratic	&	0.01	&	0.10	&	0.20	&	0.46	&	0.69	&	0.83	&	0.97 \\
		\hline
	\end{tabular}
	\caption{Quantiles of various population parameters across the stationary DGPs for observed shock identification. ``long-run var'': long-run variance of series. ``var'': variance of series. ``Largest VAR eigenvalue'': largest absolute eigenvalue of reduced-form VAR companion matrix. ``Fraction of VAR coef's $\ell \geq 5$'': $\sum_{\ell=5}^{1000} \|A_\ell^w\|/\sum_{\ell=1}^{1000} \|A_\ell^w\|$, where $A_\ell^w$ are the population VAR($\infty$) coefficient matrices and $\|\cdot\|$ is the Frobenius norm. ``Average/(max abs. value)'': $(\frac{1}{21}\sum_{h=0}^{20}\theta_h)/\max_h \lbrace |\theta_h|\rbrace$. ``$R^2$ in regression on quadratic'': $R^2$ from a regression of the impulse response function $\lbrace \theta_h \rbrace_{h=0}^{20}$ on a quadratic polynomial in $h$.}
	\label{tab:supp:dgp_summ_stationary}
\end{table}
\cref{tab:supp:dgp_summ_stationary} shows summary statistics of the 6,000 stationary DGPs (3,000 monetary DGPs and 3,000 fiscal ones, as in our baseline analysis). The persistence of the DGPs varies widely, as measured by the ratio $\tr(\mathit{LRV}(\bar{w}_t))/\tr(\var(\bar{w}_t))$, but note that none of the DGPs feature near-unit roots, with the largest absolute eigenvalue of the VAR companion matrix being below 0.91 in all cases. As in our baseline analysis, the DGPs feature marked heterogeneity in the characteristics of impulse response functions and in the degree to which the DGPs are well-approximated by a VAR(4) model.

\paragraph{Estimators.}
The estimation procedures are the same as in the baseline analysis, except for the Bayesian VAR. Rather than centering the Minnesota prior at independent random walks, we center at independent white noise, i.e., all autoregressive coefficients are shrunk towards 0. We also remove the ``sum-of-coefficients'' and ``dummy-initial-observation'' priors.

\paragraph{Results.}
\cref{fig:supp:bias_stationary,fig:supp:std_stationary,fig:supp:bestmethod_stationary} show bias, standard deviation, as well as the best estimation method choice for our experiments based on the stationary DFM, focusing here on the case of an observed shock. We summarize the findings in three lessons that are analogous to those in \cref{sec:results_1,sec:results_2,sec:results_3}.

\begin{enumerate}[1.]

\item Least-squares LP and VAR estimators again lie on opposite ends of the bias-variance spectrum. Furthermore, as in our main analysis, the slope of this trade-off is stark, with indifference between the two methods requiring the researcher to almost fully prioritize bias. Since the stationary DGPs are not highly persistent, bias correction has negligible impact, and indeed the uncorrected LP estimator has near-zero bias at all horizons (as predicted by asymptotic theory). Note also that the bias and standard deviation of the VAR estimators converge to zero as the horizon increases, a mechanical feature of VAR estimators in stationary DGPs that is not shared by LP methods.

\item LP (bias-corrected or uncorrected) remains the preferred method when the weight $\omega$ on bias in the loss function is very high, but this requires $\omega$ to be even closer to 1 than in our baseline analysis. The region (gray with vertical lines) where bias-corrected LP is preferred at the top of \cref{fig:supp:bestmethod_stationary} is very thin. Whenever the loss function puts non-trivial weight on precision, penalized LP is usually preferred to least-squares or bias-corrected LP due to the substantial variance reduction.

\item VAR methods tend to be preferred to least-squares or bias-corrected LP whenever the weight on variance in the loss function is non-trivial. However, for MSE loss, BVAR is now outperformed (on average across DGPs) at very short horizons $h \leq 2$ by penalized LP, another shrinkage method. BVAR tends to outperform least-squares VAR at shorter horizons, but this ranking is reversed at intermediate horizons; at long horizons, the methods perform similarly since the estimated impulse responses are close to 0 anyway. Though at first sight bias-corrected VAR (yellow with horizontal lines) looks attractive in \cref{fig:supp:bestmethod_stationary} at intermediate horizons when the weight on bias is moderately high, we remind the reader that least-squares VAR (purple) performs very similarly. VAR model averaging continues to perform poorly regardless of horizon and bias-variance preferences.

\end{enumerate}

\begin{figure}[tp]
\centering
\textsc{Observed shock, stationary DGPs: Bias of estimators} \\[0.5\baselineskip]
\includegraphics[width=0.85\linewidth]{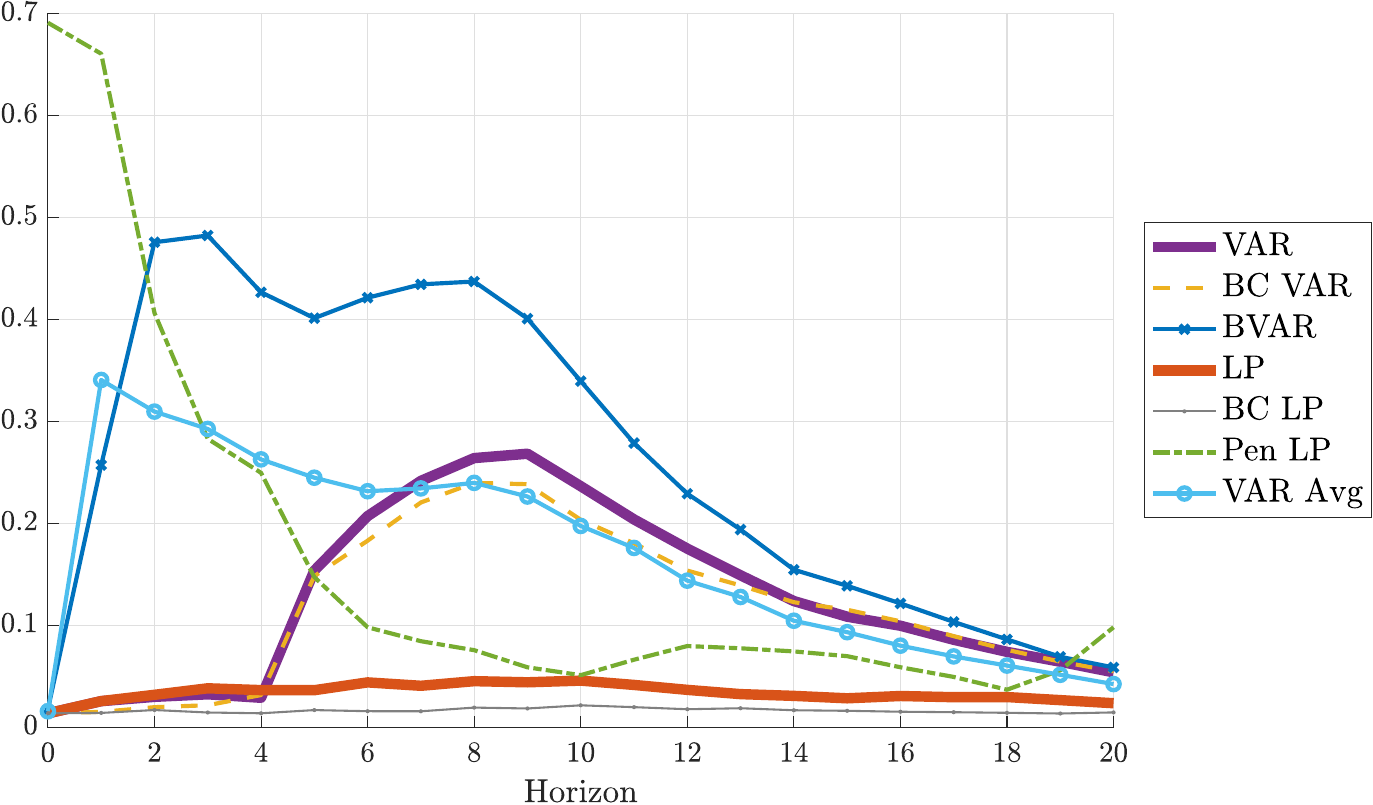}
\caption{Median (across DGPs) of absolute bias of the different estimation procedures, relative to $\sqrt{\frac{1}{21}\sum_{h=0}^{20}\theta_h^2}$.}
\label{fig:supp:bias_stationary}

\vspace*{\floatsep}

\centering
\textsc{Observed shock, stationary DGPs: Standard deviation of estimators} \\[0.5\baselineskip]
\includegraphics[width=0.85\linewidth]{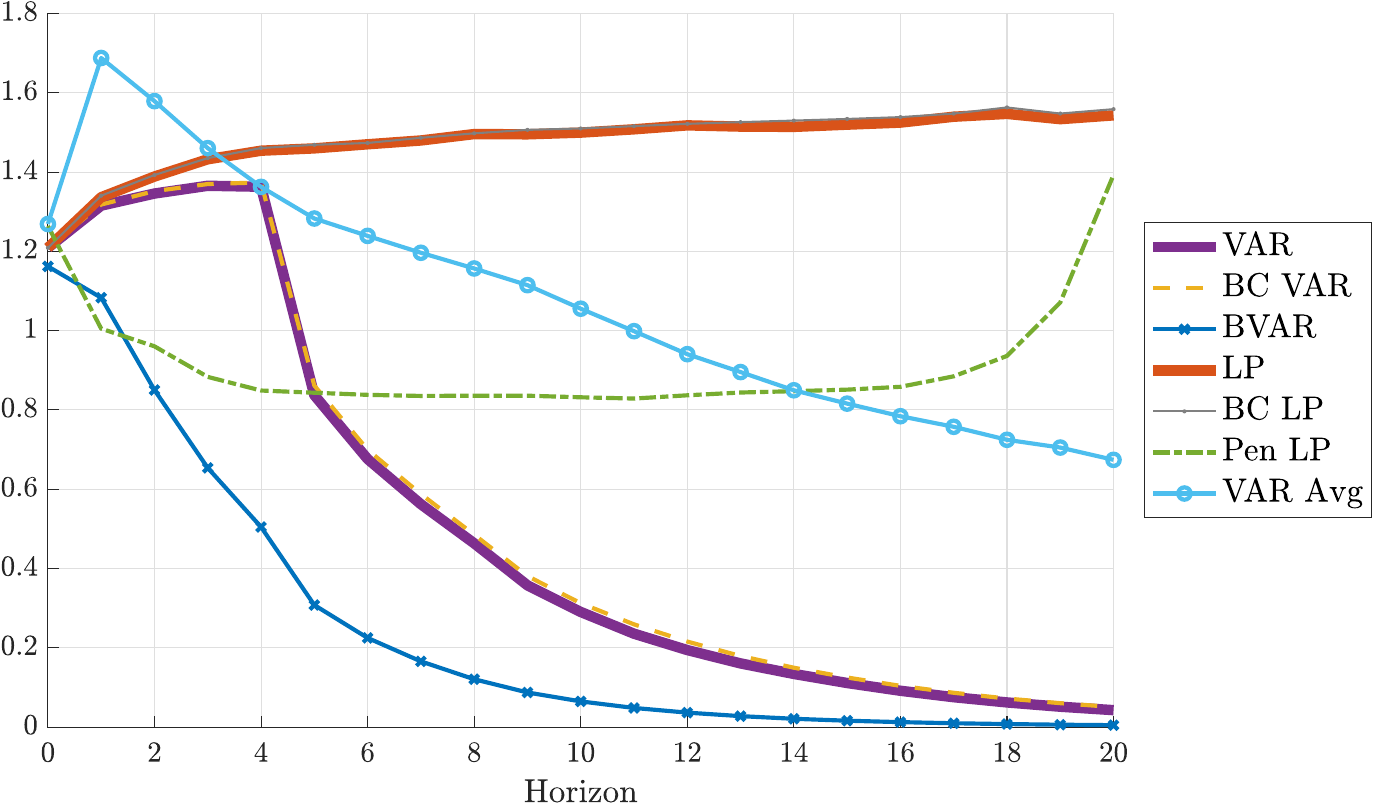}
\caption{Median (across DGPs) of standard deviation of the different estimation procedures, relative to $\sqrt{\frac{1}{21}\sum_{h=0}^{20}\theta_h^2}$.}
\label{fig:supp:std_stationary}
\end{figure}

\begin{figure}[t]
\centering
\textsc{Observed shock, stationary DGPs: Optimal estimation method} \\
\includegraphics[width=\linewidth,clip=true,trim=0 0.5em 0 2.4em]{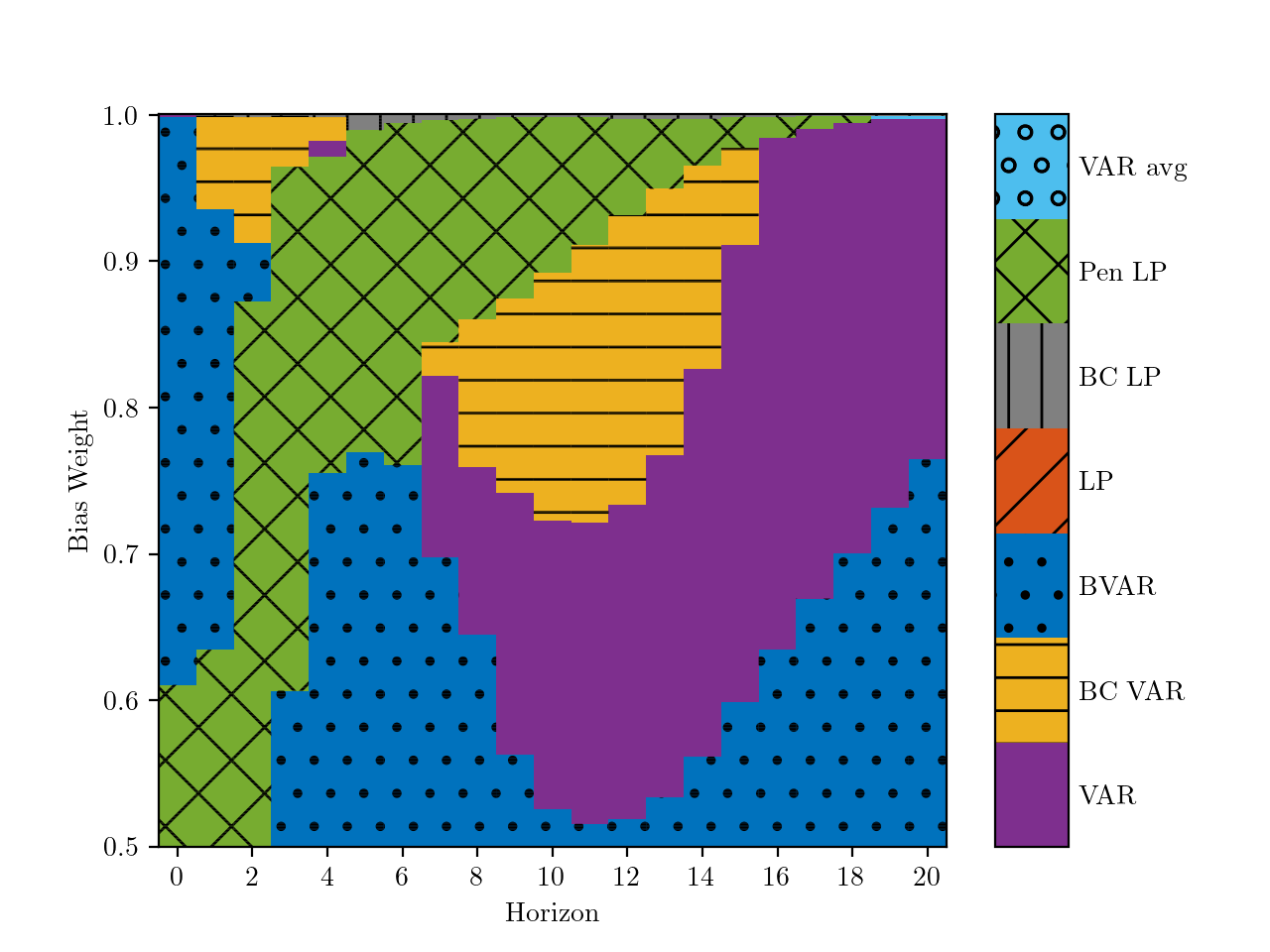}
\caption{Method that minimizes the average (across DGPs) loss function \eqref{eq:loss_simple}. Horizontal axis: impulse response horizon. Vertical axis: weight on squared bias in loss function. The loss function is normalized by the scale of the impulse response function, as in \cref{fig:supp:bias_stationary,fig:supp:std_stationary}.}
\label{fig:supp:bestmethod_stationary}
\end{figure}

\clearpage

\subsection{Recursive identification}
\label{app:results_recursive}
Here we provide results for the recursive impulse response estimand defined in \cref{sec:dgp_estimand} and \cref{app:recursive}.

\cref{tab:supp:dgp_summ_recursive} shows summary statistics for the impulse response functions in the recursive identification setting, analogous to the summary statistics for the ``observed shock'' case in the bottom half of \cref{tab:dgp_summ} in \cref{sec:dgp_summ_stat}. Note that the first three entries of the top half of that table apply without change here.\footnote{The last two entries of ourse do not apply here: the degree of invertibility is trivially equal to $1$, and we are not studying IV identification.}

\begin{table}[tp]
\centering
\textsc{Recursive identification: DGP summary statistics} \\[0.5\baselineskip]
\renewcommand{\arraystretch}{1.2}
\begin{tabular}{l|ccccccc}
Percentile	&	min	&	10	&	25	&	50	&	75	&	90	&	max	\\
\hline\hline
& \multicolumn{6}{l}{ } \\[-2ex]
\emph{Impulse responses up to $h=20$} & \multicolumn{6}{l}{} \\[0.5ex]
No. of interior local extrema	&	0	&	1	&	2	&	2	&	3	&	3	&	7		\\
Horizon of max abs. value	&	0	&	1	&	1	&	6	&	10	&	20	&	20	\\
Average/(max abs. value) &	-0.83	&	-0.70	&	-0.58	&	-0.07	&	0.38	&	0.65	&	0.94	\\
$R^2$ in regression on quadratic	&	0.01	&	0.45	&	0.65	&	0.83	&	0.93	&	0.98	&	1.00	\\
\hline
\end{tabular}
\caption{Quantiles of various population parameters across the 6,000 DGPs for recursive identification. ``Average/(max abs. value)'': $(\frac{1}{21}\sum_{h=0}^{20}\theta_h)/\max_h \lbrace |\theta_h|\rbrace$. ``$R^2$ in regression on quadratic'': R-squared from a regression of the impulse response function $\lbrace \theta_h \rbrace_{h=0}^{20}$ on a quadratic polynomial in $h$.}
\label{tab:supp:dgp_summ_recursive}
\end{table}

\cref{fig:supp:bias_recursive,fig:supp:std_recursive}  show the median (across DGPs) absolute bias and standard deviation of the various estimators. \cref{fig:supp:bestmethod_recursive} depicts the best estimation method as a function of the horizon and the bias weight $\omega$ in the loss function (which is averaged across DGPs). These three figures are reasonably similar---both qualitatively as well as quantitatively---to the corresponding figures for the ``observed shock'' estimands in \cref{sec:results}, with one notable difference: the relative bias increase for BVAR is smaller, and so now BVAR looks even more attractive, leading to the large solid-dotted blue area in \cref{fig:supp:bestmethod_recursive}. Intuitively, since recursive (i.e., Cholesky) shock identification depends heavily on estimating the reduced-form innovation variance-covariance matrix of the multivariate system, the stylized prior information about this matrix imposed by the BVAR provides helpful additional shrinkage that the other methods do not exploit.

\begin{figure}[tp]
\centering
\textsc{Recursive identification: Bias of estimators} \\[0.5\baselineskip]
\includegraphics[width=0.85\linewidth]{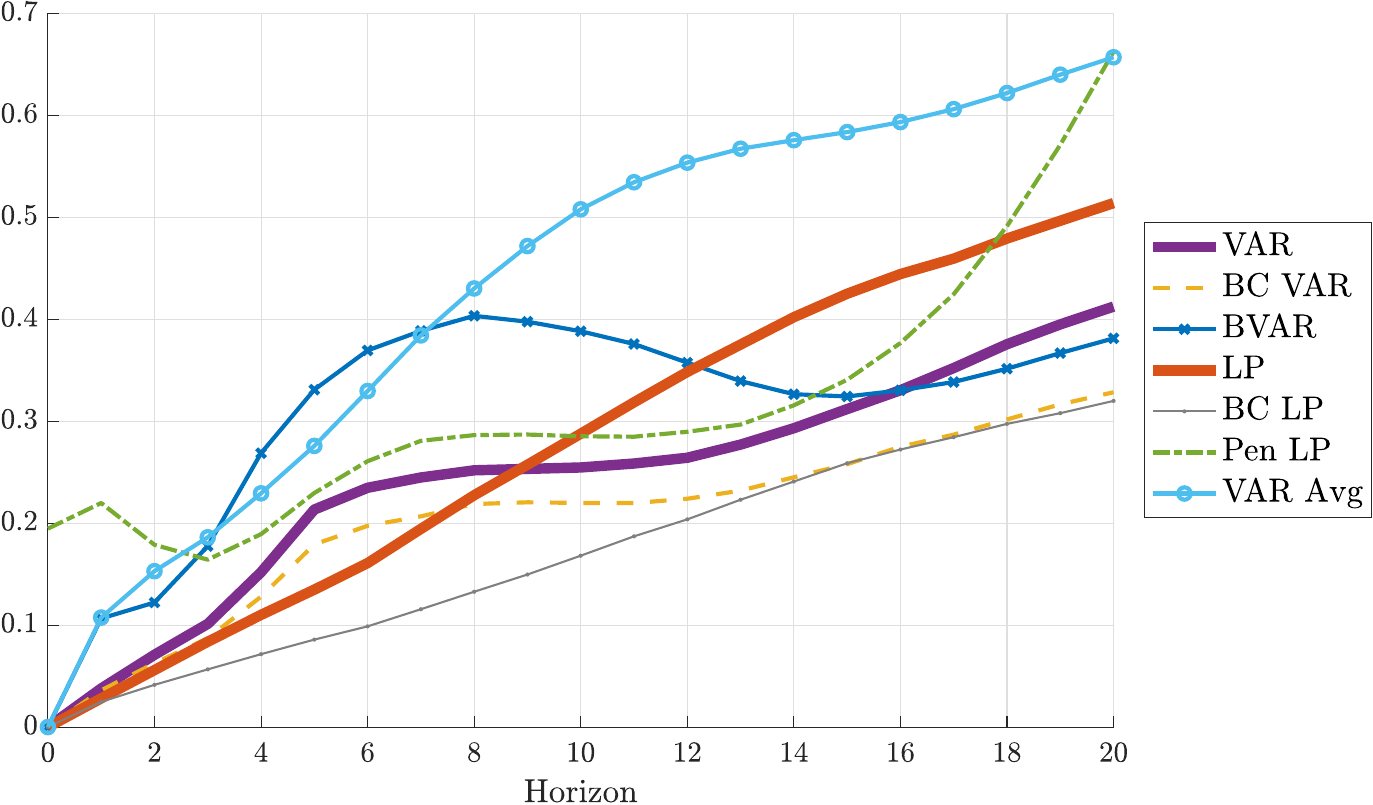}
\caption{Median (across DGPs) of absolute bias of the different estimation procedures, relative to $\sqrt{\frac{1}{21}\sum_{h=0}^{20}\theta_h^2}$.}
\label{fig:supp:bias_recursive}

\vspace*{\floatsep}

\centering
\textsc{Recursive identification: Standard deviation of estimators} \\[0.5\baselineskip]
\includegraphics[width=0.85\linewidth]{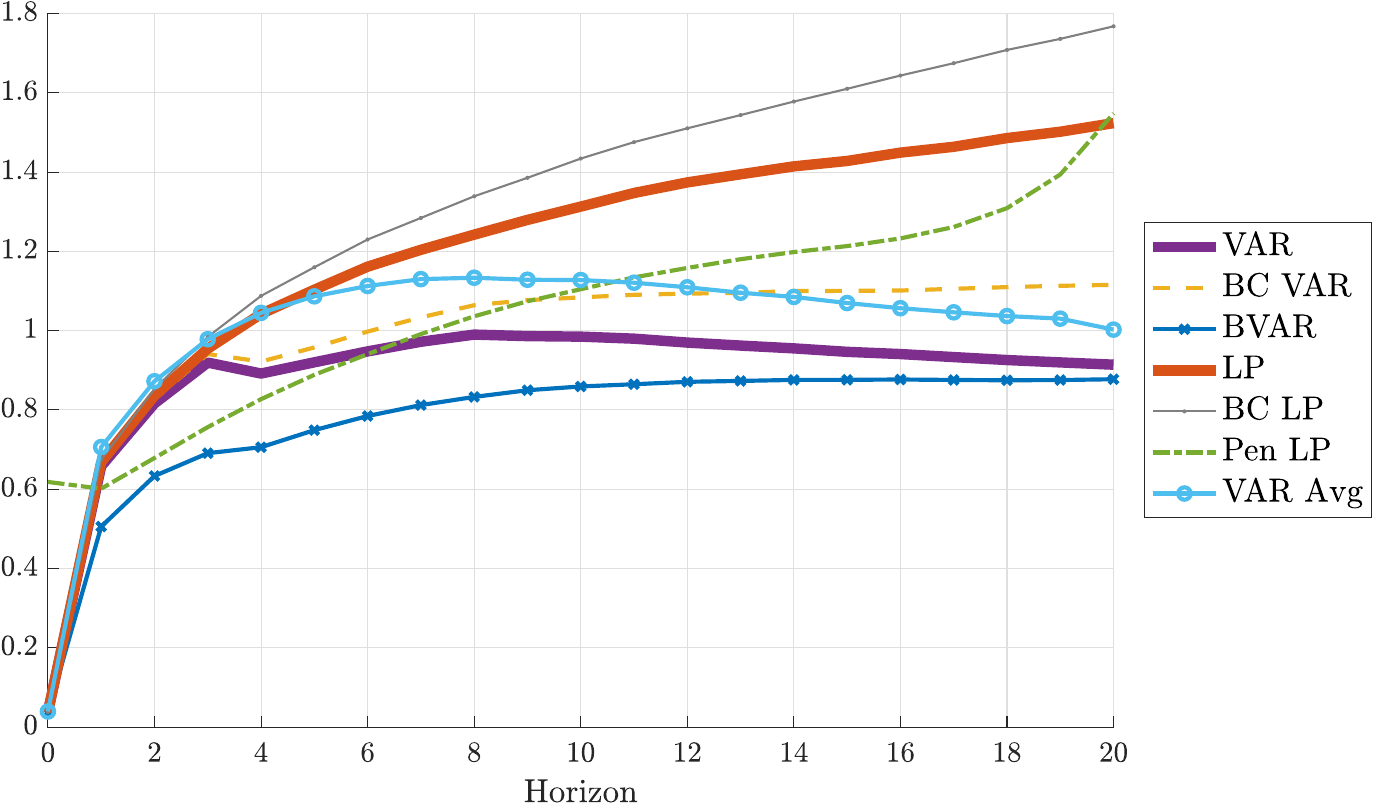}
\caption{Median (across DGPs) of standard deviation of the different estimation procedures, relative to $\sqrt{\frac{1}{21}\sum_{h=0}^{20}\theta_h^2}$.}
\label{fig:supp:std_recursive}
\end{figure}

\begin{figure}[t]
\centering
\textsc{Recursive identification: Optimal estimation method} \\
\includegraphics[width=\linewidth,clip=true,trim=0 0.5em 0 2.4em]{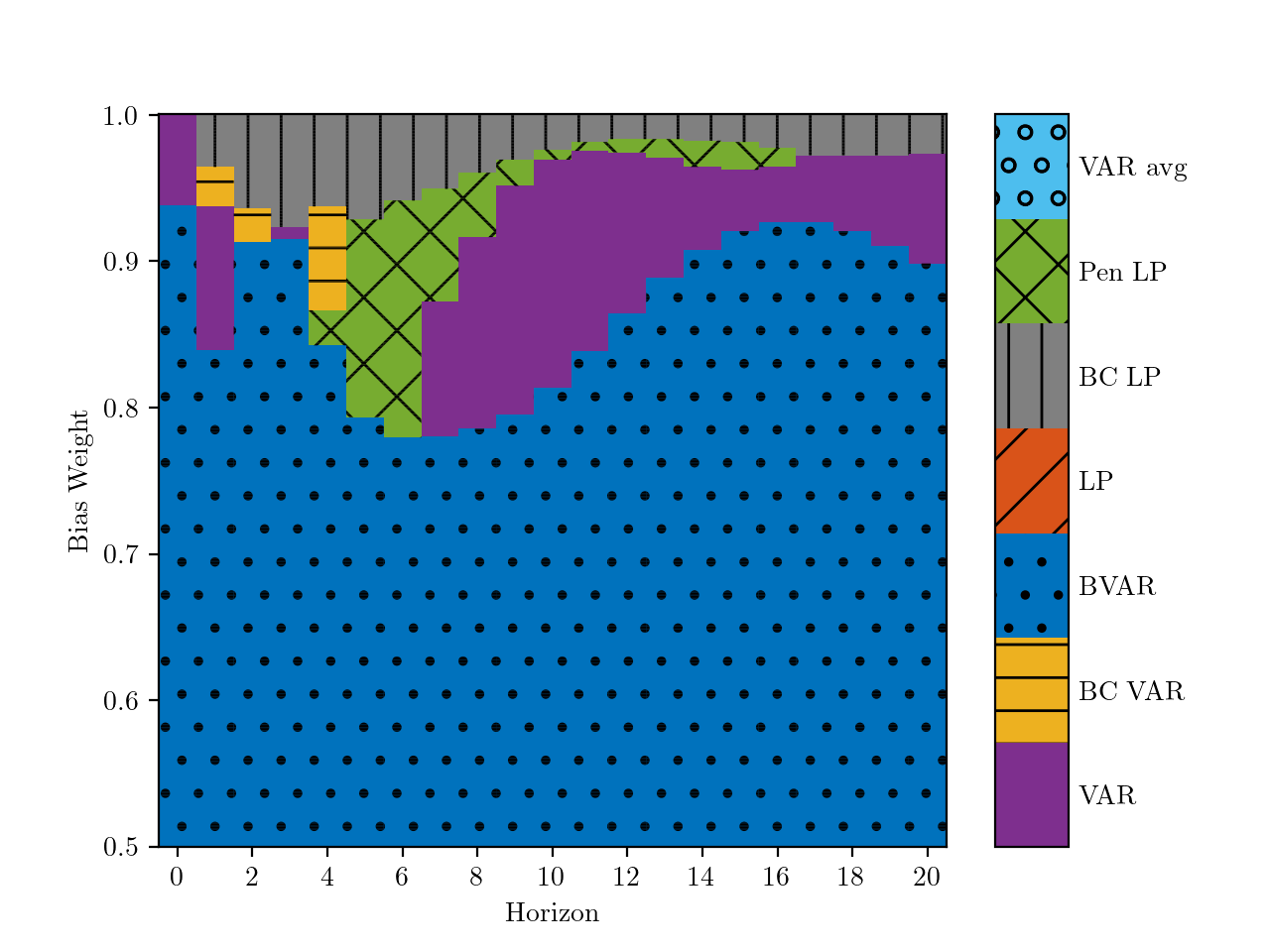}
\caption{Method that minimizes the average (across DGPs) loss function \eqref{eq:loss_simple}. Horizontal axis: impulse response horizon. Vertical axis: weight on squared bias in loss function. The loss function is normalized by the scale of the impulse response function, as in \cref{fig:bias_obsshock,fig:std_obsshock}. At $h=0$, VAR and LP are numerically identical; we break the tie in favor of VAR.}
\label{fig:supp:bestmethod_recursive}
\end{figure}

\clearpage

\subsection{Salient observables}
\label{app:results_salientobs}

We here re-run our analysis on a restricted set of DGPs that use a smaller subset of particularly salient time series. Whereas our baseline analysis randomly draws series from the large set of 207 variables included in the empirical DFM of \citet{Stock2016}, we now consider the exhaustive list of all possible five-variable combinations of 17 oft-used series.

Our subset of salient series includes (\citeauthor{Stock2016} Data Appendix series \# in brackets): \emph{real GDP (1); real consumption (2); real investment (6); real government expenditure (12); the unemployment rate (56); personal consumption expenditure prices (95); the GDP deflator (97); the core consumer price index (121); average hourly earnings (132); the federal funds rate (142); the 10-year Treasury rate (147); the BAA 10-year spread (151); an index of the U.S. dollar exchange rate relative to other major currencies (172); the S\&P 500 (181); a real house price index (193); consumer expectations (196); and real oil prices (202).} As in our baseline analysis, we force each DGP to include either the federal funds rate or government spending (for monetary or fiscal shock estimands, respectively) as well as at least one real activity series (categories 1--3) and one price series (category 6). Subject to these constraints, we then generate the exhaustive list of all five-variable combinations of the salient series. This yields a total of 1,581 DGPs (845 monetary shock DGPs and 736 fiscal shock DGPs).

Results for bias, standard deviation, and optimal method choice are reported in \cref{fig:supp:bias_obsshock_salient,fig:supp:std_obsshock_salient,fig:supp:bestmethod_obsshock_salient}. The figures look very similar to those from our baseline analysis. We therefore conclude that there is little in the way of systematic differences between the larger set of variables included in the full DFM and this smaller subset of particularly salient macroeconomic time series.

\begin{figure}[tp]
	\centering
	\textsc{Observed shock, salient observables: Bias of estimators} \\[0.5\baselineskip]
	\includegraphics[width=0.85\linewidth]{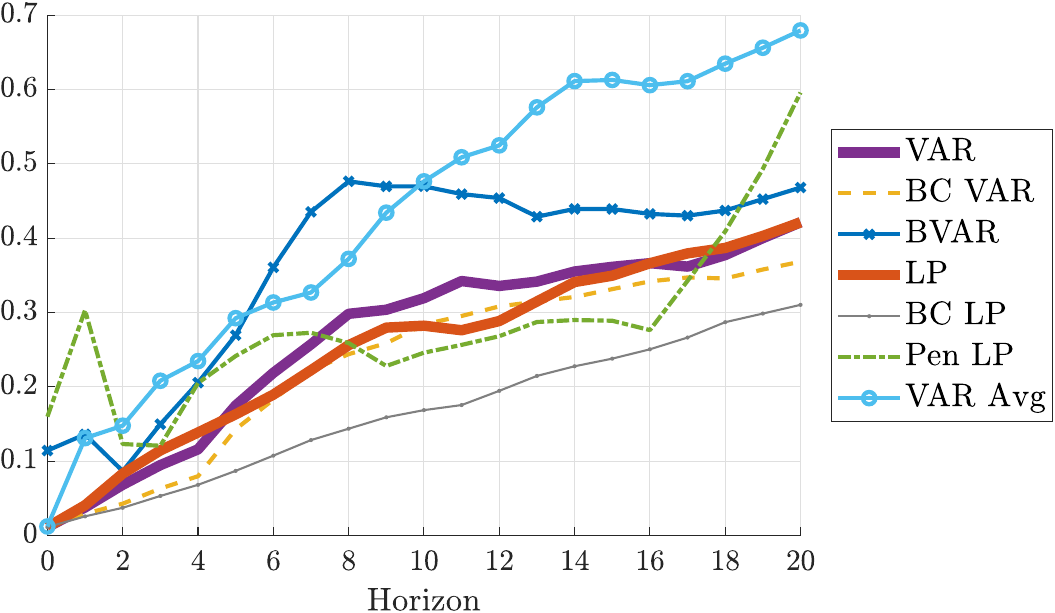}
	\caption{Median (across DGPs) of absolute bias of the different estimation procedures, relative to $\sqrt{\frac{1}{21}\sum_{h=0}^{20}\theta_h^2}$.}
	\label{fig:supp:bias_obsshock_salient}
	
	\vspace*{\floatsep}
	
	\centering
	\textsc{Observed shock, salient observables: Standard deviation of estimators} \\[0.5\baselineskip]
	\includegraphics[width=0.85\linewidth]{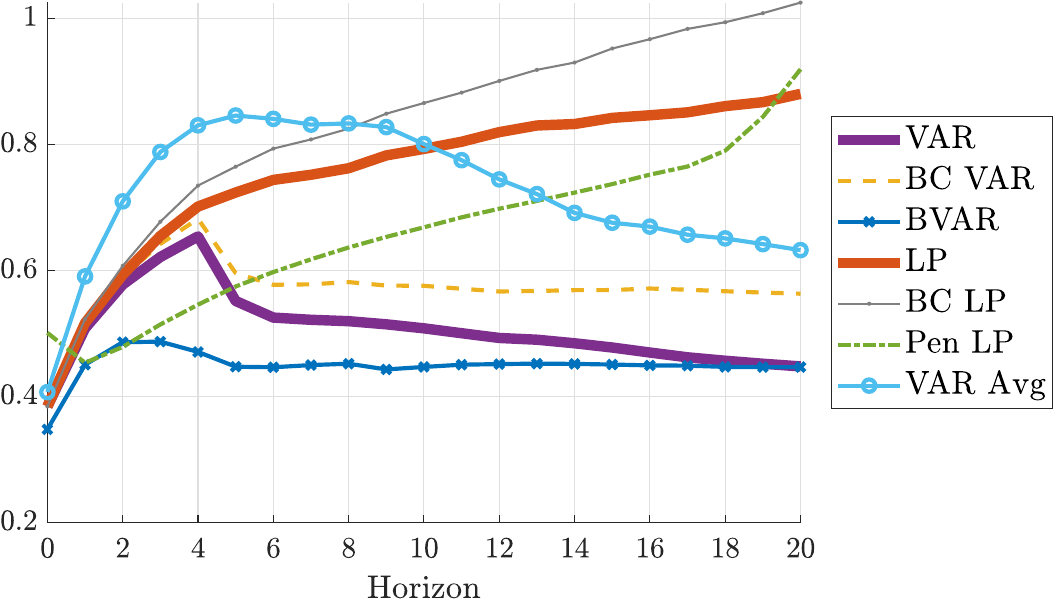}
	\caption{Median (across DGPs) of standard deviation of the different estimation procedures, relative to $\sqrt{\frac{1}{21}\sum_{h=0}^{20}\theta_h^2}$.}
	\label{fig:supp:std_obsshock_salient}
\end{figure}

\begin{figure}[t]
	\centering
	\textsc{Observed shock, salient observables: Optimal estimation method} \\
	\includegraphics[width=\linewidth,clip=true,trim=0 0.5em 0 2.4em]{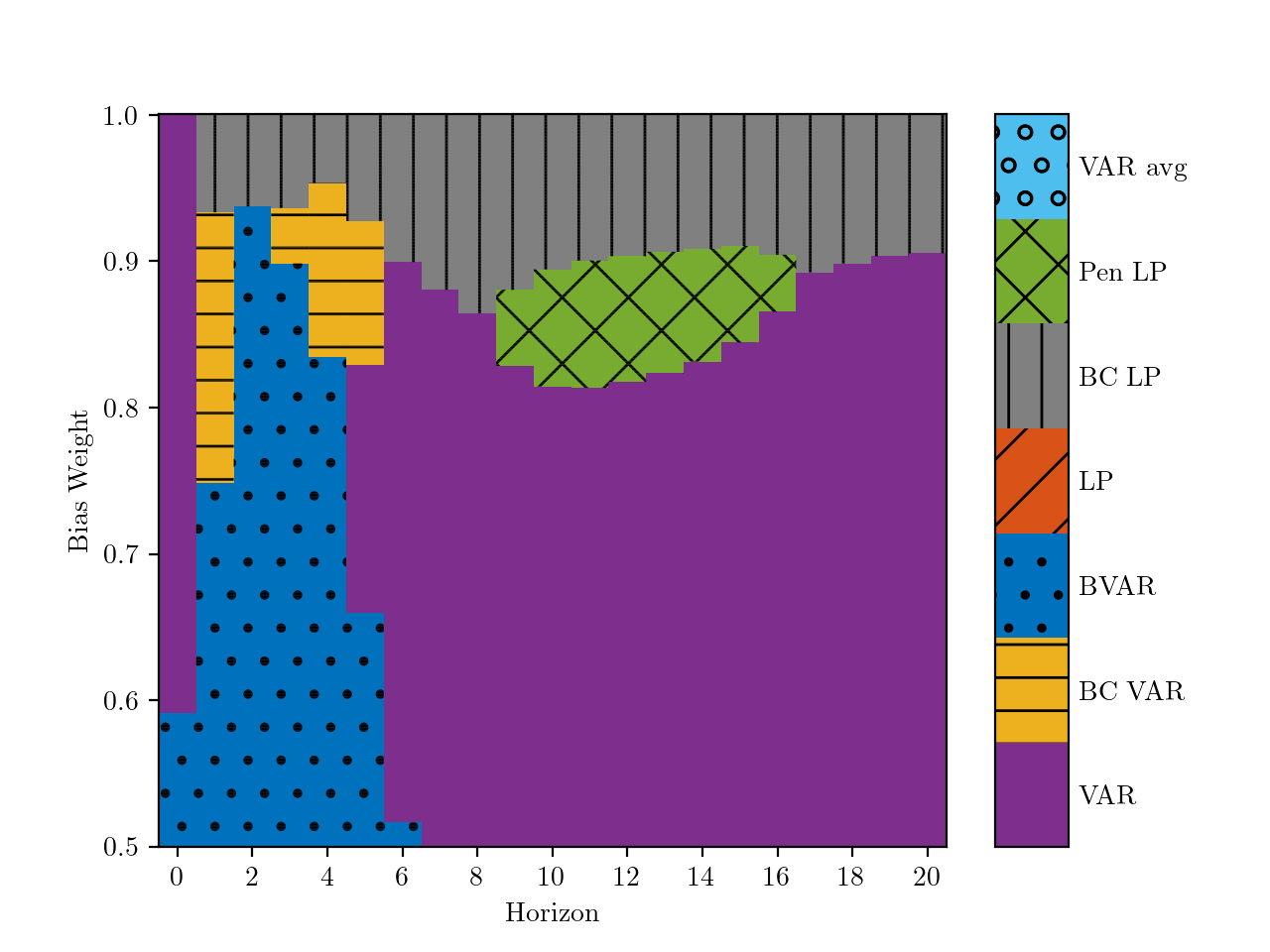}
	\caption{Method that minimizes the average (across DGPs) loss function \eqref{eq:loss_simple}. Horizontal axis: impulse response horizon. Vertical axis: weight on squared bias in loss function. The loss function is normalized by the scale of the impulse response function, as in \cref{fig:bias_obsshock,fig:std_obsshock}. At $h=0$, VAR and LP are numerically identical; we break the tie in favor of VAR.}
	\label{fig:supp:bestmethod_obsshock_salient}
\end{figure}

\clearpage

\subsection{90th percentile loss}
\label{app:results_losspcntl}

Whereas our baseline results report medians (of bias and standard deviation) across DGPs, we now report the 90th percentiles across DGPs. This places the spotlight on those DGPs that are particularly challenging for impulse response estimation.

Figures for bias, standard deviation, and optimal method choice are displayed in \cref{fig:supp:bias_obsshock_p90,fig:supp:std_obsshock_p90,fig:supp:bestmethod_obsshock_p90}. By construction, the bias and standard deviation numbers are now higher for all estimators. Importantly, however, \emph{relative} magnitudes do not change by much; that is, for DGPs in which VARs or shrinkage techniques do poorly, least-squares and bias-corrected LP tend to do just as poorly (relative to their respective median performance). As a result, optimal method choice for a researcher that evaluates loss at the 90th percentile looks similar to our baseline.

\begin{figure}[tp]
	\centering
	\textsc{Observed shock, 90th percentile of loss: Bias of estimators} \\[0.5\baselineskip]
	\includegraphics[width=0.85\linewidth]{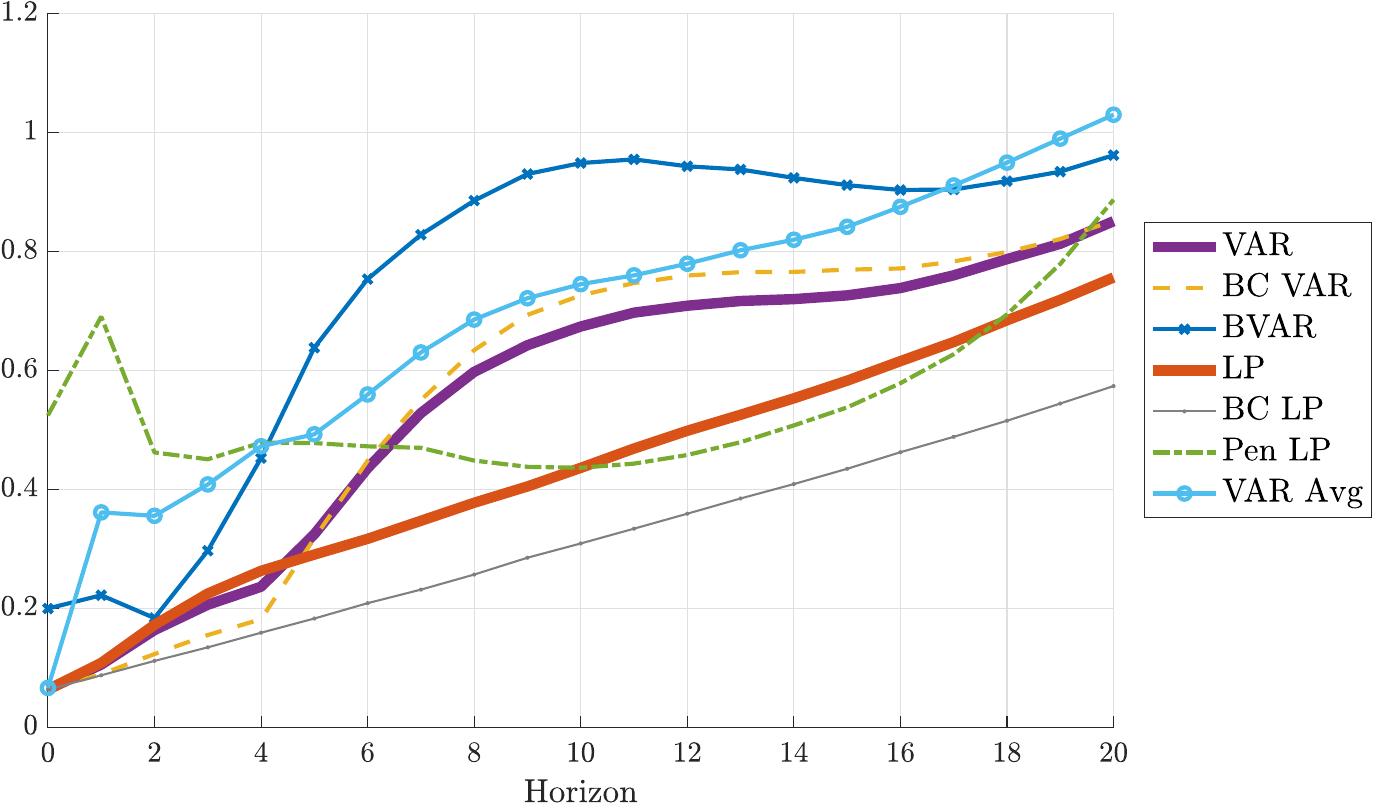}
	\caption{90th percentile (across DGPs) of absolute bias of the different estimation procedures, relative to $\sqrt{\frac{1}{21}\sum_{h=0}^{20}\theta_h^2}$.}
	\label{fig:supp:bias_obsshock_p90}
	
	\vspace*{\floatsep}
	
	\centering
	\textsc{Observed shock, 90th percentile of loss: Standard deviation of estimators} \\[0.5\baselineskip]
	\includegraphics[width=0.85\linewidth]{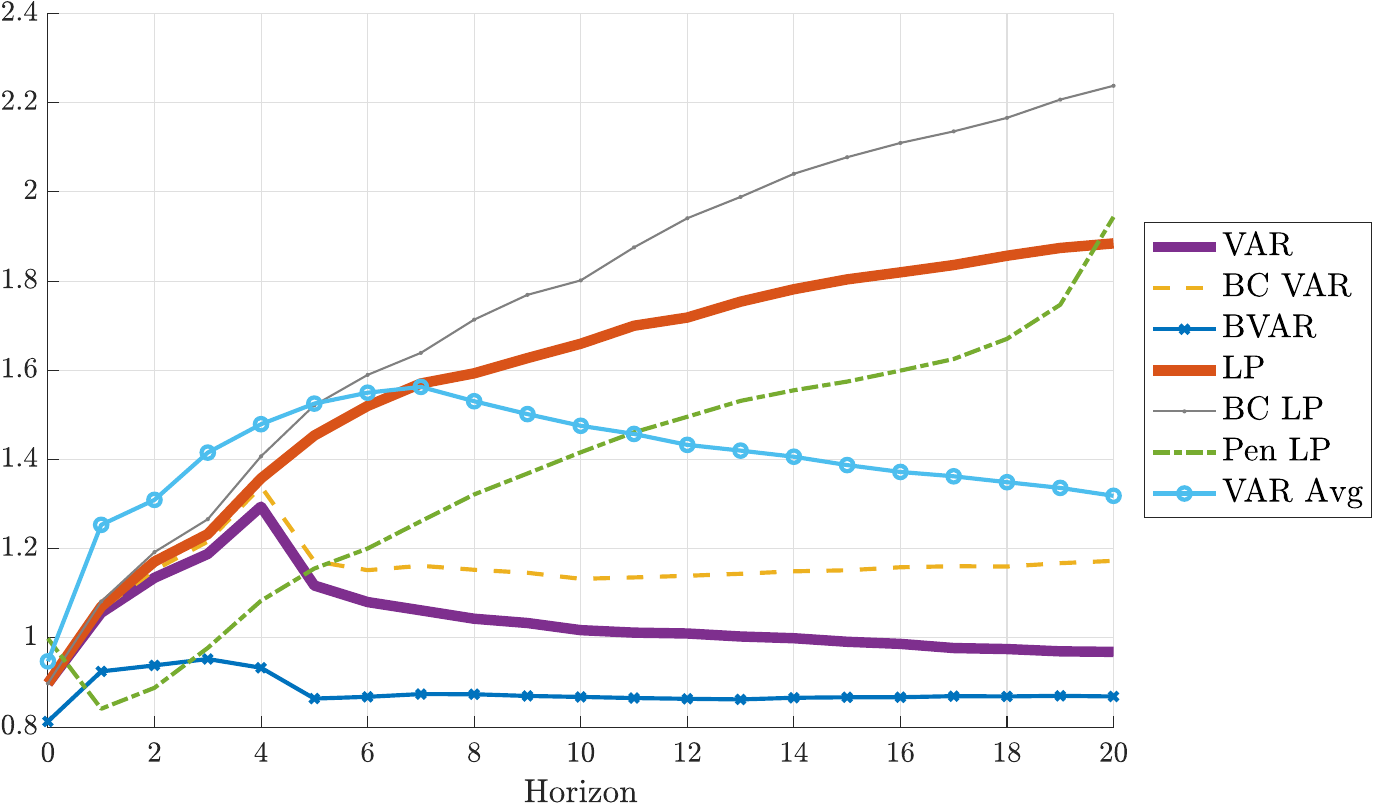}
	\caption{90th percentile (across DGPs) of standard deviation of the different estimation procedures, relative to $\sqrt{\frac{1}{21}\sum_{h=0}^{20}\theta_h^2}$.}
	\label{fig:supp:std_obsshock_p90}
\end{figure}

\begin{figure}[t]
	\centering
	\textsc{Observed shock, 90th percentile of loss: Optimal estimation method} \\
	\includegraphics[width=\linewidth,clip=true,trim=0 0.5em 0 2.4em]{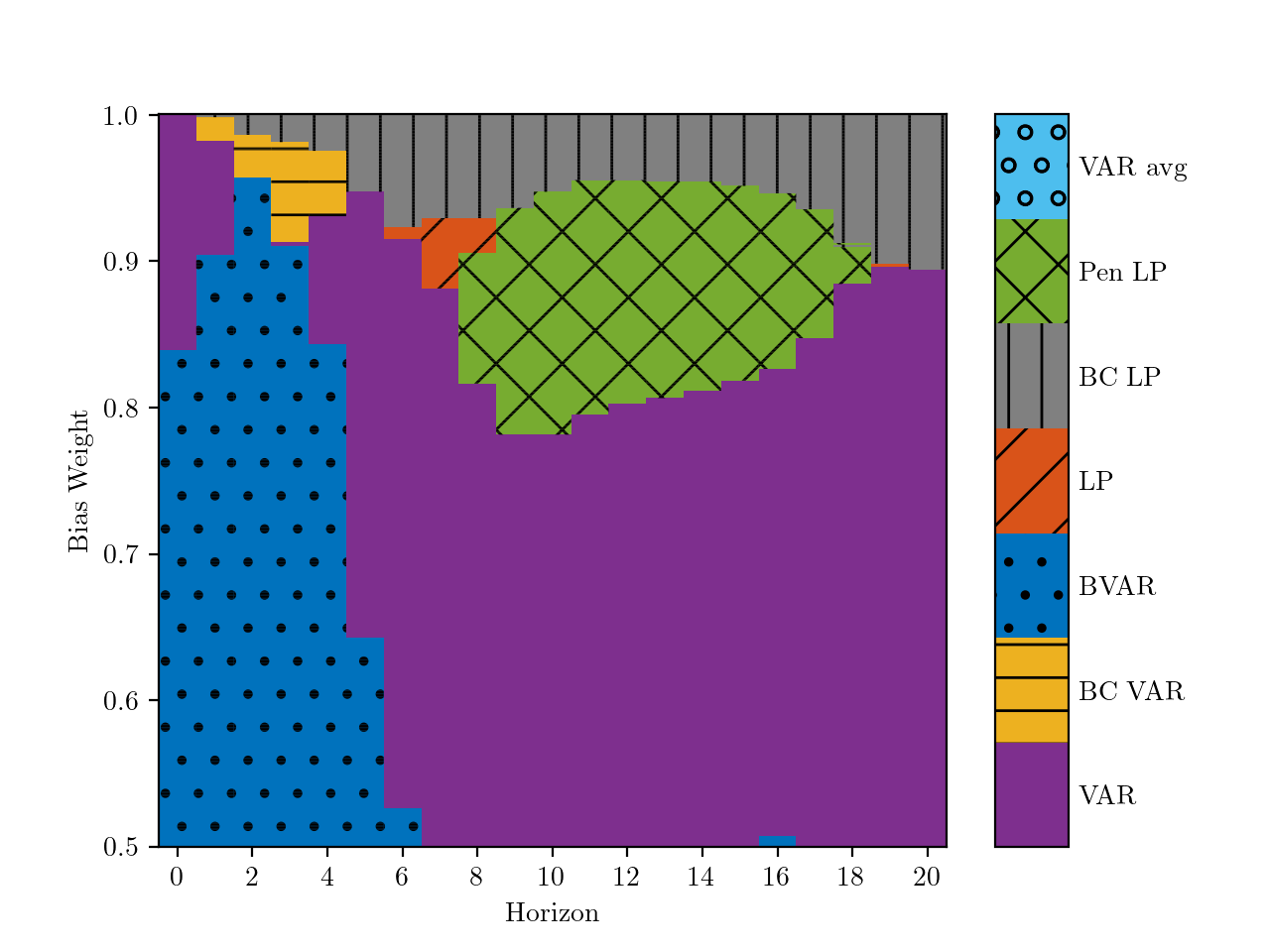}
	\caption{Method that minimizes the 90th percentile (across DGPs) of the loss function \eqref{eq:loss_simple}. Horizontal axis: impulse response horizon. Vertical axis: weight on squared bias in loss function. The loss function is normalized by the scale of the impulse response function, as in \cref{fig:bias_obsshock,fig:std_obsshock}. At $h=0$, VAR and LP are numerically identical; we break the tie in favor of VAR.}
	\label{fig:supp:bestmethod_obsshock_p90}
\end{figure}

\clearpage

\subsection{Fiscal and monetary shocks}
\label{app:results_fiscal_monetary}
Recall that the results from \cref{sec:results} combine fiscal and monetary policy shock estimands. We here break the results down by policy shock estimand.

\cref{fig:supp:bias_obsshock_fiscal,fig:supp:std_obsshock_fiscal} show the bias and standard deviation plots for the 3,000 fiscal shock DGPs, while \cref{fig:supp:bias_obsshock_monetary,fig:supp:std_obsshock_monetary} show the analogous figures for the 3,000 monetary shock DGPs. The results are qualitatively similar across the two kinds of DGPs, including the relative rankings of the various estimation procedures. However, the overall level of the standard deviations is somewhat higher in the fiscal shock case for all estimation methods.

\begin{figure}[tp]
\centering
\textsc{Observed fiscal shock: Bias of estimators} \\[0.5\baselineskip]
\includegraphics[width=0.85\linewidth]{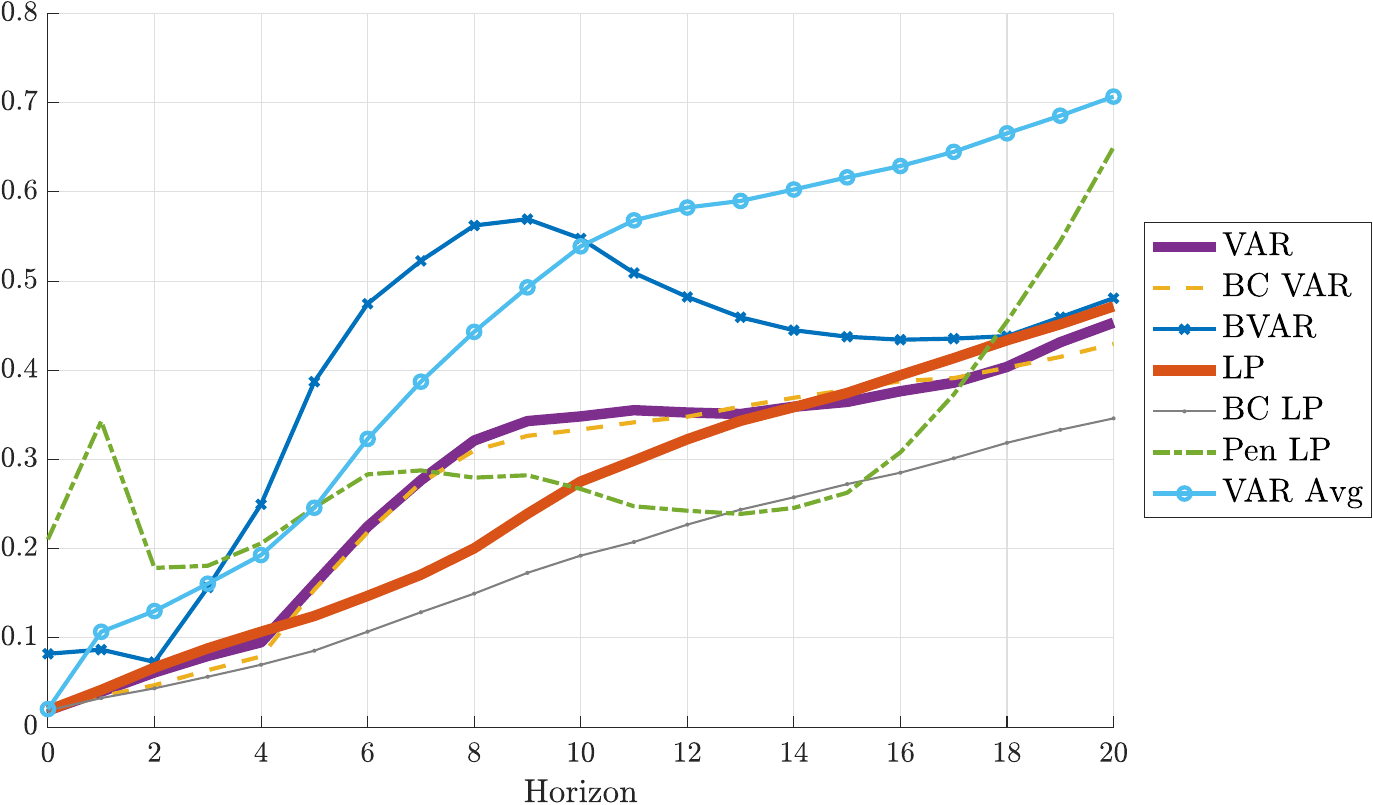}
\caption{Median (across DGPs) of absolute bias of the different estimation procedures, relative to $\sqrt{\frac{1}{21}\sum_{h=0}^{20}\theta_h^2}$.}
\label{fig:supp:bias_obsshock_fiscal}

\vspace*{\floatsep}

\centering
\textsc{Observed fiscal shock: Standard deviation of estimators} \\[0.5\baselineskip]
\includegraphics[width=0.85\linewidth]{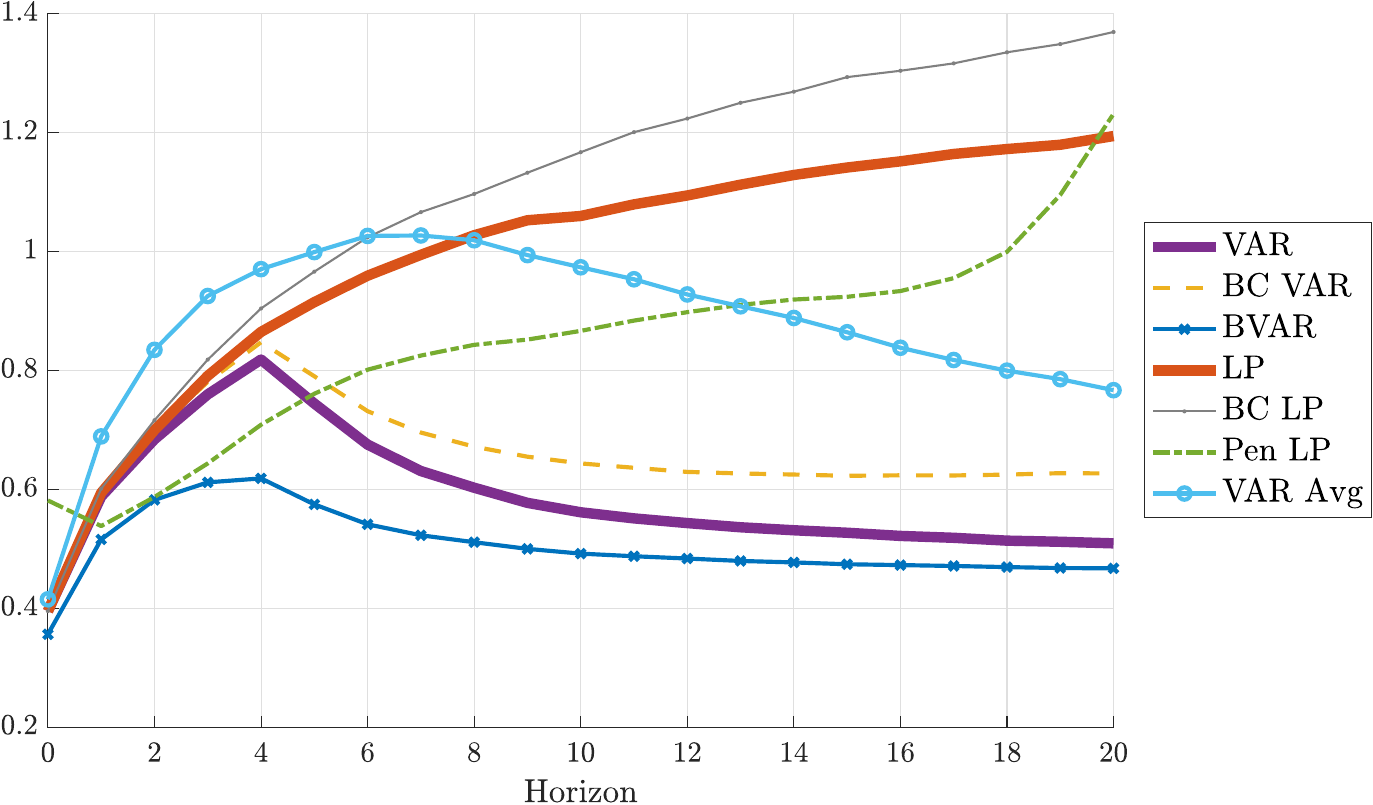}
\caption{Median (across DGPs) of standard deviation of the different estimation procedures, relative to $\sqrt{\frac{1}{21}\sum_{h=0}^{20}\theta_h^2}$.}
\label{fig:supp:std_obsshock_fiscal}
\end{figure}

\begin{figure}[tp]
\centering
\textsc{Observed monetary shock: Bias of estimators} \\[0.5\baselineskip]
\includegraphics[width=0.85\linewidth]{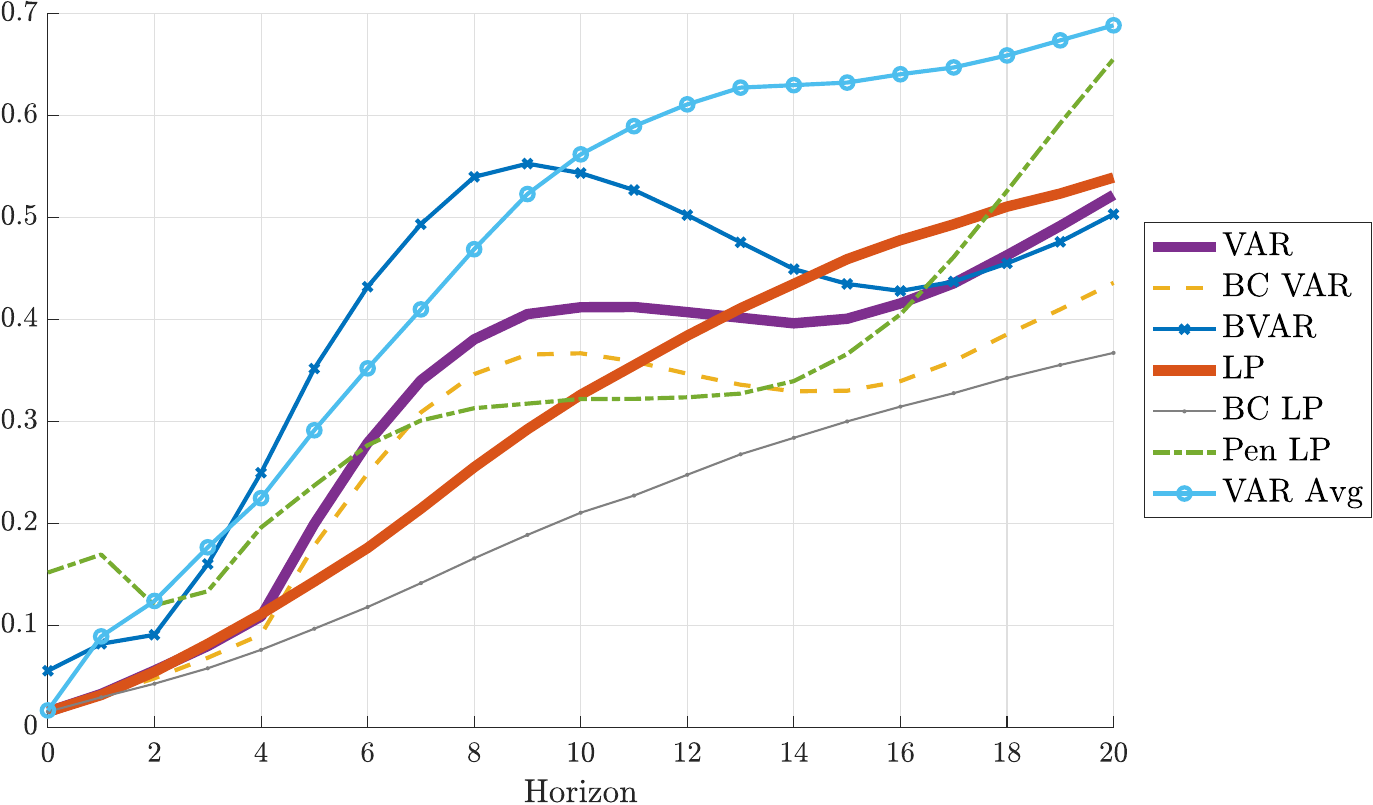}
\caption{Median (across DGPs) of absolute bias of the different estimation procedures, relative to $\sqrt{\frac{1}{21}\sum_{h=0}^{20}\theta_h^2}$.}
\label{fig:supp:bias_obsshock_monetary}

\vspace*{\floatsep}

\centering
\textsc{Observed monetary shock: Standard deviation of estimators} \\[0.5\baselineskip]
\includegraphics[width=0.85\linewidth]{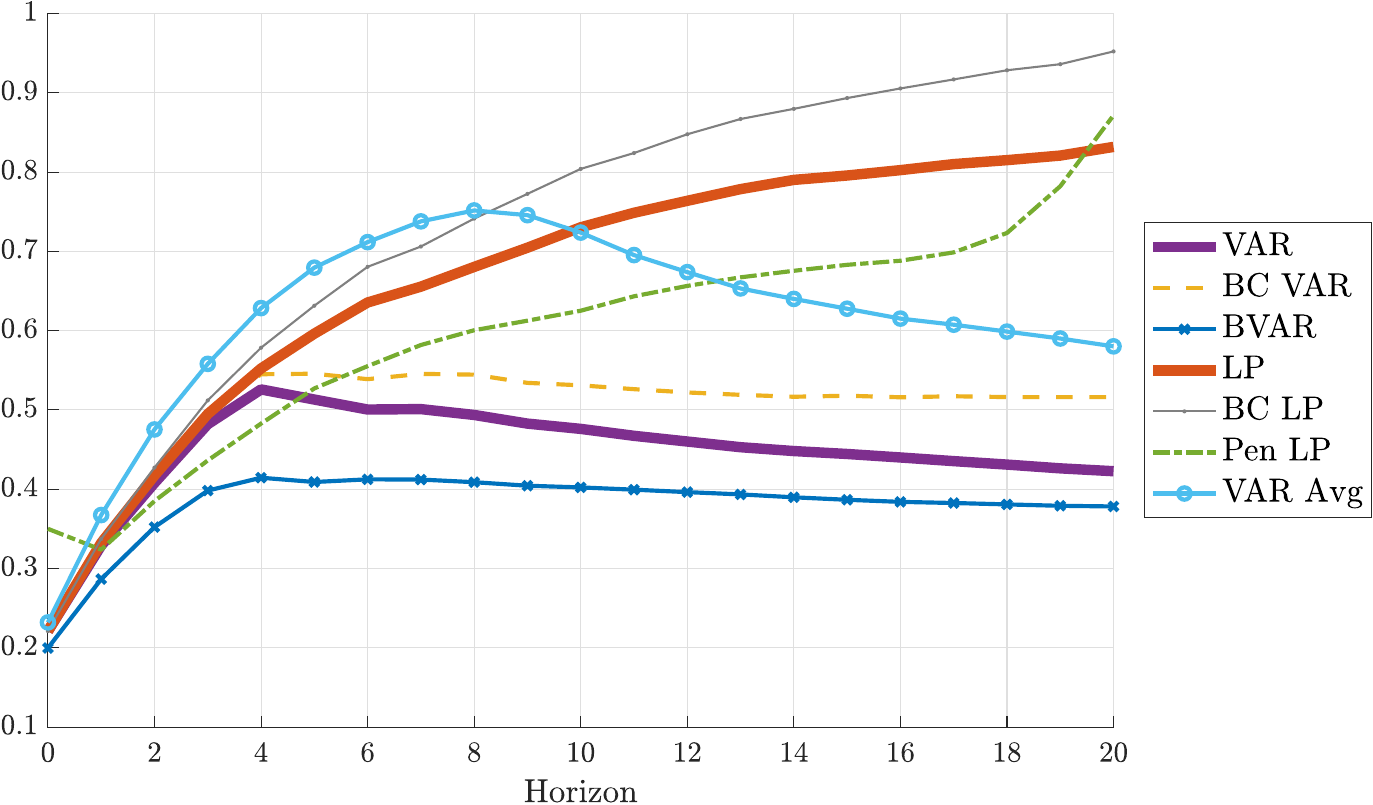}
\caption{Median (across DGPs) of standard deviation of the different estimation procedures, relative to $\sqrt{\frac{1}{21}\sum_{h=0}^{20}\theta_h^2}$.}
\label{fig:supp:std_obsshock_monetary}
\end{figure}

\clearpage

\subsection{Longer estimation lag length}
\label{app:results_lag}
Here we provide results for ``observed shock'' identification when the estimation lag length is increased to $p=8$ (recall that we set $p=4$ in \cref{sec:results}).

\cref{fig:supp:bias_obsshock_lag8,fig:supp:std_obsshock_lag8} show the median (across DGPs) absolute bias and standard deviation of the estimation methods, while \cref{fig:supp:bestmethod_obsshock_lag8} shows the optimal method choice according to the loss function (which has been averaged across DGPs). Consistent with the theoretical results in \citet{Plagborg2020}, least-squares LPs and VARs now perform similarly at all horizons $h \leq p = 8$. BVAR shrinkage now looks even more appealing than in our main analysis for loss functions with bias weight $\omega \leq 0.7$, due to the much lower standard deviation of this procedure. VAR bias correction has more bite  in terms of reducing bias than in the baseline results. Otherwise the qualitative conclusions from \cref{sec:results} are unchanged.

\begin{figure}[tp]
\centering
\textsc{Observed shock, 8 lags: Bias of estimators} \\[0.5\baselineskip]
\includegraphics[width=0.85\linewidth]{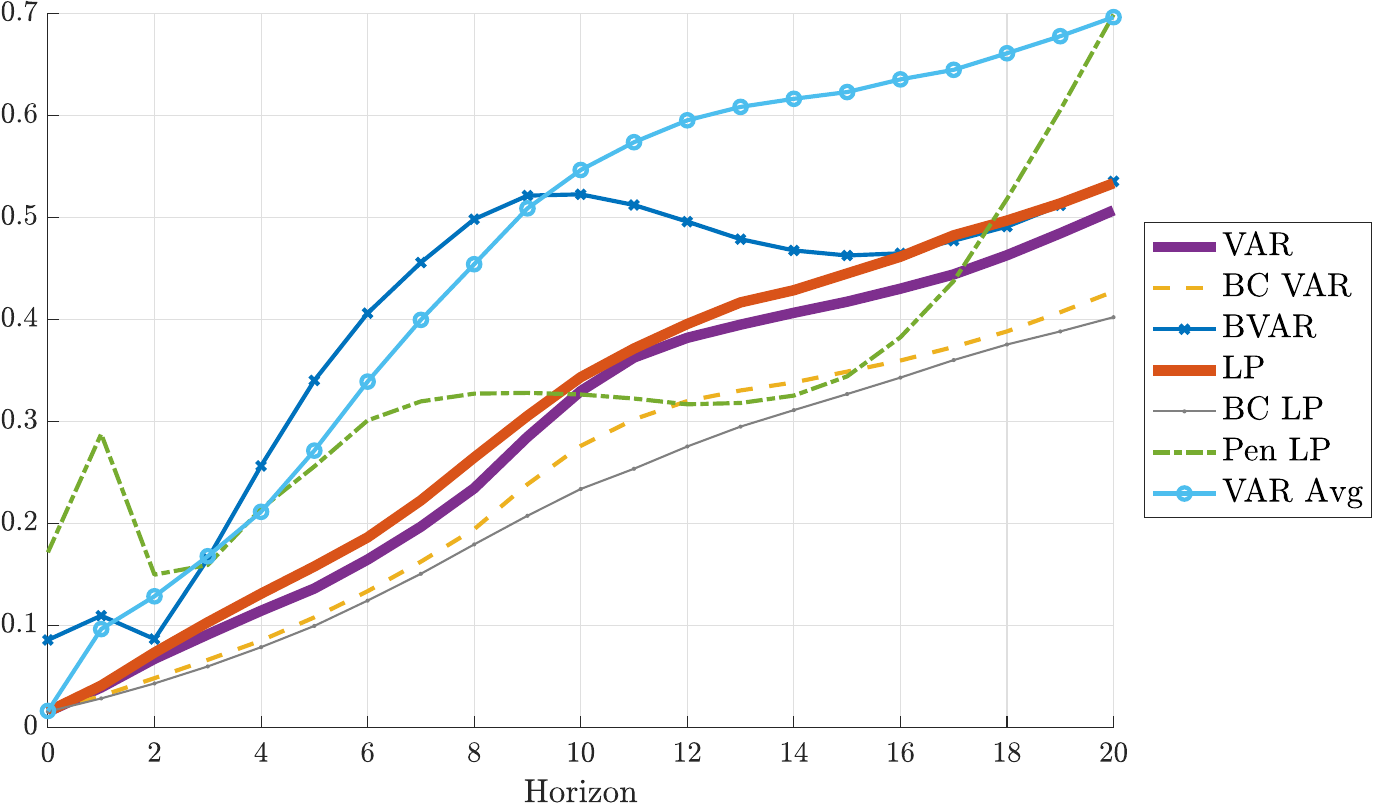}
\caption{Median (across DGPs) of absolute bias of the different estimation procedures, relative to $\sqrt{\frac{1}{21}\sum_{h=0}^{20}\theta_h^2}$.}
\label{fig:supp:bias_obsshock_lag8}

\vspace*{\floatsep}

\centering
\textsc{Observed shock, 8 lags: Standard deviation of estimators} \\[0.5\baselineskip]
\includegraphics[width=0.85\linewidth]{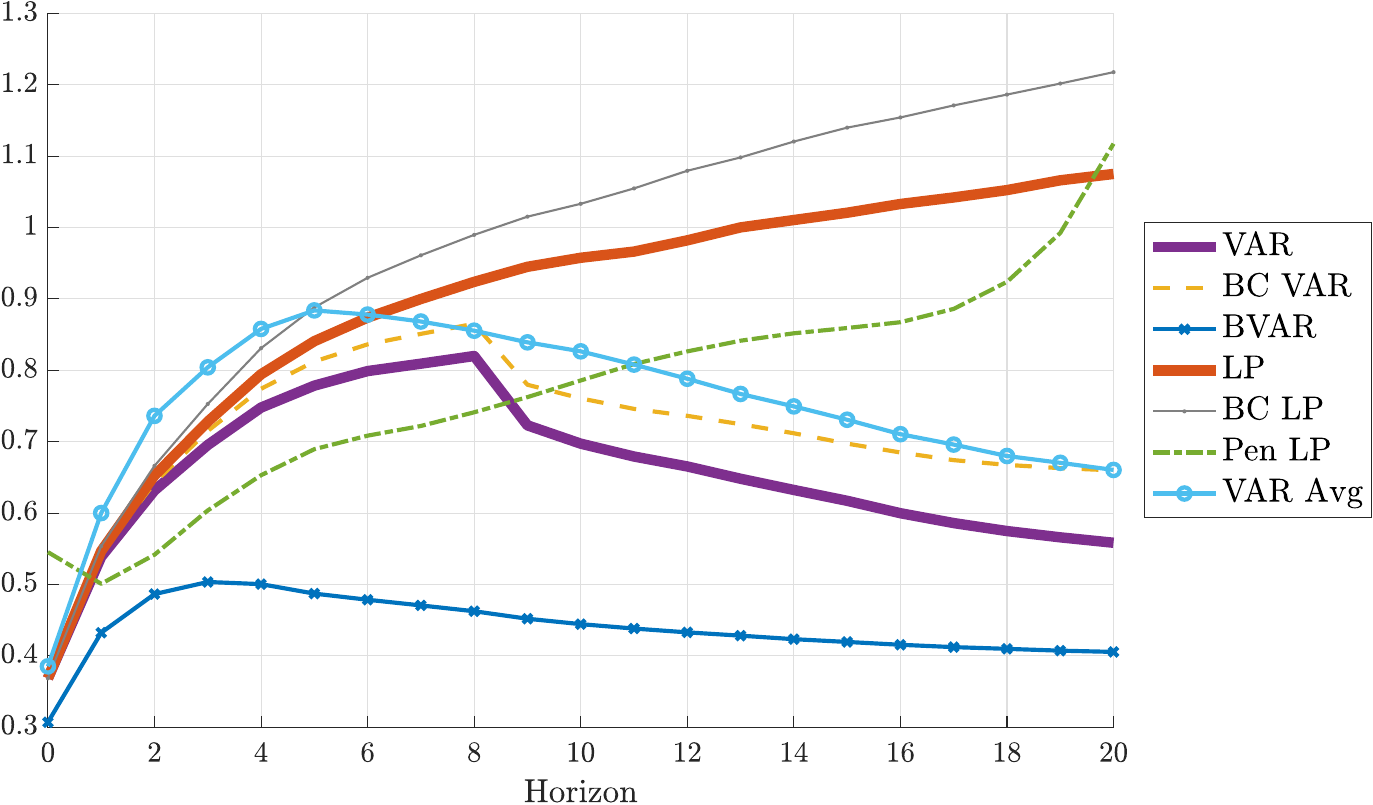}
\caption{Median (across DGPs) of standard deviation of the different estimation procedures, relative to $\sqrt{\frac{1}{21}\sum_{h=0}^{20}\theta_h^2}$.}
\label{fig:supp:std_obsshock_lag8}
\end{figure}

\begin{figure}[t]
\centering
\textsc{Observed shock, 8 lags: Optimal estimation method} \\
\includegraphics[width=\linewidth,clip=true,trim=0 0.5em 0 2.4em]{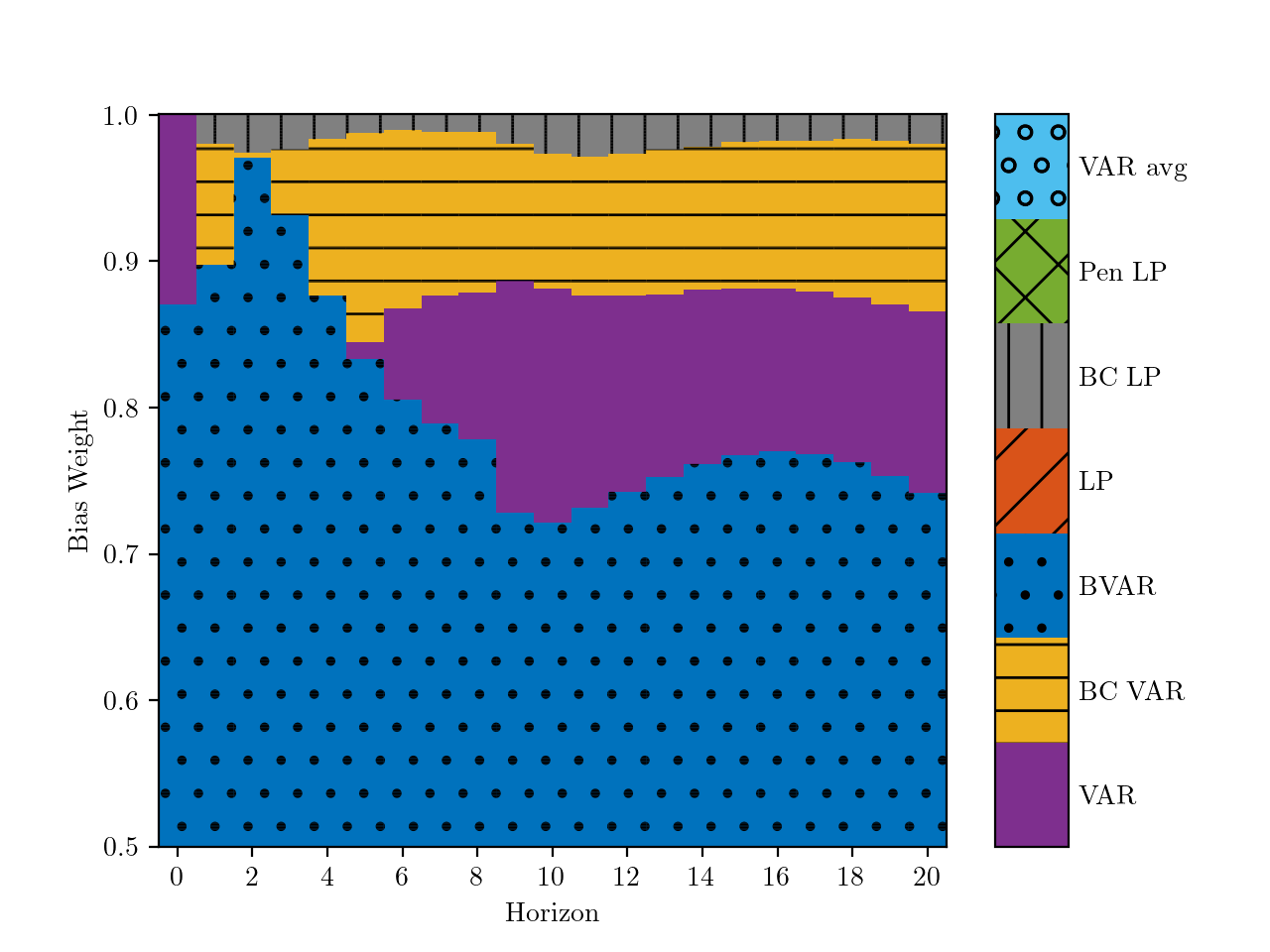}
\caption{Method that minimizes the average (across DGPs) loss function \eqref{eq:loss_simple}. Horizontal axis: impulse response horizon. Vertical axis: weight on squared bias in loss function. The loss function is normalized by the scale of the impulse response function, as in \cref{fig:bias_obsshock,fig:std_obsshock}. At $h=0$, VAR and LP are numerically identical; we break the tie in favor of VAR.}
\label{fig:supp:bestmethod_obsshock_lag8}
\end{figure}

\clearpage

\subsection{Smaller sample size}
\label{app:results_smallsample}

Recall that our baseline experiments consider a sample size of $T = 200$ quarters. We here instead present results for $T=100$. In the interest of space, we focus on results for ``observed shock'' identification.

Results for bias, standard deviation, and optimal method choice are displayed in \cref{fig:supp:bias_obsshock_small,fig:supp:std_obsshock_small,fig:supp:bestmethod_obsshock_small}, respectively. The figures are qualitatively similar to those for our baseline analysis, though there is a quantitative difference: reducing the sample size increases standard deviation by more (in relative terms) than bias. As a result, the estimation method choice plot indicates a more pronounced preference for shrinkage, with a larger area now solid-dotted blue (BVAR). Use of bias-corrected LP (grey with vertical lines) requires an even larger preference for low bias over high precision than in our baseline analysis.

\begin{figure}[tp]
\centering
\textsc{Observed shock, small sample: Bias of estimators} \\[0.5\baselineskip]
\includegraphics[width=0.85\linewidth]{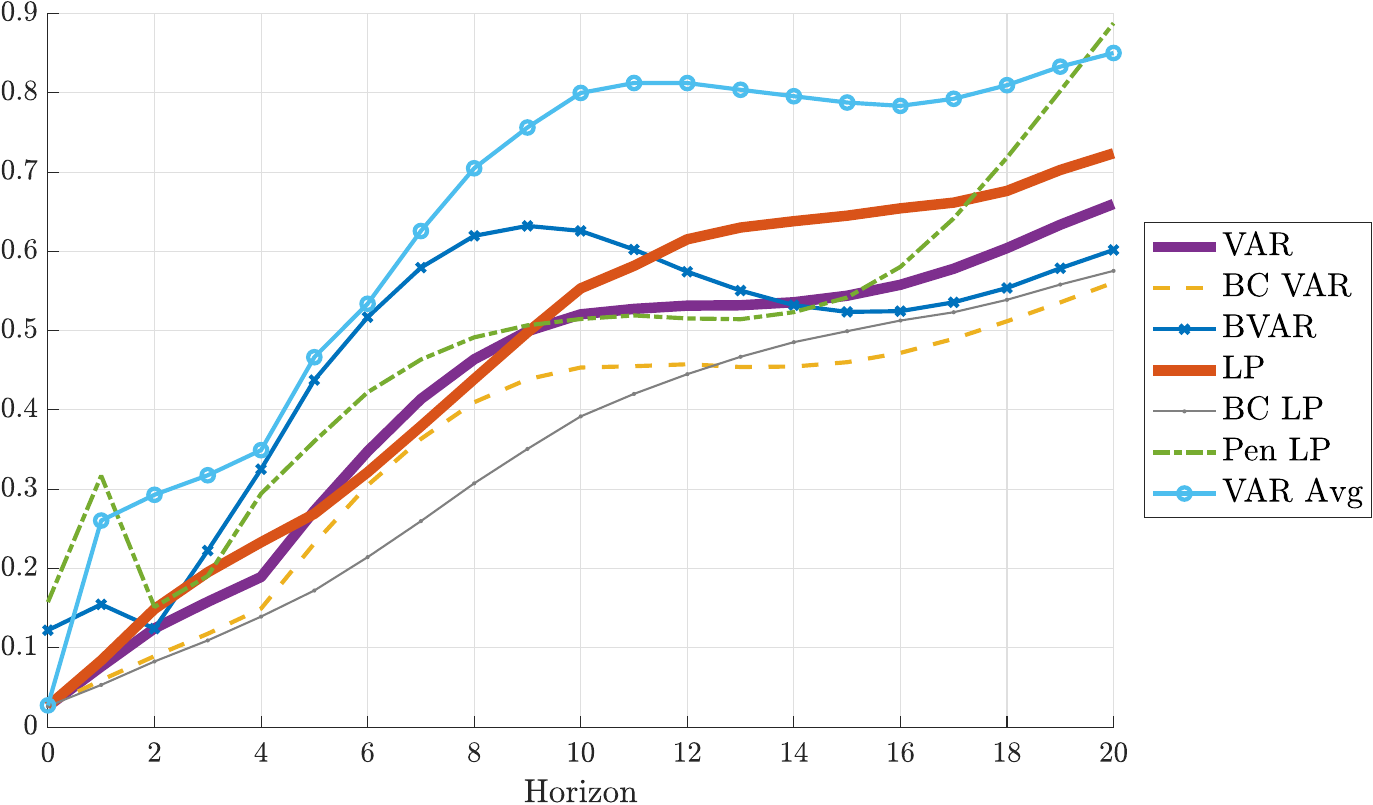}
\caption{Median (across DGPs) of absolute bias of the different estimation procedures, relative to $\sqrt{\frac{1}{21}\sum_{h=0}^{20}\theta_h^2}$.}
\label{fig:supp:bias_obsshock_small}

\vspace*{\floatsep}

\centering
\textsc{Observed shock, small sample: Standard deviation of estimators} \\[0.5\baselineskip]
\includegraphics[width=0.85\linewidth]{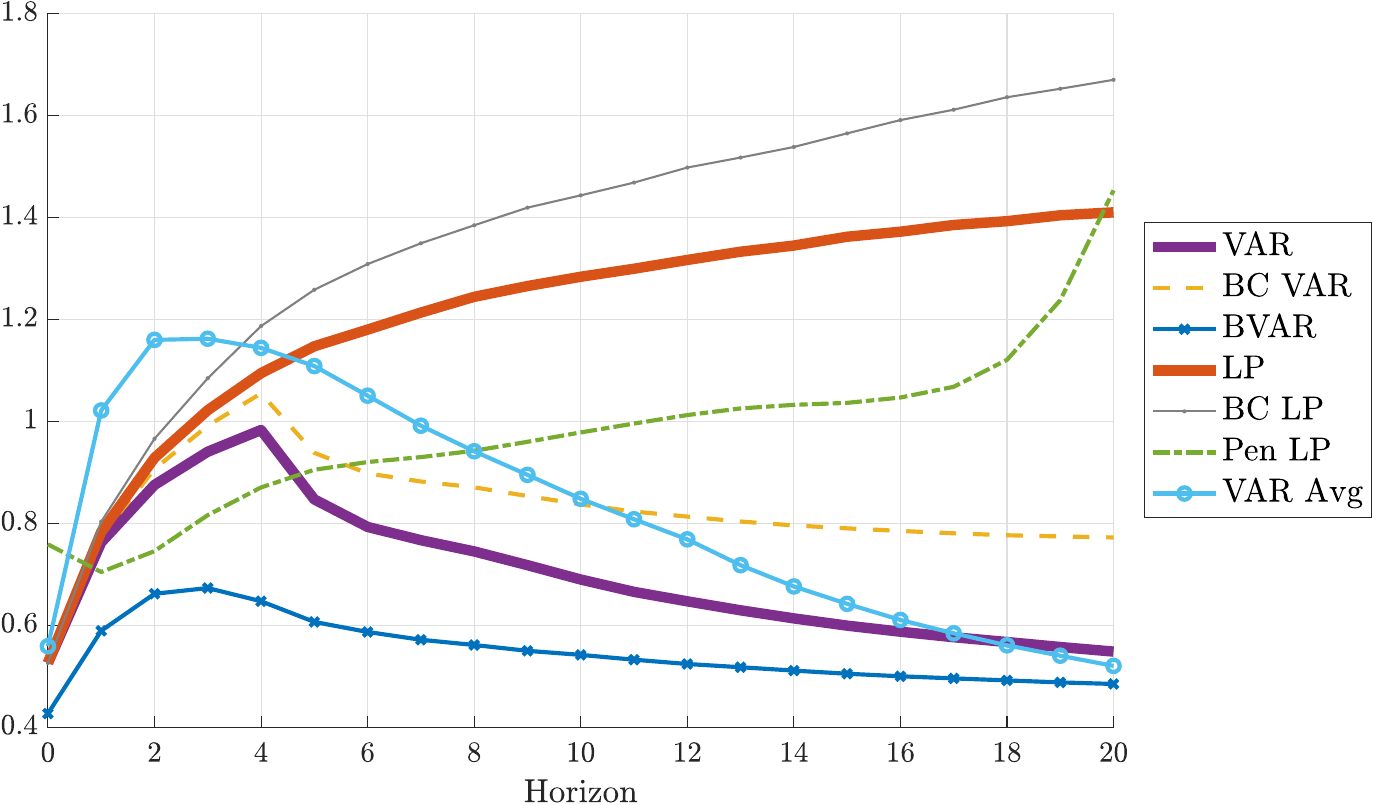}
\caption{Median (across DGPs) of standard deviation of the different estimation procedures, relative to $\sqrt{\frac{1}{21}\sum_{h=0}^{20}\theta_h^2}$.}
\label{fig:supp:std_obsshock_small}
\end{figure}

\begin{figure}[t]
\centering
\textsc{Observed shock, small sample: Optimal estimation method} \\
\includegraphics[width=\linewidth,clip=true,trim=0 0.5em 0 2.4em]{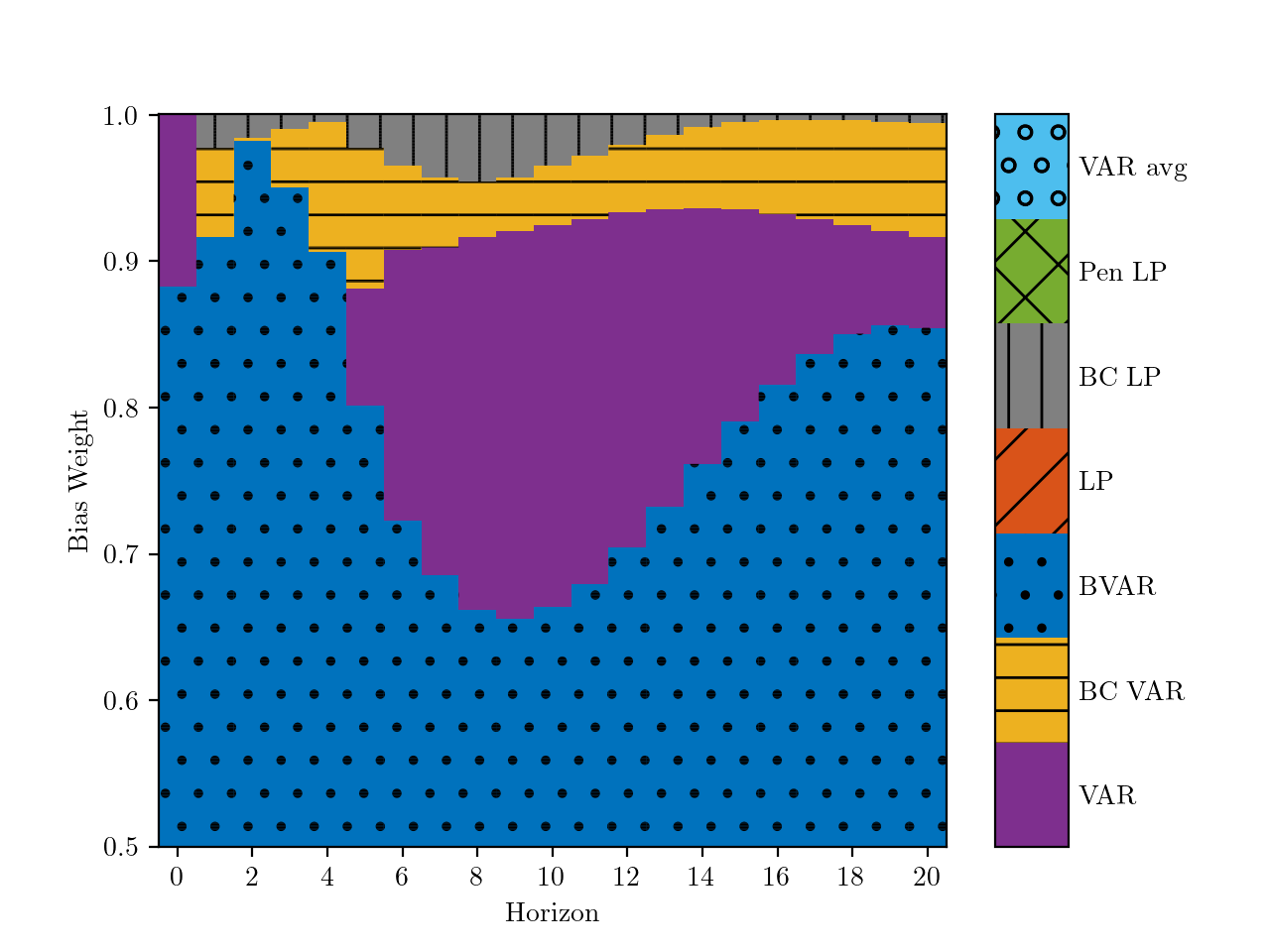}
\caption{Method that minimizes the average (across DGPs) loss function \eqref{eq:loss_simple}. Horizontal axis: impulse response horizon. Vertical axis: weight on squared bias in loss function. The loss function is normalized by the scale of the impulse response function, as in \cref{fig:bias_obsshock,fig:std_obsshock}. At $h=0$, VAR and LP are numerically identical; we break the tie in favor of VAR.}
\label{fig:supp:bestmethod_obsshock_small}
\end{figure}

\clearpage

\subsection{Larger sample size and estimation lag length}
\label{app:results_large}
\alert{While our baseline experiments have a sample size of $T=200$ and an estimation lag length of $p=4$, here we increase these to $T=720$ and $p=12$. This configuration of $(T,p)$ is reminiscent of sample sizes and specifications often seen with monthly data. However, since the parameters of the encompassing DFM remain fixed at our quarterly calibration described in \cref{sec:dgp}, we caution that it would be inappropriate to draw firm conclusions about how the estimators will perform in actual monthly data sets. In this subsection only, we use 2,000 Monte Carlo simulations per DGP, rather than 5,000.}

\alert{\cref{fig:supp:bias_obsshock_large,fig:supp:std_obsshock_large} show the bias and standard deviation of the estimators. As predicted by theory \citep{Plagborg2020}, the least-squares LP and VAR estimators give similar results for horizons $h \leq p=12$. At horizons $h>12$, the trade-off between the various estimators is qualitatively similar to our baseline results in \cref{sec:results}. \cref{fig:supp:bestmethod_obsshock_large} shows that, while the choice of optimal estimation method is overall similar to our baseline, penalized LP is more attractive at horizons $h \in [5,12]$ and for loss functions with $\omega > 0.5$ than in our baseline results, since this estimator offers useful shrinkage relative to least-squares LP or VAR. Nevertheless, BVAR remains the preferred estimator at almost all horizons for the special case of MSE loss ($\omega=1/2$).}

\begin{figure}[tp]
	\centering
	\textsc{Observed shock, large sample, 12 lags: Bias of estimators} \\[0.5\baselineskip]
	\includegraphics[width=0.85\linewidth]{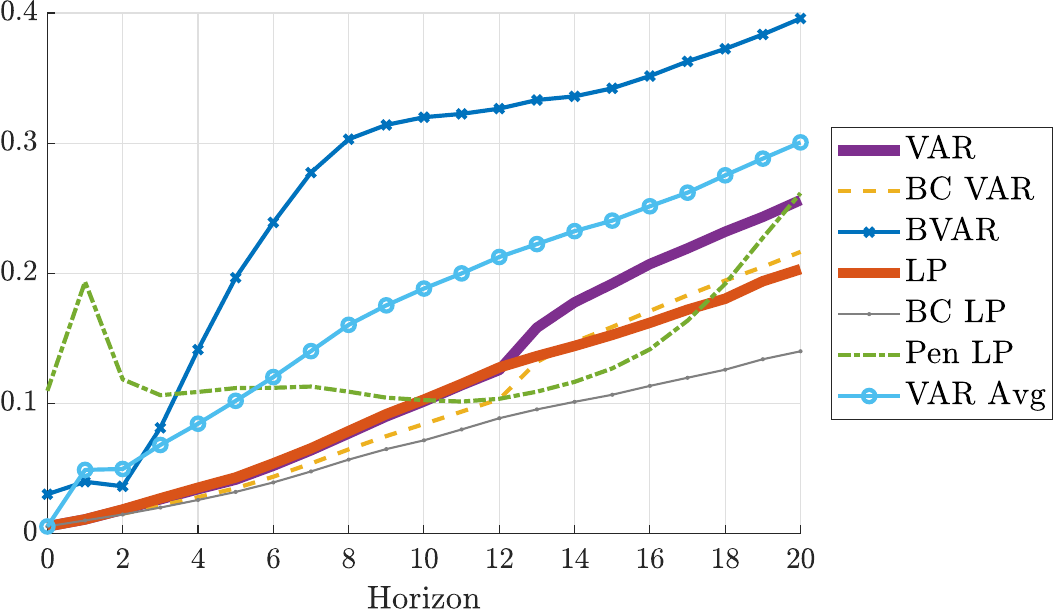}
	\caption{Median (across DGPs) of absolute bias of the different estimation procedures, relative to $\sqrt{\frac{1}{21}\sum_{h=0}^{20}\theta_h^2}$.}
	\label{fig:supp:bias_obsshock_large}
	
	\vspace*{\floatsep}
	
	\centering
	\textsc{Observed shock, large sample, 12 lags: Standard deviation of estimators} \\[0.5\baselineskip]
	\includegraphics[width=0.85\linewidth]{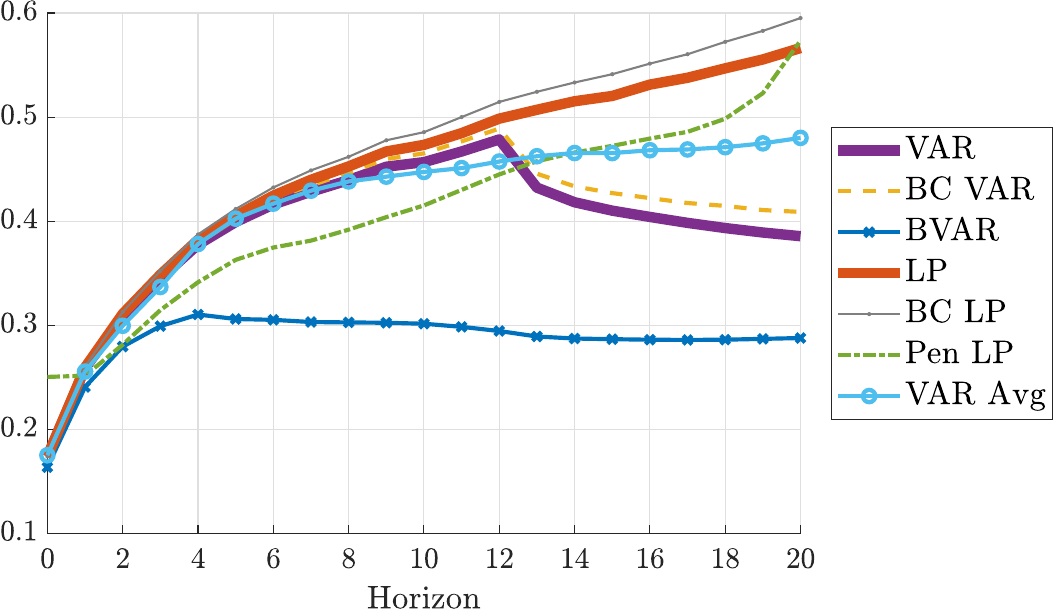}
	\caption{Median (across DGPs) of standard deviation of the different estimation procedures, relative to $\sqrt{\frac{1}{21}\sum_{h=0}^{20}\theta_h^2}$.}
	\label{fig:supp:std_obsshock_large}
\end{figure}

\begin{figure}[t]
	\centering
	\textsc{Observed shock, large sample, 12 lags: Optimal estimation method} \\
	\includegraphics[width=\linewidth,clip=true,trim=0 0.5em 0 2.4em]{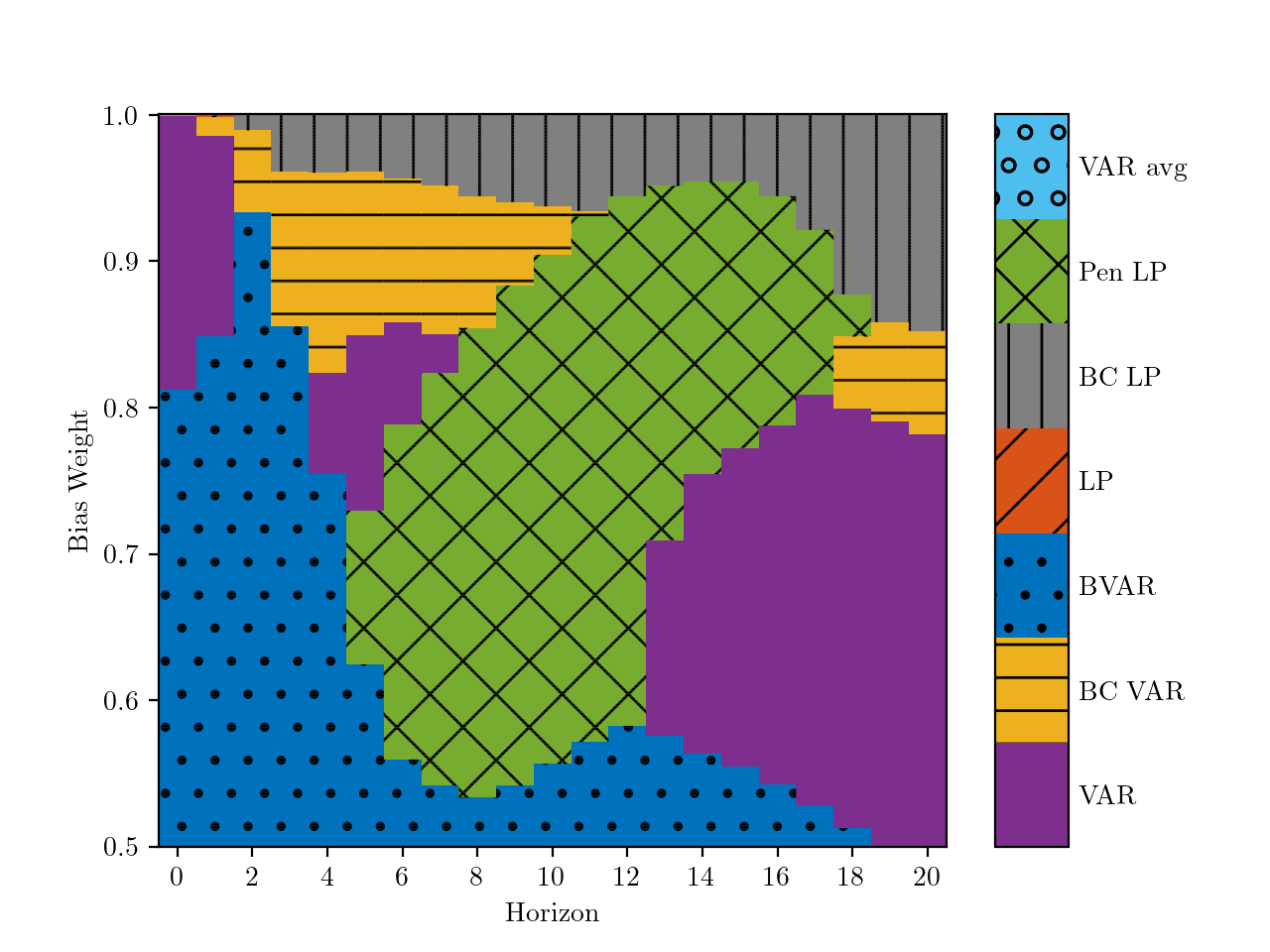}
	\caption{Method that minimizes the average (across DGPs) loss function \eqref{eq:loss_simple}. Horizontal axis: impulse response horizon. Vertical axis: weight on squared bias in loss function. The loss function is normalized by the scale of the impulse response function, as in \cref{fig:bias_obsshock,fig:std_obsshock}. At $h=0$, VAR and LP are numerically identical; we break the tie in favor of VAR.}
	\label{fig:supp:bestmethod_obsshock_large}
\end{figure}

\clearpage

\subsection{More observables}
\label{app:results_more}
\alert{Our baseline simulations assume that the econometrician observes $n_{\bar{w}}=5$ macro variables (perhaps in addition to an observed shock or IV). In this subsection we consider simulations that increase the number of macro observables from 5 to 7.}

\alert{We first consider the case of observed shock identification. \cref{fig:supp:bias_obsshock_more,fig:supp:std_obsshock_more} show that the bias and standard deviation properties of the estimators change little from our baseline results when we increase the number of observables. All qualitative conclusions emphasized in \cref{sec:results_1,sec:results_2,sec:results_3} continue to go through. Similarly, \cref{fig:supp:bestmethod_obsshock_more} shows that the choice of optimal method as a function of horizon and loss function is virtually unchanged from the baseline, with the minor exception that the region where penalized LP is preferred has shrunk somewhat.}

\begin{figure}[tp]
	\centering
	\textsc{Observed shock, 7 observables: Bias of estimators} \\[0.5\baselineskip]
	\includegraphics[width=0.85\linewidth]{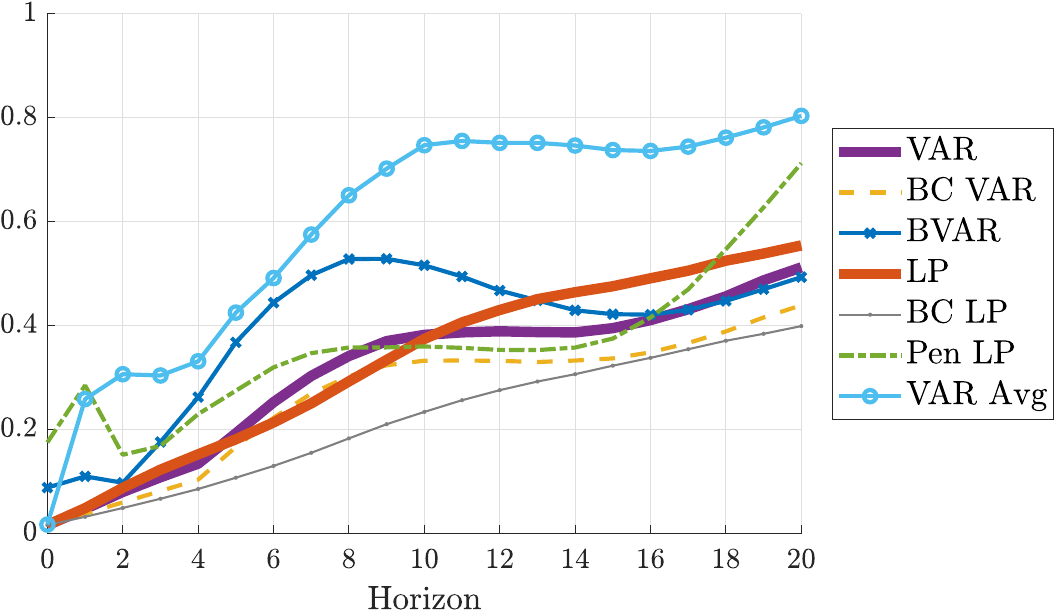}
	\caption{Median (across DGPs) of absolute bias of the different estimation procedures, relative to $\sqrt{\frac{1}{21}\sum_{h=0}^{20}\theta_h^2}$.}
	\label{fig:supp:bias_obsshock_more}
	
	\vspace*{\floatsep}
	
	\centering
	\textsc{Observed shock, 7 observables: Standard deviation of estimators} \\[0.5\baselineskip]
	\includegraphics[width=0.85\linewidth]{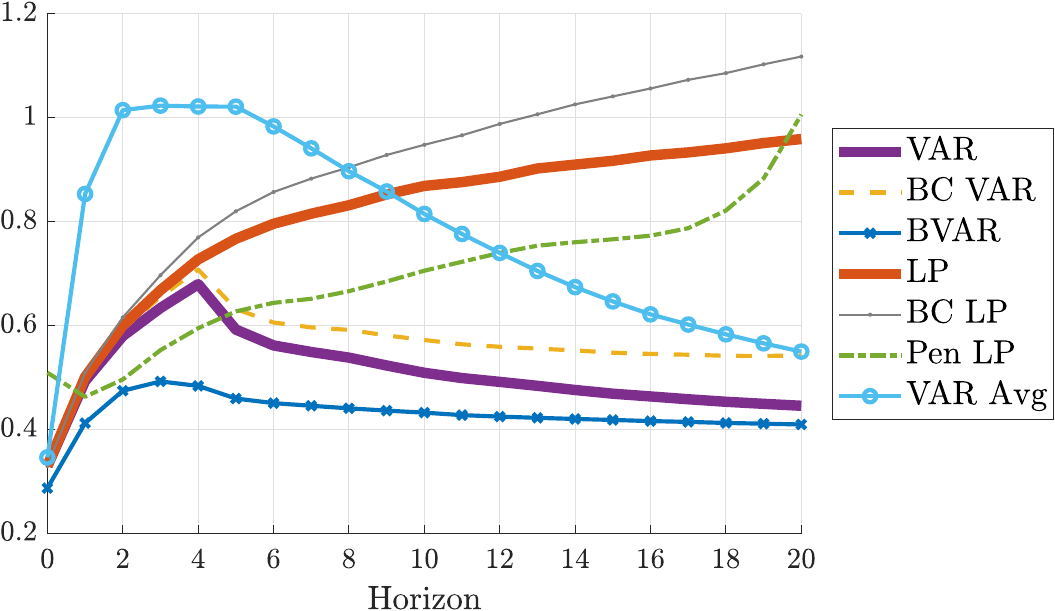}
	\caption{Median (across DGPs) of standard deviation of the different estimation procedures, relative to $\sqrt{\frac{1}{21}\sum_{h=0}^{20}\theta_h^2}$.}
	\label{fig:supp:std_obsshock_more}
\end{figure}

\begin{figure}[t]
	\centering
	\textsc{Observed shock, 7 observables: Optimal estimation method} \\
	\includegraphics[width=\linewidth,clip=true,trim=0 0.5em 0 2.4em]{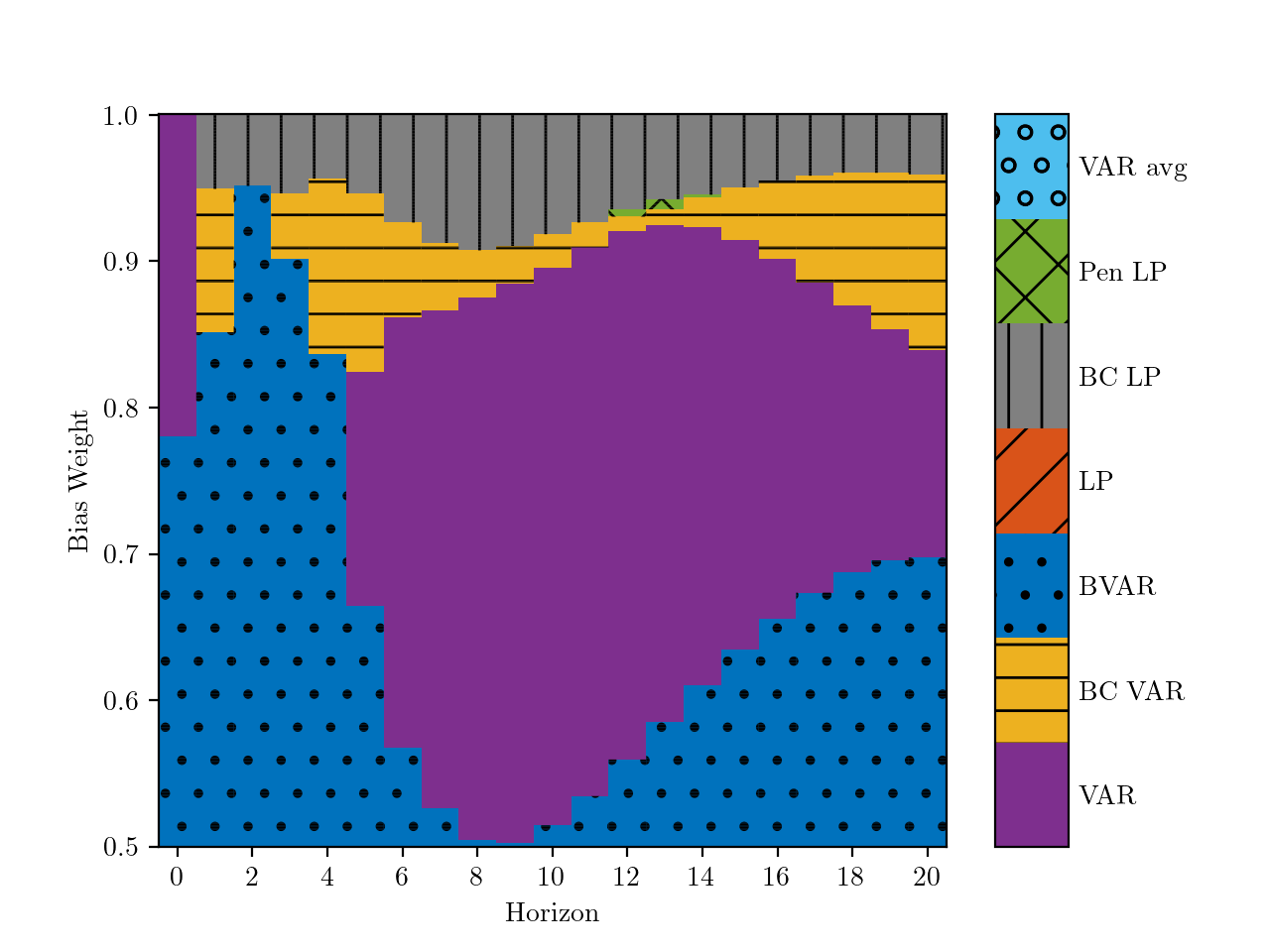}
	\caption{Method that minimizes the average (across DGPs) loss function \eqref{eq:loss_simple}. Horizontal axis: impulse response horizon. Vertical axis: weight on squared bias in loss function. The loss function is normalized by the scale of the impulse response function, as in \cref{fig:bias_obsshock,fig:std_obsshock}. At $h=0$, VAR and LP are numerically identical; we break the tie in favor of VAR.}
	\label{fig:supp:bestmethod_obsshock_more}
\end{figure}

\alert{In the case of IV identification, the median bias and interquartile range for the various estimators displayed in \cref{fig:supp:bias_iv_more,fig:supp:std_iv_more} are largely similar to our baseline results in \cref{sec:results_4}. However, the median bias of SVAR-IV relative to internal-instruments procedures is slightly smaller when the number of observables is larger, due to the mechanical increase in the degree of invertibility (the median degree of invertibility across the DGPs equals 0.45, versus 0.39 in our baseline, cf. \cref{tab:dgp_summ}).}

\begin{figure}[tp]
	\centering
	\textsc{IV, 7 observables: Median bias of estimators} \\[0.5\baselineskip]
	\includegraphics[width=0.85\linewidth]{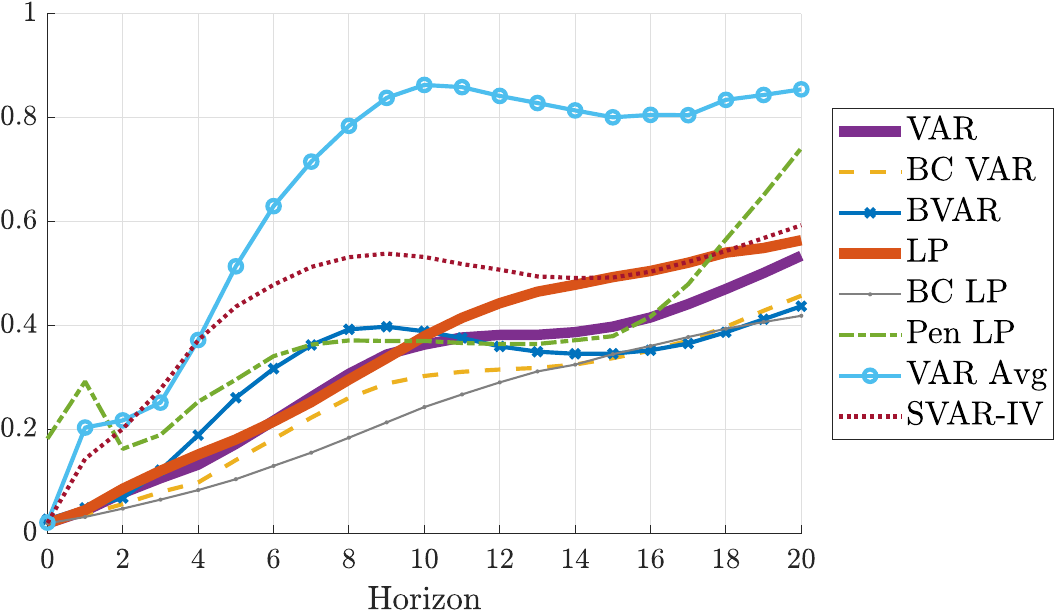}
	\caption{Median (across DGPs) of absolute median bias of the different estimation procedures, relative to $\sqrt{\frac{1}{21}\sum_{h=0}^{20}\theta_h^2}$.}
	\label{fig:supp:bias_iv_more}
	
	\vspace*{\floatsep}
	
	\centering
	\textsc{IV, 7 observables: Interquartile range of estimators} \\[0.5\baselineskip]
	\includegraphics[width=0.85\linewidth]{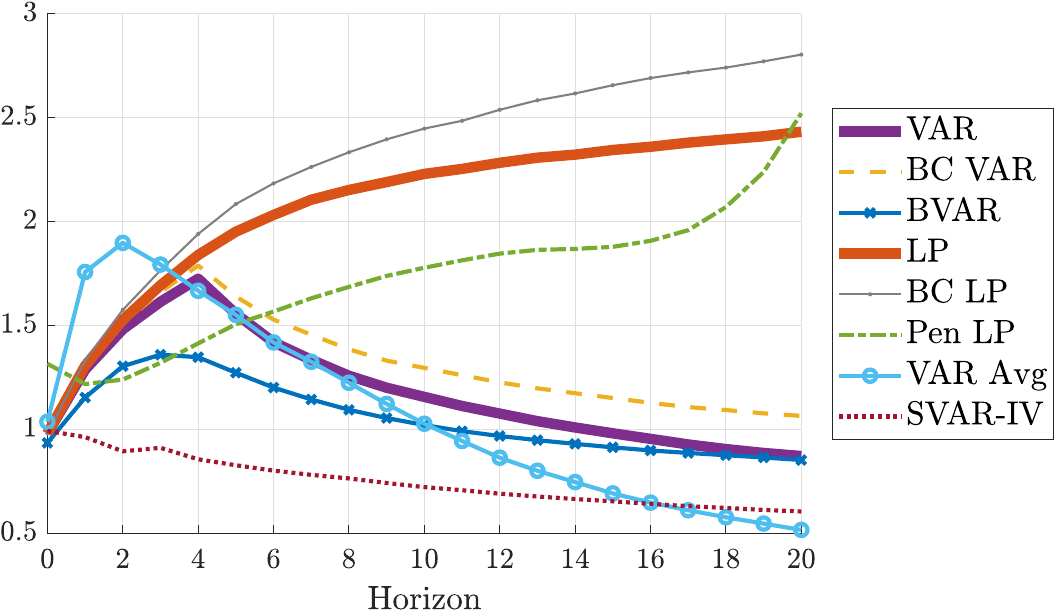}
	\caption{Median (across DGPs) of interquartile range of the different estimation procedures, relative to $\sqrt{\frac{1}{21}\sum_{h=0}^{20}\theta_h^2}$.}
	\label{fig:supp:std_iv_more}
\end{figure}

\clearpage

\subsection{Splitting by variable categories}
\label{app:results_cat}
Here we show that the categories of variables included in the DGP are not highly predictive of the bias or standard deviation of any estimator. We group the 207 variables in the \citet{Stock2016} data set into 11 categories, proceeding exactly as in the Data Appendix to that paper, except that we lump together their final three categories in a single ``Asset Price \& Sentiment'' category. Let $\Delta_{j,d}$ denote the total number of variables in category $j$ that are included in DGP $d$. We are interested in relating the absolute bias and standard deviation of the various estimation methods with $\Delta_{j,d}$, across DGPs and horizons.

Using the baseline simulation results from \cref{sec:results_1,sec:results_2,sec:results_3}, we run the following OLS regression separately for each estimation method $m$, pooling observations across all DGPs and all horizons:
\[\log \text{bias}_{m,d,h} = \sum_{j=1}^{10} \hat{\gamma}_j \Delta_{j,d}  + \sum_{i=0}^{20}\hat{\iota}_i \mathbbm{1}(i=h) + \hat{e}_{m,d,h},\quad d=1,\dots,6000,\; h=0,1,\dots,20.\]
Here $\text{bias}_{m,d,h}$ is the absolute bias of estimator $m$ in DGP $d$ at horizon $h$, $\hat{\gamma}_j$ are coefficients on the category counts, $\hat{\iota}_i$ are coefficients on indicator variables for each horizon, and $\hat{e}_{m,d,h}$ is the OLS residual. Notice that we omit category 11 from the regression for reasons of multicollinearity (there are 5 variables in each DGP). We also run the above regression with the log standard deviation on the left-hand side in place of the log absolute bias.

\begin{table}[tp]
\centering
\textsc{Bias regression} \\[0.5\baselineskip]
\renewcommand{\arraystretch}{1.2}
\begin{tabular}{l|ccccccc}
	&	VAR	&	VAR BC	&	BVAR	&	LP	&	LP BC	&	Pen LP	&	VAR Avg	\\
\hline
NIPA	& -0.12 & -0.10 & -0.03 & -0.05 & -0.06 & 0.01 & -0.04 \\
Industrial Production	&	-0.06 & -0.06 & 0.66 & -0.01 & -0.01 & 0.06 & 0.02 \\
Employment \& Unemployment	& 0.02 & 0.07 & 0.08 & -0.02 & 0.05 & 0.04 & -0.03 \\
Orders, Inventories \& Sales &	-0.05 & -0.09 & 	0.01 & -0.03 & 0.00 & 0.02 & 0.00 \\
Housing Starts \& Permits	& -0.01 & 0.01 & 0.04 & 0.02 & 0.08 & 0.08 & 0.06 \\
Prices & 0.05 & 0.15 & 0.12 & -0.18 & -0.03 & -0.17 & -0.05 \\
Productivity \& Earnings & 0.02 & -0.01 & 0.03 & -0.02 & -0.05 & -0.02 & -0.01 \\
Interest Rates & -0.01 & -0.01 & 0.03 & -0.03 & -0.03 & -0.04 & -0.03 \\
Money \& Credit & 0.04 & 0.07 & 0.08 & -0.03 & 0.01 & 0.00 & -0.03 \\
International & 0.07 & 0.12 & 0.15 & -0.11 & -0.05 & -0.09 & -0.05 \\
\end{tabular} \\
\vspace{\baselineskip}
\textsc{Standard deviation regression} \\[0.5\baselineskip]
\renewcommand{\arraystretch}{1.2}
\begin{tabular}{l|ccccccc}
	&	VAR	&	VAR BC	&	BVAR	&	LP	&	LP BC	&	Pen LP	&	VAR Avg	\\
\hline
NIPA	& 0.08 & 0.09 & 0.08 & 0.11 & 0.11 & 0.11 & 0.10 \\
Industrial Production	& -0.02 & -0.02 & -0.02 & -0.04 & -0.04 & -0.03 & -0.03 \\
Employment \& Unemployment	& 0.08 & 0.07 & 0.08 & 0.07 & 0.08 & 0.08 & 0.07 \\
Orders, Inventories \& Sales &	-0.04 & -0.04 & -0.05 & -0.04 & -0.04 & -0.03 & -0.04 \\
Housing Starts \& Permits	& 0.05 & 0.05 & 0.06 & 0.06 & 0.06 & 0.07 & 0.06 \\
Prices & 0.04 & 0.04 & 0.03 & 0.11 & 0.09 & 0.06 & 0.08 \\
Productivity \& Earnings & -0.03 & -0.03 & -0.03 & -0.04 & -0.04 & -0.04 & -0.04 \\
Interest Rates & -0.13 & -0.13 & -0.13 & -0.17 & -0.17 & -0.17 & -0.17 \\
Money \& Credit & 0.03 & 0.03 & 0.03 & 0.02 & 0.02 & 0.03 & 0.02 \\
International & 0.03 & 0.03 & 0.02 & 0.07 & 0.06 & 0.05 & 0.06 \\
\end{tabular}
\caption{Coefficients from OLS regressions of log bias (top table) or log standard deviation (bottom table) on variable category counts (along rows), controlling for horizon fixed effects. Regressions are run separately by estimation method (along columns), and observations are pooled across DGPs and horizons. Bias and standard deviation normalized as in \cref{fig:bias_obsshock,fig:std_obsshock}. Observed shock identification.} \label{tab:results_cat}
\end{table}

\cref{tab:results_cat} shows the coefficients $\hat{\gamma}_j$ on the category counts. A coefficient of 0.1, say, means that adding one variable of that category to the DGP (while removing a variable from the omitted ``Asset Price \& Sentiment'' category) leads to a 10\% higher absolute bias (resp., standard deviation). No coefficients in the table exceed 0.2 in absolute value, indicating that none of the categories are particularly predictive of the bias or standard deviation of the resulting impulse response estimates. The only exception is that adding variables from the ``Interest Rates'' category does seem to moderately lower the standard deviation of most of our estimation procedures, which is consistent with the results reported in \cref{app:results_fiscal_monetary}. Though the regression results reported in \cref{tab:results_cat} pool across all horizons $h \in [0,20]$, we obtain similar coefficients if we instead restrict the regressions only to the intermediate horizons $h \in [5,12]$.

We conclude that our baseline results do not conceal significant across-category heterogeneity. In particular, and differently from \citet{Marcellino2006}, we do not find major differences when comparing price and real activity series.

\clearpage

\section{Proofs}
\label{app:proofs}

\subsection{Auxiliary lemmas}
\label{app:proofs_aux}
Before proving \cref{prop:biasvar_simple}, we state and prove some auxiliary lemmas. All lemmas below impose the assumptions of \cref{prop:biasvar_simple}.

Define the process $\tilde{y}_t \equiv \sum_{s=1}^t( \varepsilon_{1,s}+\tau \varepsilon_{1,s-1}+\varepsilon_{2,s})$ for all $t \geq 1$, with $\tilde{y}_0=0$, and let $\tilde{w}_t \equiv (\varepsilon_{1,t},\tilde{y}_t)'$. Define also the $2 \times 2$ matrix $D_T \equiv \diag(T^{1/2},T)$.

\begin{lem} \label{lem:negligible}
	For all $j=1,2$ and $\ell \geq 0$,
	\[\frac{1}{T}\sum_{t=1}^T (y_t-\tilde{y}_t)^2 = O_p(1),\quad \frac{1}{T^{1/2}}\sum_{t=1}^T (y_t-\tilde{y}_t)\varepsilon_{j,t+\ell} = O_p(1).\]
\end{lem}
\begin{proof}
	From the DGP \eqref{eq:DGP_simple} we have
	\[y_t-\tilde{y}_t = \frac{\alpha}{\sqrt{T}}\sum_{s=1}^t\varepsilon_{1,s-2}.\]
	The first statement of the lemma follows from Markov's inequality and a simple moment calculation. Then second statement of the lemma follows from Chebyshev's inequality, since the process $\varepsilon_{j,t+\ell}\sum_{s=1}^t\varepsilon_{1,s-2}$ is serially uncorrelated and has a variance of order $O(T)$.
\end{proof}
\begin{lem} \label{lem:stoch_order}
	For any $\ell \geq 0$ and $j=1,2$,
	\[\sum_{t=1}^T \varepsilon_{j,t+\ell}\tilde{w}_{t-1}'D_T^{-1} = O_p(1).\]
\end{lem}
\begin{proof}
	This follows from Chebyshev's inequality and standard variance calculations, using that $\varepsilon_{j,t+\ell}\tilde{w}_{t-1}'$ is a serially uncorrelated process.
\end{proof}

\begin{lem} \label{lem:stoch_order2}
	For any $\ell \geq 1$ and $j=1,2$,
	\[\frac{1}{T}\sum_{t=1}^T \varepsilon_{j,t-\ell}\tilde{y}_{t-1} = O_p(1).\]
\end{lem}
\begin{proof}
	Write $\frac{1}{T}\sum_{t=1}^T \varepsilon_{j,t-\ell}\tilde{y}_{t-1} = \sum_{b=0}^{\ell-1}\frac{1}{T} \sum_{t=1}^T \varepsilon_{j,t-\ell}\Delta \tilde{y}_{t-b-1} + \frac{1}{T}\sum_{t=1}^T \varepsilon_{j,t-\ell}\tilde{y}_{t-\ell-1}$, where we define $\tilde{y}_t=0$ for $t\leq 0$. The first term is clearly $O_p(1)$, being composed of sample averages of stationary variables. The second term is $O_p(1)$ by \cref{lem:stoch_order}.
\end{proof}

\begin{lem} \label{lem:second_mom}
	$D_T^{-1}\sum_{t=1}^T w_tw_t' D_T^{-1} \stackrel{d}{\to} \diag\left(1,\lbrace (1+\tau)^2 + \sigma_2^2\rbrace\int_0^1 W(r)^2\,dr \right)$, where $W(\cdot)$ is a standard Brownian motion.
\end{lem}
\begin{proof}
	By \cref{lem:negligible}, it suffices to show that $D_T^{-1}\sum_{t=1}^T \tilde{w}_t\tilde{w}_t' D_T^{-1}$ converges to the stated limiting distribution. The $(1,1)$ element converges in probability to 1 by the law of large numbers. The $(1,2)$ element is $O_p(T^{-1/2})$ by \cref{lem:stoch_order2}. Finally, the convergence in distribution of the $(2,2)$ element follows from \citet[Thm. 3.4]{Phillips1992}, the continuous mapping theorem, and the long-run variance of the process $(\varepsilon_{1,t}+\tau\varepsilon_{1,t-1}+\varepsilon_{2,t})$ being equal to $(1+\tau)^2+\sigma_2^2$.
\end{proof}

\noindent Recall the definitions of $\theta_{h,T}$, $\hat{\beta}_h$, $\hat{\delta}_h$, $\hat{A}$, $\hat{\kappa}$, and the unit vector $e_j$ from \cref{sec:biasvar}.

\begin{lem} \label{lem:lp}
	For any $0 \leq \ell \leq h$, define $b_{\ell,h} \equiv \mathbbm{1}(\ell > 0) + \tau \mathbbm{1}(0<\ell<h)$. Then
	\begin{equation} \label{eq:lem_lp}
		\hat{\beta}_h-\theta_{h,T} = \frac{1}{T}\sum_{t=2}^{T-h} \sum_{\ell=0}^h ( b_{\ell,h} \varepsilon_{1,t+\ell} + \varepsilon_{2,t+\ell} ) \varepsilon_{1,t} + o_p(T^{-1/2}).
	\end{equation}
\end{lem}
\begin{proof}
	Let $\hat{\varepsilon}_{1,t} \equiv \varepsilon_{1,t} - \hat{b}'w_{t-1}$ be the residual from an auxiliary regression of $\varepsilon_{1,t}$ on $w_{t-1}$. Using \cref{lem:negligible,lem:stoch_order,lem:second_mom}, it is straightforward to show that
	\[D_T\hat{b} = \left\lbrace D_T^{-1}\sum_{t=2}^{T-h}w_{t-1}w_{t-1}'D_T^{-1}  \right\rbrace^{-1} \left\lbrace D_T^{-1}\sum_{t=2}^{T-h} \tilde{w}_{t-1} \varepsilon_{1,t} + O_p(T^{-1/2}) \right\rbrace = O_p(1).\]
	By the Frisch-Waugh theorem and sample orthogonality of $\hat{\varepsilon}_{1,t}$ and $w_{t-1}$, we may write
	\begin{align}
		\hat{\beta}_h &= \theta_{h,T} + \frac{\frac{1}{T}\sum_{t=2}^{T-h} (y_{t+h}-\theta_{h,T} \hat{\varepsilon}_{1,t}) \hat{\varepsilon}_{1,t}}{\frac{1}{T}\sum_{t=2}^{T-h} \hat{\varepsilon}_{1,t}^2} \nonumber \\
		&= \theta_{h,T} + \frac{\frac{1}{T}\sum_{t=2}^{T-h} (y_{t+h}- \theta_{h,T} \varepsilon_{1,t}- \tau\varepsilon_{1,t-1}-  y_{t-1}) \hat{\varepsilon}_{1,t}}{\frac{1}{T}\sum_{t=2}^{T-h} \hat{\varepsilon}_{1,t}^2}. \label{eq:proof_lp_estim}
	\end{align}
	\cref{lem:negligible,lem:second_mom} and $D_T\hat{b}=O_p(1)$ yield $\frac{1}{T}\sum_{t=2}^{T-h} \hat{\varepsilon}_{1,t}^2 \stackrel{p}{\to} E(\varepsilon_{1,t}^2)=1$. We can therefore focus on the numerator in the fraction in \eqref{eq:proof_lp_estim}, which we decompose as
	\begin{equation} \label{eq:proof_lp_decomp}
		\frac{1}{T}\sum_{t=2}^{T-h} (y_{t+h}- \theta_{h,T} \varepsilon_{1,t}- \tau\varepsilon_{1,t-1}- y_{t-1}) \varepsilon_{1,t} + \frac{1}{T}\sum_{t=2}^{T-h} (y_{t+h}- \theta_{h,T} \varepsilon_{1,t}- \tau\varepsilon_{1,t-1}-y_{t-1}) (\hat{\varepsilon}_{1,t}-\varepsilon_{1,t}).
	\end{equation}
	We first show that the first term above equals the right-hand side of \eqref{eq:lem_lp}. Iteration on the DGP \eqref{eq:DGP_simple} implies
	\[y_{t+h}- \theta_{h,T} \varepsilon_{1,t}- \tau\varepsilon_{1,t-1}- y_{t-1} = \sum_{\ell=0}^h (b_{\ell,h}\varepsilon_{1,t+\ell}+\varepsilon_{2,t+\ell}) + \frac{\alpha}{\sqrt{T}}\sum_{\ell=0}^h \varepsilon_{1,t+\ell-2}\mathbbm{1}(\ell \neq 2).\]
	The desired conclusion then follows from
	\[\frac{1}{T}\sum_{t=2}^{T-h} \frac{\alpha}{\sqrt{T}}\sum_{\ell=0}^h \varepsilon_{1,t+\ell-2}\mathbbm{1}(\ell \neq 2)\varepsilon_{1,t} = O_p(T^{-1}),\]
	which can be easily verified with Chebyshev's inequality, using that $\varepsilon_{1,t+\ell}\varepsilon_{1,t}$ is a serially uncorrelated process for $\ell \neq 0$.
	
	We finish the proof by showing that the second term in \eqref{eq:proof_lp_decomp} is $O_p(T^{-1})$, i.e.,
	\[-\frac{1}{T}\sum_{t=2}^{T-h}\left\lbrace\sum_{\ell=0}^h (b_{\ell,h}\varepsilon_{1,t+\ell}+\varepsilon_{2,t+\ell}) + \frac{\alpha}{\sqrt{T}}\sum_{\ell=0}^h \varepsilon_{1,t+\ell-2}\mathbbm{1}(\ell \neq 2) \right\rbrace w_{t-1}'\hat{b} = O_p(T^{-1}).\]
	This follows from $D_T\hat{b} = O_p(1)$ and \cref{lem:negligible,lem:stoch_order,lem:stoch_order2}.
\end{proof}

\begin{lem} \label{lem:var}
	Define $A_0 \equiv \begin{psmallmatrix} 0 & 0 \\ \tau & 1 \end{psmallmatrix}$. We have
	\begin{align*}
		\hat{A} - A_0 &= \frac{1}{T}\sum_{t=2}^T \begin{pmatrix}
			\varepsilon_{1,t} \\
			\varepsilon_{1,t} + \varepsilon_{2,t}
		\end{pmatrix} \varepsilon_{1,t-1}e_1' + o_p(T^{-1/2}), \\
		\hat{\kappa} - 1 &= \frac{1}{T}\sum_{t=2}^{T} \varepsilon_{1,t}\varepsilon_{2,t} + o_p(T^{-1/2}).
	\end{align*}
\end{lem}
\begin{proof}
	By appealing repeatedly to \cref{lem:negligible,lem:stoch_order,lem:second_mom,lem:stoch_order2}, we get
	\begin{align*}
		\hat{A}-A_0 &= \left(\frac{1}{T}\sum_{t=2}^T \begin{pmatrix}
			\varepsilon_{1,t} \\
			\varepsilon_{1,t} + \varepsilon_{2,t} + \frac{\alpha}{\sqrt{T}}\varepsilon_{1,t-2}
		\end{pmatrix} w_{t-1}'D_T^{-1}T^{1/2} \right)\left(D_T^{-1}\sum_{t=2}^T w_{t-1}w_{t-1}' D_T^{-1} \right)^{-1}D_T^{-1}T^{1/2} \\
		&= \Bigg(\frac{1}{T}\sum_{t=2}^T \begin{pmatrix}
			\varepsilon_{1,t} \\
			\varepsilon_{1,t} + \varepsilon_{2,t}
		\end{pmatrix} \tilde{w}_{t-1}'D_T^{-1}T^{1/2} + \frac{\alpha}{T}e_2\sum_{t=2}^T \varepsilon_{1,t-2}
		\tilde{w}_{t-1}'D_T^{-1} + O_p(T^{-1}) \Bigg) \lbrace e_1e_1' + o_p(1) \rbrace \\
		&= \Bigg(\frac{1}{T}\sum_{t=2}^T \begin{pmatrix}
			\varepsilon_{1,t} \\
			\varepsilon_{1,t} + \varepsilon_{2,t}
		\end{pmatrix} \tilde{w}_{t-1}'D_T^{-1}T^{1/2} + O_p(T^{-1}) \Bigg)\lbrace e_1e_1' + o_p(1) \rbrace \\
		&= \Bigg(\frac{1}{T}\sum_{t=2}^T \begin{pmatrix}
			\varepsilon_{1,t} \\
			\varepsilon_{1,t} + \varepsilon_{2,t}
		\end{pmatrix} (\varepsilon_{1,t-1}e_1'+T^{-1/2}\tilde{y}_{t-1}e_2') + O_p(T^{-1}) \Bigg)\lbrace e_1e_1' + o_p(1) \rbrace \\
		&= \frac{1}{T}\sum_{t=2}^T \begin{pmatrix}
			\varepsilon_{1,t} \\
			\varepsilon_{1,t} + \varepsilon_{2,t}
		\end{pmatrix} \varepsilon_{1,t-1}e_1'+\frac{1}{T^{3/2}}\sum_{t=2}^T \begin{pmatrix}
			\varepsilon_{1,t} \\
			\varepsilon_{1,t} + \varepsilon_{2,t}
		\end{pmatrix} \tilde{y}_{t-1}e_2' \times o_p(1) + o_p(T^{-1/2}) \\
		&= \frac{1}{T}\sum_{t=2}^T \begin{pmatrix}
			\varepsilon_{1,t} \\
			\varepsilon_{1,t} + \varepsilon_{2,t}
		\end{pmatrix} \varepsilon_{1,t-1}e_1' + O_p(T^{-1/2}) \times o_p(1) +  o_p(T^{-1/2}).
	\end{align*}
	This proves the first statement of the lemma.
	
	Next, by the Frisch-Waugh Theorem, $\hat{\kappa} \equiv \hat{\Sigma}_{21}/\hat{\Sigma}_{11}$ equals the coefficient on $\varepsilon_{1,t}$ in an OLS regression of $y_t$ on $\varepsilon_{1,t}$ and $w_{t-1}$. In other words, $\hat{\kappa}$ equals the impact LP estimate $\hat{\beta}_0$. The second statement of the lemma then follows from \cref{lem:lp} applied to $h=0$.
\end{proof}

\subsection{Proof of \texorpdfstring{\cref{prop:biasvar_simple}}{Proposition \ref{prop:biasvar_simple}}}

We derive the asymptotic distributions of the LP and VAR estimators in that order.

\paragraph{LP.} 
It follows from \cref{lem:lp} and a standard martingale central limit theorem that
\begin{equation*}
	\sqrt{T}(\hat{\beta}_h-\theta_{h,T}) \stackrel{d}{\to} N(0,\avar_\text{LP}),
\end{equation*}
where
\begin{align*}
	\avar_\text{LP} &= E(\varepsilon_{1,t}^2) E\left(\left\lbrace \sum_{\ell=0}^h (b_{\ell,h}\varepsilon_{1,t+\ell} + \varepsilon_{2,t+\ell}) \right\rbrace^2\right) \\
	&= \sum_{\ell=0}^h b_{\ell,h}^2 + (h+1)\sigma_2^2 \\
	&= \lbrace 1+(h-1)(1+\tau)^2)\rbrace \mathbbm{1}(h \geq 1) + (h+1)\sigma_2^2.
\end{align*}

\paragraph{VAR.}
We derive the asymptotic distribution of $\hat{\delta}_h$ by appealing to the delta method. Let $f_h(A,\kappa) \equiv e_2'A^h\gamma$, where $\gamma=(1,\kappa)'$, so that $\hat{\delta}_h = f_h(\hat{A},\hat{\kappa})$. We need the Jacobians of this transformation with respect to $\ve(A)$ and $\kappa$. In fact, we only require the Jacobians with respect to $\ve(Ae_1)$ and $\kappa$, since \cref{lem:var} implies that the second column of $\hat{A}$ is super-consistent. The Jacobians should be evaluated at $\plim \hat{A} = A_0 \equiv \begin{psmallmatrix} 0 & 0 \\ \tau & 1 \end{psmallmatrix}$ and $\plim \hat{\kappa}=1$. Thus, $\gamma$ should be evaluated at $\gamma_0 \equiv (1,1)'$.

First, for $h\geq 2$, the Jacobian with respect to $Ae_1$ equals \citep[p. 208]{Magnus2007}
\begin{align*}
	\frac{\partial e_2'A^h \gamma}{\partial \ve(Ae_1)
	} \Big|_{A=A_0,\gamma=\gamma_0} &= (\gamma' \otimes e_2') \sum_{j=1}^h \left( (A_0')^{h-j} \otimes A_0^{j-1}\right) (e_1 \otimes I) \\
	&= (\gamma' \otimes e_2') \sum_{j=1}^h \lbrace(A_0')^{h-j}e_1\rbrace \otimes A_0^{j-1} \\
	&= (\gamma' \otimes e_2') (e_1 \otimes A_0 + 0 )\\
	&= (\gamma'e_1 \otimes e_2'A_0) \\
	&= (\tau,1),
\end{align*}
where the third equality uses that $A_0$ is an idempotent matrix, and $A_0'e_1=0$. For $h=1$, the Jacobian with respect to $\ve(Ae_1)$ obviously equals $(0,1)$. So we can write the Jacobian for all $h \geq 1$ as $(\tau \mathbbm{1}(h \geq 2),1)$.

Second, for any $h \geq 1$, the Jacobian with respect to $\kappa$ equals
\[\frac{\partial e_2'A^h \gamma}{\partial \kappa
} \Big|_{A=A_0,\gamma=\gamma_0} = e_2'A_0^he_2 = e_2'A_0e_2 = 1.\]

Next, \cref{lem:var} and a standard martingale central limit theorem imply
\[\sqrt{T}(\hat{A}-A_0)e_1 \stackrel{d}{\to} N(0,\avar(\hat{A}e_1)),\quad \sqrt{T}(\hat{\kappa}-1) \stackrel{d}{\to} N(0,\avar(\hat{\kappa})),\]
where
\begin{align*}
	\avar(\hat{A}e_1) &= \var\left(\begin{pmatrix}
		\varepsilon_{1,t} \\
		\varepsilon_{1,t}+\varepsilon_{2,t}
	\end{pmatrix}\varepsilon_{1,t-1}\right) = \begin{pmatrix}
		1 & 1 \\
		1 & 1+\sigma_2^2
	\end{pmatrix}, \\
	\avar(\hat{\kappa}) &= \var(\varepsilon_{1,t}\varepsilon_{2,t}) = \sigma_2^2.
\end{align*}
Moreover, $\hat{A}e_1$ and $\hat{\kappa}$ are asymptotically independent by \cref{lem:var}, since
\[\cov\left(\begin{pmatrix}
	\varepsilon_{1,t} \\
	\varepsilon_{1,t}+\varepsilon_{2,t}
\end{pmatrix}\varepsilon_{1,t-1},\varepsilon_{1,t}\varepsilon_{2,t}\right) = 0.\]
Given all the preceding ingredients, we can apply the delta method to conclude that
\[\sqrt{T}(\hat{\delta}_h-\theta_{h,T}) = \sqrt{T}(\hat{\delta}_h-e_2'A^h\gamma) + \sqrt{T}(e_2'A^h\gamma-\theta_{h,T}) \stackrel{d}{\to} N(\abias_\text{VAR},\avar_\text{VAR}),\]
where, for $h \geq 1$,
\begin{align*}
	\pushQED{\qed}
	\abias_\text{VAR} &= \sqrt{T}(e_2'A^h\gamma-\theta_{h,T}) = -\alpha \mathbbm{1}(h \geq 2), \\[0.5\baselineskip]
	\avar_\text{VAR} &= (\tau \mathbbm{1}(h \geq 2),1)\begin{pmatrix}
		1 & 1 \\
		1 & 1+\sigma_2^2
	\end{pmatrix}\begin{pmatrix}
		\tau \mathbbm{1}(h \geq 2) \\
		1
	\end{pmatrix} + \sigma_2^2 \\
	&= (1+\tau \mathbbm{1}(h \geq 2))^2+2\sigma_2^2. \qedhere
	\popQED
\end{align*}

\clearpage
\phantomsection
\addcontentsline{toc}{section}{References}
\bibliography{lp_var_simul_ref}

\end{appendices}